\documentclass[onecolumn]{IEEEtran}
\usepackage[utf8]{inputenc} 
\usepackage[T1]{fontenc}
\usepackage{url}              % provides \url{...}
\usepackage[noadjust]{cite}             % improves presentation of citations
\usepackage[hidelinks]{hyperref}

\usepackage[cmex10]{amsmath}  % Use the [cmex10] option to ensure complicance
\usepackage{cleveref}

\usepackage{mleftright}       % fix to wrong spacing of \left-,
\mleftright                   % \middle- \right-commands 

\usepackage{graphicx}         % provides \includegraphics{...} to
\usepackage{booktabs}         % fixes poor spacing in tables and

\usepackage{amssymb,amsfonts,amsthm,bm,bbm,dsfont}
\newtheorem{theorem}{Theorem}
\newtheorem{remark}{Remark}
\newtheorem{claim}{Claim}
\newtheorem{proposition}{Proposition}
\newtheorem{example}{Example}
\newtheorem{lemma}{Lemma}
\newtheorem{definition}{Definition}

\newtheorem{corollary}{Corollary}

\crefname{property}{property}{properties}

\usepackage{soul}

\DeclareSymbolFont{bbold}{U}{bbold}{m}{n}
\DeclareSymbolFontAlphabet{\mathbbold}{bbold}

\usepackage{color}

\newcommand{\eps}{\varepsilon}
\newcommand{\bits}{\{0,1\}}

\newcommand{\cA}{\mathcal{A}}

\newcommand{\cC}{\mathcal{C}}

\newcommand{\cH}{\mathcal{H}}

\newcommand{\cJ}{\mathcal{J}}

\newcommand{\cS}{\mathcal{S}}

\newcommand{\VT}{\mathcal{VT}}
\newcommand{\Enc}{\mathrm{Enc}}
\newcommand{\Encstruct}{\mathrm{Enc}_{\mathrm{struct}}}
\newcommand{\Decstruct}{\mathrm{Dec}_{\mathrm{struct}}}
\newcommand{\Encsmall}{\mathrm{Enc}_{\mathrm{short}}}
\newcommand{\Decsmall}{\mathrm{Dec}_{\mathrm{short}}}

\newcommand{\poly}{\mathrm{poly}}

\newcommand{\boldq}{\mathbf{q}}

\newcommand{\bolds}{\mathbf{s}}

\newcommand{\boldu}{\mathbf{u}}

\newcommand{\boldw}{\mathbf{w}}
\newcommand{\boldx}{\mathbf{x}}
\newcommand{\boldy}{\mathbf{y}}

\DeclareMathOperator{\Var}{Var}

\usepackage{xcolor}
\usepackage{accents}
\usepackage{mathtools}
\usepackage[shortlabels]{enumitem}
\newcommand{\fD}{\mathfrak{D}}

\colorlet{ColorYPChen}{blue!80!black}

\newcommand*{\boldone}{\text{\usefont{U}{bbold}{m}{n}1}}
\newcommand{\tn}{\textnormal}
\newcommand{\polylog}{\mathrm{polylog}}

\begin{document}

\title{Correcting Contextual Deletions in DNA Nanopore Readouts}

\author{Yuan-Pon Chen, Olgica Milenkovic, João Ribeiro, and Jin Sima
\thanks{
The work of Y.-P.\ Chen, O.\ Milenkovic, and J.\ Sima was supported in part by the NSF Grant 2008125.
The work of J.\ Ribeiro was funded by the European Union (LESYNCH, 101218842) and by national funds through FCT – Fundação para a Ciência e a Tecnologia, I.P., and, when eligible, co-funded by EU funds under project/support UID/50008/2025 – Instituto de Telecomunicações, with DOI 10.54499/UID/50008/2025.
Views and opinions expressed are however those of the authors only and do not necessarily reflect those of the European Union or the European Research Council Executive Agency. Neither the European Union nor the granting authority can be held responsible for them.
}% <-this % stops a space
\thanks{Yuan-Pon Chen, Olgica Milenkovic, and Jin Sima are with the Department of Electrical and Computer Engineering, University of Illinois Urbana-Champaign, 61801 Urbana IL, USA (email: yuanpon2@illinois.edu; milenkov@illinois.edu; jsima@illinois.edu).
João Ribeiro is with the Department of Mathematics, Instituto Superior Técnico, Universidade de Lisboa, 1049-001 Lisboa, Portugal, and with Instituto de Telecomunicações, 1049-001, Lisboa, Portugal (email: jribeiro@tecnico.ulisboa.pt).}
}
\maketitle

\begin{abstract}
The problem of designing codes for deletion-correction and synchronization has received renewed interest due to applications in DNA-based data storage systems that use nanopore sequencers as readout  platforms. In almost all instances, deletions are assumed to be imposed independently of each other and of the sequence context. These assumptions are not valid in practice, since nanopore errors tend to occur within specific contexts. We study contextual nanopore deletion-errors through the example setting of deterministic single deletions following (complete) runlengths of length at least $k$. The model critically depends on the runlength threshold $k$, and we examine two regimes for $k$:  a) $k=C\log n$ for a constant $C\in(0,1)$; in this case, we study error-correcting codes that can protect from a constant number $t$ of contextual deletions, and show that the minimum redundancy (ignoring lower-order terms) is between $(1-C)t\log n$ and $2(1-C)t\log n$, meaning that it is a ($1-C$)-fraction of that of arbitrary $t$-deletion-correcting codes. To complement our non-constructive redundancy upper bound, we design efficiently and encodable and decodable codes for any constant $t$. In particular, for $t=1$ and $C>1/2$ we construct efficient codes with redundancy that essentially matches our non-constructive upper bound; b) $k$ equal a constant; in this case we consider the extremal problem where the number of deletions is not bounded and a deletion is imposed after every run of length at least $k$, which we call the \emph{extremal contextual deletion channel}. This combinatorial setting arises naturally by considering a probabilistic channel that introduces contextual deletions after each run of length at least $k$ with probability $p$ and taking the limit $p\to 1$. We obtain sharp bounds on the maximum achievable rate under the extremal contextual deletion channel for arbitrary constant $k$.
\end{abstract}


\section{Introduction}
\label{sec:introduction}

In recent years, there has been a surge of interest in the study of codes that can recover from symbol deletions. Unlike erasures, where the receiver knows the position of the missing symbols, deletions remove symbols without indicating their positions, causing a loss of synchronization between the sender and receiver. This misalignment of the symbols, on the receiver side, makes deletion correction a challenging problem. Despite remarkable progress, many fundamental questions regarding deletion correction remain unresolved. For example, we still do not known the minimal redundancy required to correct a constant number of worst-case deletions nor the exact capacity of the binary i.i.d. deletion channel. These and related questions continue to motivate diverse lines of research, outlined in several comprehensive overviews of the subject~\cite{sloane2002single,Mitzenmacher08survey,mercier2010survey,HS21,cheraghchi2020overview}.

Work on deletion correction has also been driven by existing and emerging practical application domains, the former including magnetic, optical and flash data storage, file synchronization, and multimedia data transmission. In the latter context, deletion-correcting codes also play a crucial role in DNA-based storage systems~\cite{yazdi2017portable,milenkovic2024dna}. DNA-based storage offers compelling advantages over classical storage media, including non-volatility, extremely high data density, and long-term stability. These properties make it a promising solution for archival storage at massive scales.
The idea of using DNA as a storage medium is not new~\cite{bancroft2001long} and several teams demonstrated read, write, random access, and safeguarding protocols~\cite{church2012next,goldman2013towards,grass2015robust,tabatabaei2015rewritable}. 
These works led to a large body of follow-up works in areas as diverse as synthetic biology, chemical engineering, coding theory, computational biology, etc (e.g., see~\cite{yazdi2017portable,heckel2019characterization,lopez2019dna,tabatabaei2020dna,antkowiak2020low,press2020hedges,lee2020photon,doricchi2022emerging,sima2023error,tabatabaei2022expanding,pan2022rewritable,milenkovic2024dna}, to list a few). Deletions as well as bursts of deletions, alongside insertions and substitutions, occur in DNA-based data storage systems that use nanopore sequencers as readout platforms during the data reconstruction phase.

The first experimental validation and theoretical study of nanopore sequencers as DNA-based data storage readout platforms was reported in~\cite{yazdi2017portable} and it revealed that nanopores mostly introduce synchronization errors in a \emph{contextual manner}. For the ONT (Oxford Nanopore Technologies) platforms available at the time of the study, deletion errors of certain bases in the DNA alphabet (such as $A$) were significantly more likely after sequence alignment. 
Furthermore, symbol deletions following longer runs (e.g., homopolymers) appeared at a significantly higher rate than those following shorter runs. 
This phenomenon can be attributed to the fact that finding the event boundaries in nanopore analog ion current signals is challenging, and the detection delay effect manifests itself by one or more ``absorbed'' (deleted) symbols following the runs. Hence, the length of the runs in stored data plays a crucial role as synchronization becomes more challenging as the runlengths increase.
Subsequent works have continued the study of
error statistics and correlations in various DNA-based data storage systems~\cite{heckel2019characterization,weindel2023embracing}. Furthermore, recent theoretical works have also described general classes of probabilistic channels with context-dependent synchronization errors~\cite{MDK18,con2025channels} that resemble those reported in~\cite{yazdi2017portable}. In contrast, this work focuses on a combinatorial setting for context-dependent synchronization errors.

To make the first contextual error-correction models more tractable for theoretical analysis, we simplify the assumptions to only include  \emph{symmetric deletion errors} (i.e., errors that do not discriminate among the symbols of the homopolymers) and \emph{single deletions}  following sufficiently long runlengths. We also consider different runlength threshold regimes and a bounded number of contextual deletion errors. All our results are presented for binary alphabets but can be extended to other alphabet sizes as well.

\subsection{The Model}
We start by introducing relevant notation and definitions. For simplicity, we focus on strings over binary rather than quaternary alphabets used in DNA encodings, since all approaches have natural extensions to larger alphabets. A \emph{deletion} is the operation in which a symbol is completely removed from a string, e.g., the deletion of the second and the fifth bit of $010001$ will give $0001$. 
A \emph{substring} of a string $s$ is a string obtained by taking consecutive symbols from $s$: for example, $s_{i}s_{i+1}\ldots s_{i+\ell-1}$ is a substring of  $s=s_1s_2\ldots s_n$ that has length length $\ell$ and which starts at position $i$ of $s$. Furthermore, a sequence of not necessarily consecutive symbols in a string $s$ is called a \emph{subsequence} of $s$, and it is obtained by deleting symbols from $s$.
A \emph{run} in a string $s$ is a single-symbol substring of $s$ such that the symbol before the run and the symbol after the run are different from the symbol of the run. For an example string $s=0111001$, we have four runs of respective lengths can write as the concatenation of alternating runs $0\circ111\circ00\circ 1$. Clearly, every binary string can be uniquely written as a concatenation
of runs of alternating symbols.

When sequencing fairly long runs of symbols using nanopores, a typical context-dependent error would be a ``deletion'' of the first symbol of the following run. This deletion arises due to the fact that in this case it is hard to detect a change in the ion current corresponding to a runlength change. Formal definitions of contextual deletions,
    contextual deletion channels,
and zero-error contextual deletion-correcting codes are given next.
\begin{definition}[Contextual deletion]\label{def:contextual_deletion}
    The deletion of $x_i$ in the binary sequence $\boldx=(x_1,\ldots,x_n)$ is called a \emph{contextual deletion with threshold $k$} if and only if $x_i$ is 
    the first bit in a run and the previous adjacent run has length at least $k$.
\end{definition}

This definition naturally leads to the following combinatorial error models for contextual deletions.

\begin{definition}[Zero-error contextual deletion-correcting code]\label{def:contextual-deletion-channel}
    A code $\cC\subseteq\{0,1\}^n$ is a \emph{$(t,k)$-contextual deletion-correcting code} if it can correct any pattern of up to $t$ contextual deletions with threshold $k$.
\end{definition}

Note that every $t$-deletion-correcting code is also a $(t,k)$-contextual deletion-correcting code for any threshold $k\geq 1$.
Our goal, then, is to understand what improvements (e.g., in terms of redundancy) are possible by only having to correct the more structured patterns of \emph{contextual deletions}, as a function of the threshold $k$.

The following definition extends the combinatorial formulation to a probabilistic setting in which the symbols following a runlength longer than the threshold $k$ is deleted in a deterministic manner,  with probability one.

\begin{definition}[Contextual deletion channel]
    Fix $p\in[0,1]$.
    The \emph{contextual deletion channel with threshold $k$ and deletion probability $p$},
        denoted by $\fD_{k,p}$,
        is defined as follows:
    For any input $\mathbf{x}=(x_1,\ldots,x_n)\in\bits^n$,
    each bit $x_i$ that is a possible location for a contextual deletion (see \Cref{def:contextual_deletion})
    gets deleted independently with probability $p$.
\end{definition}

Although our focus is on combinatorial contextual deletions and zero-error codes, we introduce the probabilistic model both because of technical relevance (since it more accurately captures actual nanopore sequencing errors) and because it motivates a curious extremal combinatorial setting.
Ideally, we would like to determine the capacity of the contextual deletion channel as a function of $k$ and $p$ (this channel falls into the general class of context-dependent channels studied in prior work, for which we know that information capacity equals coding capacity).
A natural first step towards this is to understand the limiting behavior of the capacity when $p=0$ and $p= 1$.
For many synchronization channels, these limiting points are trivial.
However, for the contextual deletion channel, the capacity at $p=1$ is far from immediate. It corresponds to a channel that deletes \emph{all} input bits that are possible locations for a contextual deletion.
Determining the capacity of this extremal channel is a purely combinatorial problem.

\subsection{Our contributions}

We focus on the study of combinatorial contextual deletions and obtain results for two complementary regimes:
\begin{itemize}
    \item \emph{threshold $k$ logarithmic in the block length $n$ and constant number of errors $t$.} In this case, we obtain upper and lower bounds on the redundancy and construct explicit codes. In particular,  when $k\geq \frac{1}{2}\log n$ our codes require a strictly smaller redundancy than that of any $t$-deletion-correcting code;

    \item \emph{the extremal contextual deletion channel with constant threshold $k$.}
    We obtain bounds on the coding capacity of the extremal contextual deletion channel with a constant threshold $k$. These are bounds on the rate (equivalently, bounds on the redundancy) of the largest zero-error code for the the extremal contextual deletion channel.
\end{itemize}
Pointers between results discussed in this section and their respective derivations in later sections can be found in \Cref{sec:org}.

\paragraph{Logarithmic threshold, constant number of deletions}

With respect to the first setting, we note that we are essentially interested in understanding how much better one can do than naively use a $t$-deletion-correcting code and how the redundancy behaves depending on whether $C\geq 1$ or $C<1$. For $C\geq 1$, there exist $(t,k)$-contextual deletion-correcting codes with constant redundancy, independent of $t$. This simply follows because we can encode any bitstring into a bitstring with runs of length at most $\log n$ by adding only a constant number of redundant bits, and such runlength-limited strings are not affected by contextual deletions with a threshold $k\geq \log n$. Although this is a fairly simple observation, we summarize it in the following theorem to contrast the result with that for  contextual deletion-correcting codes with  threshold $k<\log n$. In order to formally state  the result, we remark that our asymptotic notation is with respect to the block length, i.e., for $n\to \infty$. 

\begin{theorem}[Constant-redundancy codes for $k\geq \log n$]\label{thm:redundancy-const}
    If $k\geq \log n$, then there exists a $(t=n,k)$-contextual deletion-correcting code $\mathcal{C}\subseteq\{0,1\}^n$ with redundancy $O(1)$.
\end{theorem}

In contrast, for $C<1$, the redundancy required to correct a constant number of contextual deletions grows with $n$.
\begin{theorem}[Redundancy lower bound]\label{thm:redundancy-lb}
    Fix a constant integer $t\geq 0$ and let $k=C\log n$, where $C\in(0,1)$ is a constant.
    Then, any $(t,k)$-contextual deletion-correcting code $\mathcal{C}\subseteq\{0,1\}^n$ has redundancy at least $(1-C)\,t\log n-O(t\log^2\log n)$.
\end{theorem}
The above lower bound suggests the possibility that codes correcting contextual deletions when $k = C\log n$ may require much less redundancy than codes correcting worst-case deletions, as we know that for the latter the redundancy is at least $(1-o(1))\,t\log n$~\cite{Lev65}. 
We show that this is indeed the case when $C > 1/2$. 

\begin{theorem}[Non-constructive redundancy upper bound] \label{thm:redundancy-ub}
    For any constant integer $t\geq 0$ and $k=C\log n$ with $C\in(0,1)$ a constant, there exists a $(t,k)$-contextual deletion-correcting code with redundancy at most $(2(1-C)+o(1))\,t\log n$.
\end{theorem}

It is instructive to compare the upper bound from \Cref{thm:redundancy-ub} with the best known redundancy upper bound for $t$-deletion-correcting codes, which is $(2+o(1))\,t\log n$~\cite{Lev65,ABGHK24}.
\Cref{thm:redundancy-ub} improves on this bound for any $C>0$, and goes below the \emph{lower bound} on the redundancy of $t$-deletion-correcting codes when $C>1/2$. 

We prove \Cref{thm:redundancy-ub} via a Gilbert-Varshamov-type argument.
Then, it is natural to ask what redundancy can be achieved with codes supporting \emph{efficient} encoding and decoding procedures (i.e., encoding and decoding procedures computable in time polynomial in the block length $n$).
We make progress in this direction, as summarized in the following theorem.

\begin{theorem}[Efficiently encodable and decodable codes]\label{thm:eff-codes}
    Let $k=C\log n,$ where $C\in(1/2,1)$ is a constant.
    Then, for any constant integer $t\geq 1$, small enough $\varepsilon>0$,
    and large enough $n$, there exist efficiently encodable and decodable $(t,k)$-contextual deletion-correcting codes of block length $n$ with redundancy
\begin{enumerate}
        \item $(2(1-C)+\varepsilon)\log n$ for $t=1$;
        \item $(8(1-C)+\varepsilon)\log n$ for $t=2$;
        \item $(8(1-C)t+\varepsilon)\log n$ for $t\geq 3$.
\end{enumerate}
    In all cases the encoding and decoding procedure runs in time $n^{O(t)}$.
\end{theorem}

When $t=1$ the redundancy of our efficient codes in \Cref{thm:eff-codes} matches the redundancy guaranteed by the nonconstructive bound from \Cref{thm:redundancy-ub} for any $C\in(1/2,1)$. On the other hand, for $t=2$ the redundancy exceeds that of \Cref{thm:redundancy-ub} by a multiplicative factor of $2$, and for general $t>2$ the redundancy of our efficient codes exceeds that of \Cref{thm:redundancy-ub} by a multiplicative factor of $4$.
For all of these cases the redundancy of our efficient codes beats that of the best known codes correcting single, double, or $t$  
worst-case deletions~\cite{Lev65,guruswami2021explicit,sima2020optimal-systematic} for any $C\in(1/2,1)$, and becomes smaller than known lower bounds on the redundancy of single, double, or $t$ worst-case deletion-correcting codes when $C$ is large enough.
We leave it as an interesting open problem to construct non-trivial efficient $(t,k=C\log n)$-contextual deletion-correcting codes for a wider range of $C$.

The decoding complexity of the codes behind \Cref{thm:eff-codes} is $n^{O(t)}$.
Motivated by this, we construct another family of $(t,k=C\log n)$-contextual deletion-correcting codes with decoding complexity $\poly(n)$ independent of $t$, but with worse redundancy.

\begin{theorem}\label{thm:eff-codes-gen-poly}
    Let $k=C\log n$, where $C\in(0,1)$ is a constant.
    Then, 
        for each constant integer $t\geq 1$
        and $n$ large enough,
        there exists a
    $(t,k)$-contextual deletion-correcting code of block length $n$ with redundancy 
\begin{align*}
    18t(1-C)\log n+\left((2C+4)\left\lceil\frac{3}{C}\right\rceil+4\right)\log n + o(\log n),
\end{align*}
    where
    the encoding and decoding time complexities are 
    $\poly(n)$.
In particular,
    the runtime of both the encoding and decoding procedure is upper-bounded by a polynomial in $n$ whose degree does not depend on $t$.
\end{theorem}

\paragraph{Extremal contextual deletion channel, constant threshold}

To complement the results where we focus on logarithmic threshold $k$ and small $t$, we also study codes for the extremal contextual deletion channel with small threshold $k$.

A naive lower bound on the coding capacity of the extremal contextual deletion channel with threshold $k$ can be obtained by either considering unconstrained deletion correcting codes, or by only allowing codewords of length  $n$ that satisfy the symmetric \((0,k)\) run–length–limited (RLL) constraint. The former are clearly suboptimal. Furthermore, the RLL  constraint requires that strings do not contain runs of zeros or ones of length longer than \(r\). Put differently, the codebook is the set of strings obtained by forbidding the patterns $0^k$ and $1^k$.
Such strings are never subject to contextual deletions with threshold $k$, and are uniquely decodable. Still, as illustrated in Table~\ref{tab:capacity_bounds_numerical}, the RLL approach, which comes with highly efficient encoders and decoders, leads to significant reductions in the coding rate. A simple improvement is achieved by noting that it suffices to forbid the patterns $0^k10$ and $1^k01$.

We go beyond these simple capacity lower bounds by analyzing codes induced by more sophisticated sets of forbidden patterns.
To complement this, we also obtain capacity \emph{upper bounds} by identifying sets of patterns such that every possible channel output is produced by some string avoiding these patterns.
We can then count the number of strings avoiding these patterns using standard techniques~\cite{odlyzko1985enumeration}.
To illustrate this, Appendix~\ref{app:0k10} describes, as an example, the application of these standard techniques to enumerate the number of strings avoiding the patterns $0^k 1 0$ and $1^k 0 1$, which yields a worse lower bound.
The set of forbidden patterns we consider and the associated decoding correctness argument are more complex, but the techniques for counting them extend easily, as discussed below.

The following result summarizes the sets of forbidden patterns we study and the links to capacity bounds for the extremal contextual deletion channel.
\begin{theorem}\label{thm:capacity_p1}
    Define
    \begin{align*}
    \mathcal{E}_0&\coloneqq \{0^k100,0^k1010,\ldots,0^k101^{k-2}0,0^k101^k\},\\
    \mathcal{F}_0&\coloneqq \{0^{k+1}1^k001,0^{k+1}1^k0001,\ldots,0^{k+1}1^k0^{k-1}1,0^{k+1}1^k0^{k+1}\},
\end{align*}
    and let $\mathcal{E}_1$ and $\mathcal{F}_1$ denote the sets of bit-wise complements of strings in $\mathcal{E}_0$ and $\mathcal{F}_0$, respectively.
    Define $\mathcal{E}=\mathcal{E}_0\cup \mathcal{E}_1$ and $\mathcal{F}=\mathcal{F}_0\cup \mathcal{F}_1$. Let $\mathcal{H}_n$ be the collection of length-$n$ binary sequences
    that contain no substrings from $\mathcal{E}\cup\{0^{k+1}1^k00,1^{k+1}0^k11\}$,
    and
    let $\mathcal{J}_n$ be the collection of length-$n$ binary sequences
    that forbid substrings from $\mathcal{E}\cup\mathcal{F}$.
    Then,
    the capacity of the extremal contextual deletion channel with threshold $k$ (i.e., the channel $\fD_{k,1}$) 
    is lower-bounded by $\log\xi_k$ and upper-bounded by $\log\nu_k$,
    where
\begin{align}
    \xi_k &\coloneqq \liminf_{n\rightarrow\infty}|\mathcal{H}_n|^{\frac{1}{n}},\label{eq:xi_k}\\
    \nu_k &\coloneqq \limsup_{n\rightarrow\infty}\left|\bigcup_{i=1}^n\mathcal{J}_{i}\right|^{\frac{1}{n}}.\label{eq:nu_k}
\end{align}
For detailed derivations, please refer to Appendix~\ref{app:0k10} and~\cite{odlyzko1985enumeration}, in which these two limits are connected to specialized roots of polynomials arising from appropriately constructed generating functions.
\end{theorem}

By \Cref{thm:capacity_p1}, these values then yield capacity bounds for the extremal contextual deletion channel with threshold $k$.
\Cref{tab:capacity_bounds_numerical} reports the bounds obtained for selected values of $k$ and compares them to
    the RLL lower bound \cite{immink2022innovation} (forbidding the patterns $0^k$ and $1^k$) and 
    the baseline lower bound obtained by forbidding $0^k10$ and $1^k01$.
It can be seen that for $k=2$,
    even with the simple forbidden pattern set $\{0^k10,1^k01\}$,
    we already get a significant capacity gain over RLL codes.
Moreover, the improved bounds we obtain compared to the RLL and baseline lower bounds are quite sharp already for small values of $k$.
For example, our best upper and lower bounds for $k=3$ differ by less than $0.002$. As a relative comparison, for $k=3,4,5$ the gap between our best lower and upper bounds is more than $7$, $15$, and $30$ times smaller than the gap between our upper bound and the baseline lower bound, respectively.

\begin{table}
    \centering
    \begin{tabular}{c||c||c||c||c}
         $k$ & 
         RLL lower bound & 
         baseline lower bound &  $\log\xi_k$ (lower bound) & $\log\nu_k$ (upper bound)  \\
         \hline
2 & 0 & 0.6942419  & 0.7911962  & 0.8128328\\
3  & 0.6942419 & 0.8791464  & 0.8929480  & 0.8949465\\
4  & 0.8791464 & 0.9467772 & 0.9491365  & 0.9493038\\
5 & 0.9467772 & 0.9752253  & 0.9756974 & 0.9757134\\
6  & 0.9752253 & 0.9881087 & 0.9882125 &                     \\
7 & 0.9881087 & 0.9941917 & 0.9942159           &          \\
8 & 0.9941917 & 0.9971343 & 0.9971401            &  
    \end{tabular}
    \vspace{0.1in}
    \caption{Some values of the capacity bounds $\log \xi_k$ and $\log \nu_k$ for the extremal contextual deletion channel,
        where $\xi_k$ and $\nu_k$
        are defined in \Cref{thm:capacity_p1}.
        For the sake of comparison, we also include
            the RLL lower bound \cite{immink2022innovation} (obtained by forbidding the patterns $0^k$ and $1^k$)
            as well as
            the baseline lower bound obtained by forbidding the patterns $0^k10$ and $1^k01$.}
    \label{tab:capacity_bounds_numerical}
\end{table}

\subsection{Related work}

    \paragraph{Binary codes correcting worst-case deletions} The original work by Levenshtein~\cite{Lev65} established that the minimal redundancy $\text{red}(n,t)$ of a binary code correcting $t$ worst-case deletions satisfies 
    \begin{equation} \label{eq:ub-lb-adv}
        t\cdot \log n - O_t(1) \leq \text{red}(n,t) \leq 2t\cdot \log n + O_t(1)\;,    
    \end{equation}
   where $O_t(\cdot)$ indicates that the (hidden) constant may depend on $t$. 
    Subsequently, the lower bound was improved by Kulkarni and Kiyavash~\cite{KK13} and Cullina and Kiyavash~\cite{CK14}, and the upper was improved by Alon, Bourla, Graham, He, and Kravitz~\cite{ABGHK24}.
    
    Constructing efficiently encodable and decodable codes that achieve or get close to the above bounds for all constant values of $t$ remains an important open problem. 
    For $t=1$, Levenshtein \cite{Lev65} showed that the Varshamov-Tenengolt codes~\cite{varshamov1965codes}, originally designed to correct an asymmetric error, are also optimal for correcting a single deletion (or insertion) error. The case $t=2$ was studied in a sequence of works~\cite{gabrys2018codes,sima2019two,guruswami2021explicit}. In particular, Guruswami and H{\aa}stad~\cite{guruswami2021explicit} constructed efficient codes with redundancy $4\log n + O(\log \log n)$, asymptotically matching the upper bound in \Cref{eq:ub-lb-adv}. 
    For $t>2$, the first efficient construction with redundancy subpolynomial in $n$ was obtained by Brakensiek, Guruswami, and Zbarsky~\cite{brakensiek2017efficient}, which equals $O(t^2\log t \log n)$. 
    Later, Sima and Bruck \cite{sima2020optimal} presented an efficient construction with redundancy $8t\log n + o(\log n)$ while Sima, Gabrys, and Bruck~\cite{sima2020optimal-systematic} introduced an efficient \emph{systematic} construction with redundancy $4t \log n + O(\log \log n)$.
    Other works have studied edit error-correcting codes in the regime where the number of errors grows with the block length, and we now have efficient codes with order-optimal redundancy for a wide range of the number of errors $t$~\cite{Hae19,CJLW22}.

    \paragraph{Channels with context-dependent synchronization errors}
    Some relatively recent works have studied probabilistic channels with context-dependent synchronization errors~\cite{MDK18,LT21,con2025channels}, mostly motivated by connections to DNA-based data storage~\cite{yazdi2017portable}.
    In particular, these works extend the noisy channel coding theorem from channels with independent and identically distributed synchronization errors due to Dobrushin~\cite{Dob67} to channels with a wide range of context-dependent synchronization errors.
    The probabilistic contextual deletion channel from \Cref{def:contextual-deletion-channel} satisfies the conditions laid out in~\cite{con2025channels}, and so their results apply to this channel as well.
    Since our focus in this work is on combinatorial errors, the aforementioned results are not of direct relevance.

\subsection{Organization}\label{sec:org}

We start our exposition by introducing the deletion models and by providing a review of the main results. 
We then present proofs of our bounds on the redundancy for correcting contextual deletions in \Cref{sec:bounds}. 
More precisely, we prove \Cref{thm:redundancy-const} in \Cref{sec:Cgeq1}, \Cref{thm:redundancy-lb} in \Cref{sec:red-lb}, and \Cref{thm:redundancy-ub} in \Cref{sec:red-ub}.
Efficiently encodable and decodable codes for threshold $k=C\log n$ and arbitrary constant $t$ are studied in \Cref{sec:vtcodes}.
More precisely,
    we prove the $t=1$, $t=2$,  and $t\geq 3$ claims of~\Cref{thm:eff-codes} from Sections  \ref{subsec:1contexul_del_efficient},
    \ref{subsec:2contexul_del_efficient},
    and \ref{subsec:t_contexul_del_efficient}, 
    respectively.
Finally, our bounds on the coding capacity of the extremal contextual deletion channel, described in \Cref{thm:capacity_p1}, are proved in \Cref{sec:capacity}.

\section{Bounds on the redundancy of contextual deletion-correcting codes for logarithmic threshold and constant number of errors} \label{sec:bounds}

In this section, we study the redundancy of $(t,k)$-contextual deletion-correcting codes $\mathcal{C}\subseteq\{0,1\}^n$ with logarithmic threshold $k=\Theta(\log n)$ and constant number of deletions $t$.
The results obtained in this section are summarized in \Cref{thm:redundancy-const,thm:redundancy-lb,thm:redundancy-ub}, which are proved in \Cref{sec:Cgeq1,sec:red-lb,sec:red-ub}, respectively.

\subsection{The case $k\geq \log n$}\label{sec:Cgeq1}

We begin by considering the regime where the threshold $k\geq \log n$.
We show that in this case there are $(t,k)$-contextual deletion-correcting codes with constant redundancy, leading to \Cref{thm:redundancy-const}.
In short, the results of the theorem hold because we can encode any binary string into another binary string with runs of length at most $\log n$ using only a constant number of redundant bits, and such runlength-limited strings do not suffer from contextual deletions with threshold $k\geq \log n$.

We now formally prove this claim by invoking a result that asserts that one can encode an arbitrary $\ell$-bit string into an $(\ell+1)$-bit string without ``long'' runs.

\begin{theorem}[\protect{\cite[Appendix B]{schoeny17codes}}]\label{thm:rll}
There exists an injective mapping $E:\{0,1\}^\ell\rightarrow\{0,1\}^{\ell+1}$ such that for any $\boldx\in\{0,1\}^{\ell}$ it holds that
$E(\boldx)$ only has runs of length at most $\lceil\log \ell\rceil+3$.
Furthermore, both $E$ and its inverse $E^{-1}$ can be computed with time complexity $O(\ell)$.
\end{theorem}

The above result is used for encoding a binary string $\mathbf{x}\in\{0,1\}^n$ as follows:
\begin{enumerate}
    \item Split $\mathbf{x}$ into
    $64$
    consecutive substrings  $\boldx_1,\boldx_2,\ldots,
    \boldx_{64},$
    each of length $\lceil n/64\rceil$, except for the last substring which may have shorter length (or be empty).

    \item Encode each block $\boldx_i$ using the runlength-limited encoding from \Cref{thm:rll} to obtain $E(\boldx_i)$.
    Note that $E(\boldx_i)$ only has runs of length at most
    $\lceil\log(\lceil n/64\rceil)\rceil+3\leq \log n -1$.

    \item To finalize the encoding, we concatenate the blocks $E(\boldx_1),\dots,
    \boldx_{64}$
    as follows.
    For $i\in\{1,\dots,
        63\}$
        define $y_i$ to be the bit-complement of the last bit of $E(\boldx_i)$.
    Then, the encoding of $\boldx$ is
    \begin{equation*}
    E^\star(\boldx)=(E(\boldx_1),y_1,E(\boldx_2),y_2,\ldots,
    E(\boldx_{64})).
    \end{equation*}
\end{enumerate}

Since each $E(\boldx_i)$ only has runs of length at most $\log n-1$, it follows that after prepending $y_{i-1}$ to each $E(\boldx_i)$ the maximal run length increases by at most $1$.
Therefore, $E^\star(\boldx)$ only has runs of length at most $\log n$, and so is not subject to contextual deletions with threshold $k\geq \log n$.
Regarding the redundancy, each encoding $E(\boldx_i)$ adds one bit of redundancy, and so do the buffers $y_1,\dots,
    y_{63}$.
In total, there are
    127
    redundant bits.
This yields \Cref{thm:redundancy-const}.

\subsection{Redundancy lower bound for threshold $k<\log n$}\label{sec:red-lb}

We now turn our attention to the regime where $k=C\log n$ for some constant $C\in(0,1)$ and $t\geq 0$ is an arbitrary constant.
We begin by establishing the redundancy lower bound in \Cref{thm:redundancy-lb}, which in particular shows that in this setting the redundancy grows as  $\Omega(\log n)$ (recall that when $k\geq \log n,$ constant redundancy suffices).
We restate the result for convenience. 
\begin{theorem}[\Cref{thm:redundancy-lb}, restated]\label{thm:redundancy-lb-restated}
Fix a constant integer $t\geq 0$ and $k=C\log n$ with $C\in(0,1)$ a constant.
    Then, any $(t,k)$-contextual deletion-correcting code $\mathcal{C}\subseteq\{0,1\}^n$ has redundancy at least $(1-C)t\log n-O(t\log\log n)$.
\end{theorem}
\begin{proof}
We first show that the number of length-$n$ sequences with fewer than 
$\frac{n}{(k+2)^22^{k+2}}$
runs of length at least $k$ is at most $2^{n-\frac{kn}{(k+2)^22^{k+2}}}$.
More precisely, we define
\begin{align*}
    i^{\star}\coloneqq \frac{n}{(k+2)^22^{k+2}},
\end{align*}
    and 
\begin{align*}
\mathcal{A}\coloneqq\{\boldx\in\{0,1\}^n~:~\boldx \tn{ has fewer than } i^{\star} \tn{ runs of length at least }k\}.
\end{align*}
We aim to show that
\begin{align}
    |\mathcal{A}|\leq 2^{n-ki^{\star}}.
    \label{eq:red_lb_A_ub_goal}
\end{align}

To this end, we consider another set of length-$n$ binary sequences described as follows.
We assume that $k+2$ divides $n$ for simplicity\footnote{If $k+2$ does not divide $n$, we should use $\lfloor \frac{n}{k+2}\rfloor$ instead but this does not affect the overall conclusion of the analysis.}.
Then, for any $\boldx\in\{0,1\}^n$ we can split it
    into exactly $\frac{n}{k+2}$ blocks of length $k+2$.
Now, we let
\begin{align*}
\mathcal{B}\coloneqq\{\boldx\in\{0,1\}^n~:~\boldx \tn{ has fewer than } i^{\star} \tn{ blocks equal to either }01^k0\tn{ or }10^k1\}.
\end{align*}
First, we claim that $\mathcal{A}\subseteq\mathcal{B}$.
To see this,
    we show the contrapositive statement $\mathcal{B}^C\subseteq\mathcal{A}^C$.
If $\boldx\in\mathcal{B}^C$,
    then $\boldx$ has at least $i^{\star}$ blocks equal to either $01^k0$ or $10^k1$.
Then, these blocks alone guarantees that
    $\boldx$ has at least $i^{\star}$ runs of length at least $k$.
Thus, $\boldx\in\mathcal{A}^C$ and $\mathcal{B}^C\subseteq\mathcal{A}^C$, which proves the claim.

We now establish an upper bound on $|\mathcal{B}|$.
To this end,
    for each $i\in[0,\frac{n}{k+2}]$,
    let $B(i)$ denote the number of length-$n$ binary sequences with exactly $i$ blocks equal to $01^k0$ or $10^k1$.
The exact formula for $B(i)$ is
\begin{align*}
    B(i)=\binom{\frac{n}{k+2}}{i}2^i(2^{k+2}-2)^{\frac{n}{k+2}-i},
\end{align*}
so that
\begin{align}
    |\mathcal{B}|=\sum_{i=0}^{i^{\star}}B(i).
    \label{eq:red_lb_B_size}
\end{align}
Note that $B(i)$ is an increasing function for $i\in [0,i^{\star}]$.
This can be seen by considering the following ratio for each $i\in[i^{\star}]$:
\begin{align}
    \frac{B(i)}{B(i-1)}
    &=\frac{\frac{n}{k-2}-i+1}{i}\frac{2}{2^{k+2}-2}\nonumber\\
    &\geq \frac{\frac{n}{k-2}-i^{\star}+1}{i^{\star}}\frac{2}{2^{k+2}-2}.
    \label{eq:red_lb_ratio}
\end{align}
We now analyze the asymptotic order of every term in \eqref{eq:red_lb_ratio}.
In particular,
    we have
    $\frac{n}{k-2}=\Theta(\frac{n}{\log n})$,
    $i^{\star}=\Theta(\frac{n^{1-C}}{\log^2n})$,
    and
    $2^{k+2}=\Theta(n^{C})$.
Therefore,
    we have
\begin{align*}
    2(\frac{n}{k-2}-i^{\star}+1)&=\Theta(\frac{n}{\log n}),\\
    i^{\star}(2^{k+2}-2)&=\Theta(\frac{n}{\log^2n}).
\end{align*}
It consequently follows that
\begin{align*}
    \frac{\frac{n}{k-2}-i^{\star}+1}{i^{\star}}\frac{2}{2^{k+2}-2}=\Theta(\log n).
\end{align*}
Therefore,
    from \eqref{eq:red_lb_ratio}
    we can deduce that
    $\frac{B(i)}{B(i-1)}=\Omega(\log n)$
    for each $i\in[1,i^{\star}]$.
In particular,
    we have for $n$ large enough that $B(i)\geq B(i-1)$ for each $i\in[1,i^{\star}]$.

Since $B(i)$ is increasing on $ [0,i^{\star}]$,
    from \eqref{eq:red_lb_B_size}
    we have
\begin{align}
    |\mathcal{B}|\leq(i^{\star}+1)B(i^{\star}).
    \label{eq:red_lb_B_size_ub}
\end{align}
We now upper-bound the quantity 
\begin{align}
    B(i^{\star})=\binom{\frac{n}{k+2}}{i^{\star}}2^{i^{\star}}(2^{k+2}-2)^{\frac{n}{k+2}-i^{\star}}.
    \label{eq:red_lb_B_i_star}
\end{align}
First,
    using the inequality $\binom{a}{b}\leq \left(\frac{ae}{b}\right)^b$,
    we obtain
\begin{align*}
    \binom{\frac{n}{k+2}}{i^{\star}}&\leq \left(\frac{ne}{(k+2) i^{\star}}\right)^{i^{\star}}\nonumber\\
    &=\left(e(k+2)2^{k+2}\right)^{i^{\star}},
\end{align*}
    and thus taking the logarithm of both sides arrive at
\begin{align}
    \log \binom{\frac{n}{k+2}}{i^{\star}} \leq \left(k+2 + \log (k+2) + \log e\right)i^{\star}.
    \label{eq:red_lb_log_binom_ub}
\end{align}
Next,
    we simplify
\begin{align}
    (2^{k+2}-2)^{\frac{n}{k+2}-i^{\star}}
    &=\left(2^{k+2}\left(1-\frac{1}{2^{k+1}}\right)\right)^{\frac{n}{k+2}-i^{\star}}\nonumber\\
    &=2^{n-(k+2)i^{\star}}\left(1-\frac{1}{2^{k+1}}\right)^{\frac{n}{k+2}-i^{\star}}.\label{eq:red_lb_simplify}
\end{align}
Using the inequality $1-x\leq e^{-x}$ with $x=\frac{1}{2^{k+1}}$, from~\eqref{eq:red_lb_simplify} we obtain 
\begin{align}
    (2^{k+2}-2)^{\frac{n}{k+2}-i^{\star}}
    &\leq
    2^{n-(k+2)i^{\star}}
    \left(e^{-\frac{1}{2^{k+1}}}\right)
    ^{\frac{n}{k+2}-i^{\star}}\nonumber\\
    &=
    2^{n-(k+2)i^{\star}}
    e^{-\frac{n}{2^{k+1}(k+2)}+\frac{i^{\star}}{2^{k+1}}}\nonumber\\
    &=2^{n-(k+2)i^{\star}}
    e^{-2(k+2)i^{\star}+\frac{i^{\star}}{2^{k+1}}}.
    \label{eq:red_lb_ineq_after_simplify}
\end{align}
Taking the logarithm of both sides of~\eqref{eq:red_lb_ineq_after_simplify} leads to
\begin{align}
    \log\left((2^{k+2}-2)^{\frac{n}{k+2}-i^{\star}}\right)
    &\leq 
    n-(k+2)i^{\star}+\left(-2(k+2)i^{\star}+\frac{i^{\star}}{2^{k+1}}\right)\log e\nonumber\\
    &=n-(k+2)i^{\star}-2(\log e)(k+2)i^{\star}+\frac{(\log e) }{2^{k+1}}i^{\star}\nonumber\\
    &=n-\left((2\log e+1)(k+2)-\frac{\log e}{2^{k+1}}\right)i^{\star}.
    \label{eq:red_lb_log_ineq_after_simplify}
\end{align}
Then,
    adding \eqref{eq:red_lb_log_binom_ub} and~\eqref{eq:red_lb_log_ineq_after_simplify} establishes
\begin{align}
    \log\left(\binom{\frac{n}{k+2}}{i^{\star}}(2^{k+2}-2)^{\frac{n}{k+2}-i^{\star}}\right)
    &\leq
    n-\left((2\log e)(k+2)-\log(k+2)-\log e -\frac{\log e}{2^{k+1}}\right)i^{\star}.
    \label{eq:red_lb_added}
\end{align}
Next, we add $i^{\star}$ to both sides of~\eqref{eq:red_lb_added}
    and recall the formula of $B(i^{\star})$ in \eqref{eq:red_lb_B_i_star} to obtain
\begin{align}
    \log\left(B(i^{\star})\right)\leq
    n-\left((2\log e)(k+2)-\log(k+2)-\log e -\frac{\log e}{2^{k+1}}-1\right)i^{\star}.
    \label{eq:red_lb_log_B_i_star_ub}
\end{align}
Finally,
    since $(2\log e)\approx 2.885>1$ and
\begin{align*}
    (2\log e)(k+2)-\log(k+2)-\log e -\frac{\log e}{2^{k+1}}-1=(2\log e+o(1))(k+2),
\end{align*}
    for $n$ large enough we have $(2\log e)(k+2)-\log(k+2)-\log e -\frac{\log e}{2^{k+1}}-1\geq k+1$.
Therefore,
    from~\eqref{eq:red_lb_log_B_i_star_ub}, for $n$ large enough we also have
\begin{align*}
    \log\left(B(i^{\star})\right)\leq n-(k+1)i^{\star},
\end{align*}
    or equivalently,
\begin{align}
    B(i^{\star})\leq 2^{n-(k+1)i^{\star}}.
    \label{eq:red_lb_B_i_star_lb_result}
\end{align}

Finally,
    note that $i^{\star}+1\leq 2^{i^{\star}}$
    since $x+1\leq 2^x$ for $x\geq 1$.
Therefore, inserting this inequality and~\eqref{eq:red_lb_B_i_star_lb_result} into~\eqref{eq:red_lb_B_size_ub} results in
\begin{align*}
    |\mathcal{B}|\leq 2^{i^{\star}}2^{n-(k+1)i^{\star}}=2^{n-ki^{\star}},
\end{align*}
    which, together with the fact $\mathcal{A}\subseteq\mathcal{B}$,
    implies
\begin{align}
    |\mathcal{A}|\leq 2^{n-ki^{\star}}.
    \label{eq:red_lb_A_size_ub}
\end{align}

Now we can proceed with the redundancy lower bound.
Let $\mathcal{C}\subseteq\{0,1\}^n$ be a $(t,k)$-contextual deletion-correcting code.
We decompose $\mathcal{C}$
    into
    $\mathcal{C}=(\mathcal{C}\cap\mathcal{A})\cup(\mathcal{C}\cap\mathcal{A}^C)$,
    and thus
\begin{align}
    |\mathcal{C}|=|\mathcal{C}\cap\mathcal{A}|+|\mathcal{C}\cap\mathcal{A}^C|.
    \label{eq:red_lb_C_size_decompose}
\end{align}

We can upper-bound the first term in \eqref{eq:red_lb_C_size_decompose} as follows:  
From \eqref{eq:red_lb_A_size_ub} we have for $n$ large enough that
\begin{align}
|\mathcal{C}\cap\mathcal{A}| \le|\mathcal{A}|
    \le 2^{n-ki^{\star}}.
    \label{eq:red_lb_C_cap_A_size_ub}
\end{align}

We now focus on upper bounding the size of $\mathcal{C}\cap\mathcal{A}^C$, which is the set of codewords containing at least $i^{\star}$ runs of length at least $k$.
Before that,
    we introduce the notion of \emph{contextual deletion balls}.
For any $\boldx\in\{0,1\}^n$, let its contextual ball 
$\mathcal{D}_t^{(k)}(\boldx)$ with radius $t$ be the set of sequences that can be obtained from $\boldx$ via at most $t$ contextual deletions with threshold $k$.
Note that for any two distinct codewords in $\mathcal{C}\cap\mathcal{A}^C$,
    their contextual deletion balls with radius $t$ do not overlap. 
Furthermore,
    the size of the contextual deletion ball of a codeword in $\{0,1\}^n\setminus \mathcal{A}$ with radius $t$ is at least $\binom{i^{\star}}{t}$, since there are at least $i^{\star}$ possible contextual deletion locations. 
Consequently,
    the number of codewords in $\mathcal{C}\cap\mathcal{A}^C$ is at most 
\begin{align*}
    |\mathcal{C}\cap\mathcal{A}^C|\leq\frac{2^n}{\binom{i^{\star}}{t}}.
\end{align*}
Using the inequality $\binom{a}{b}\geq \left(\frac{a}{b}\right)^b$,
    we obtain
\begin{align}
    |\mathcal{C}\cap\mathcal{A}^C|&\leq\frac{2^n}{\left(\frac{i^{\star}}{t}\right)^t}\nonumber\\
    &=2^{n-t(\log i^{\star}-\log t)}\nonumber\\
    &=2^{n-t(\log n - (k + 2) -2\log (k+2)-\log t) }\nonumber\\
    &=2^{n-t(1-C)\log n+O(t\log\log n)}.
    \label{eq:red_lb_C_cap_AC_size_ub}
\end{align}
Finally,
    plugging~\eqref{eq:red_lb_C_cap_A_size_ub}
    and~\eqref{eq:red_lb_C_cap_AC_size_ub}
    into~\eqref{eq:red_lb_C_size_decompose}
    yields 
\begin{align}
    |\mathcal{C}|
    &\le 2^{n-ki^{\star}}+2^{n-t(1-C)\log n+O(t\log\log n)}.
    \label{eq:red_lb_almost}
\end{align}
We now compare the exponents of the two terms in \eqref{eq:red_lb_almost}.
Since $ki^{\star}=\Theta(\frac{n^{1-C}}{\log n})$,
    we have $ki^{\star}>t(1-C)\log n+O(t\log\log n)$
    for $n$ large enough.
In other words,
    the first term in \eqref{eq:red_lb_almost} is at most the second term,
    which leads to 
\begin{align*}
    |\mathcal{C}|
    &\le 2^{n-t(1-C)\log n+O(t\log\log n)}.
\end{align*}
That is,
    $\mathcal{C}$ has redundancy at least $t(1-C)\log n+O(t\log\log n)$.
\end{proof}

%%%%% Commands defined by Yuan-Pon
\def\E#1{\ensuremath{\mathbb{E}\left[#1\right]}}
\def\Prob#1{\ensuremath{\mathbb{P}\left(#1\right)}}
\def\Var#1{\ensuremath{\textnormal{Var}\left(#1\right)}}
\def\Cov#1#2{\ensuremath{\textnormal{Cov}\left({#1},{#2}\right)}}
%%%%%

\subsection{A Gilbert-Varshamov-type bound for contextual deletion-correcting codes}\label{sec:red-ub}

We show next that there exists a $(t,k)$-contextual deletion-correcting code with redundancy at most $2t(1-C)\log n$.
Note that this is at most a fraction $(1-C)$ of the redundancy required for a general $t$-deletion correcting code.

To this end, we first show
    via the probabilistic method
    that for most strings the total length of runs of length at least $k$ can be upper bounded by
    $\frac{n\log^2n}{2^{k-1}}$.
\begin{lemma}\label{lem:runlength}
Let $\widehat{\mathcal{R}}_k$ be the collection of length-$n$ binary sequences such that
    the total length of runs of length at least $k$ is at most $\frac{n\log^2n}{2^{k-1}}$.
Then, we have $|\widehat{\mathcal{R}}_k|=2^n(1-o(1))$.
\end{lemma}
\begin{proof}
Let $\boldx$ be sampled uniformly at random from $\bits^n$.
We first show that
    the probability that $\mathbf{x}$ has more than $\frac{n\log n}{2^k}$ runs of length at least $k$ is at most $\frac{1}{\log n}=o(1)$.
Let $N_{\geq k}$ be the random variable denoting the number of runs of length at least $k$ in $\mathbf{x}$.
Note that 
    $N_{\geq k}=\sum_{i=1}^{n-k+1}N_i$,
    where $N_i$ is the indicator that the $i$th bit of $\mathbf{x}$ is the start of a run of length at least $k$. We have $\E{N_1}=2^{-k+1}$ and $\E{N_i}=2^{-k}$ for all $2\leq i\leq n-k+1$.
Therefore, by linearity of expectation
    we have
\begin{align*}
    \E{N_{\geq k}} = \frac{n-k+2}{2^k}.
\end{align*}
Applying Markov's inequality,
    we get
\begin{align}
    \Prob{N_{\geq k}>\frac{n\log n}{2^k}} \leq \Prob{N_{\geq k}>\log n\cdot  \E{N_{\geq k}}} \leq \frac{1}{\log n} = o(1).
    \label{eq:lem_runlth_1}
\end{align}

We show next that 
    the probability that $\mathbf{x}$ has a run of length at least $2\log n$ 
    is also $o(1)$.
For each $i\in[1,n-2\log n+1]$,
    define $B_i$ to be the event that all the bits 
    $x_{i},\ldots,x_{i+2\log n -1}$ are equal (i.e., part of a run).
Then, 
    $\mathbf{x}$ has a run of length at least $2\log n$
    if and only if
    at least one of the events $B_i$ occurs.
It is clear that $\Prob{B_i}=2^{-2\log n +1}=2n^{-2}$, so that the union bound yields
\begin{align}
    \Prob{\bigcup_{i}B_i} \leq \frac{2(n-2\log n+1)}{n^2} \leq \frac{2}{n}= o(1).
    \label{eq:lem_runlth_2}
\end{align}

Now let $A$ be the event that
    $\mathbf{x}$ has at most $\frac{n\log n}{2^k}$ runs of length at least $k,$
    and that it has no run of length at least $2\log n$. The event $A$ has probability $1-o(1)$
    by \Cref{eq:lem_runlth_1,eq:lem_runlth_2}.
Then, note that for each sequence $\mathbf{x}$ for which $A$ is true,
    the total length of runs of length at least $k$ in $\mathbf{x}$
    at most $\frac{n\log n}{2^{k}}\cdot 2\log n = \frac{n\log^2 n}{2^{k-1}}$.
That is,
    we have
    $\Prob{\mathbf{x}\in\widehat{\mathcal{R}}_k} \geq \Prob{A}=1-o(1)$.
Finally,
    since
    $\mathbf{x}$ is uniformly random over $\{0,1\}^n$,
    we have
    $|\widehat{\mathcal{R}}_k|= 2^n\cdot \Prob{\mathbf{x}\in\widehat{\mathcal{R}}_k}\geq 2^n(1-o(1))$.
\end{proof}

\begin{theorem}\label{thm:GV}
Let $\widehat{\mathcal{R}}_k$ be as defined in \Cref{lem:runlength}, and let $\mathcal{B}_t^{(k)}(\boldx)$ be the set of all the binary sequences that can result in any sequence in $\mathcal{D}_t^{(k)}(\mathbf{x})$ after at most $t$ contextual deletions,
    where $\mathcal{D}_t^{(k)}(\mathbf{x})$
    is defined in the proof of \Cref{thm:redundancy-lb-restated}.
Formally,
    $\mathcal{B}_t^{(k)}(\boldx)\coloneqq\{\boldx'\in\{0,1\}^{*}~:~\mathcal{D}_t^{(k)}(\boldx)\cap\mathcal{D}_t^{(k)}(\boldx')\neq \emptyset\}$.
For any
$\boldx\in\widehat{\mathcal{R}}_k,$
we have $|\mathcal{B}_t^{(k)}(\boldx)|\le 2^{(2(1-C)+o(1))t\log n}$ whenever $t\leq n^{1-C}$.
Consequently, there exists a $(t,k)$-contextual deletion-correcting code with redundancy at most
$(2(1-C)+o(1))t\log n$ whenever $t\leq n^{1-C}$.  
\end{theorem}
\begin{proof}
We first establish the number of possible positions at which we can add back a bit after one contextual deletion.
In particular,
    we show that
    for any binary sequence $\mathbf{s}$,
    the number of length-$(|\mathbf{s}|+1)$ sequences $\mathbf{x}$ satisfying
    $\mathbf{s}\in\mathcal{D}_1^{(k)}(\mathbf{x})$
    is exactly the number of occurrences of the substrings $0^k$ and $1^k$ in $\mathbf{s}$.
    
    On the one hand,
    for any $0^k$ in $\mathbf{s}$,
    adding a $1$ right after it results in a valid input sequence $\mathbf{x}$ satisfying $\mathbf{s}\in\mathcal{D}_1^{(k)}(\mathbf{x})$.
A similar argument holds for $1^k$.
On the other hand,
    if $\mathbf{s}$ is obtained from contextually deleting the bit $x_i$ from $\mathbf{x}$,
    then by definition all the $k$ bits preceding $x_i$ in $\mathbf{x}$ must be $1-x_i$.
    
Next,
    note that one can always perform contextual deletions sequentially from right to left. To be more precise,
    suppose $\boldy$ is obtained from $\boldx$ via exactly $t$ contextual deletions by deleting the bits $x_{i_1},\ldots,x_{i_t}$ from $\boldx$,
    where $i_1<\cdots<i_t$.
Then,
    consider the following recursive definition of sequences:
Define $\bolds^{(0)}\coloneqq\boldx$,
    and for each $j\in[t]$,
    define $\bolds^{(j)}$ to be the sequence obtained by deleting $x_{i_{t-j+1}}$ from $\bolds^{(j-1)}$.
Then
    we have $\bolds^{(t)}=\boldy$ and for each $j\in[t]$ that $\bolds^{(j)}\in\mathcal{D}_1^{(k)}(\bolds^{(j-1)})$.

Now let $\mathbf{s}\in\mathcal{D}_t^{(k)}(\mathbf{x})$
    with $\mathbf{x}\in\widehat{\mathcal{R}}_k$,
    where $\widehat{\mathcal{R}}_k$ is defined in \Cref{lem:runlength}.
We claim that the number of occurrences of $0^k$ and $1^k$ in $\mathbf{s}$
    is upper-bounded by $\frac{n\log^2 n}{2^{k-1}}+(k-1)t$.
First,
    note that
    the number of $0^k$ and $1^k$ in $\boldx$ is at most $\frac{n\log^2 n}{2^{k-1}}$,
    since each run of length $\ell\geq k$ in $\mathbf{x}$ contributes to $\ell-k+1\leq \ell$ occurrences of such patterns,
    and we know that the total length of all such runs is upper-bounded by $\frac{n\log^2 n}{2^{k-1}}$ by the definition of $\widehat{\mathcal{R}}_k$.
Second,
    note that a contextual deletion can only increase the number of $0^k$ and $1^k$ by at most $k$.
The reason is that if we delete $x_i$ from $\mathbf{x}=(x_1,\ldots,x_n)$,
    then
    the new length-$k$ substrings induced by this deletion are
    $x_{i+j-k},\ldots,x_{i-1},x_{i+1},\ldots,x_{i+j}$
    for $j\in[1,k-1]$.
Since there are at most $k-1$ new substrings,
    the number of $0^k$ and $1^k$ can only increase by at most $k-1$.
The claim then follows from the sequential property of contextual deletions.

We can now upper-bound the number of possible input sequences that can result in $\bolds$ 
    after at most $t$ contextual deletions.
By the sequential property of contextual deletions,
    we can add back the contextually deleted bits one by one and upper-bound the number of possible inputs.
More precisely,
    for each $t'\leq t$,
    we first identify $\bolds=\bolds^{(t')}$
    and count the number of possible $\bolds^{(t'-1)}$
    such that $\bolds^{(t')}\in\mathcal{D}_1^{(k)}(\bolds^{(t'-1)})$.
Then,
    for each possible $\bolds^{(t'-1)}$,
    we count the number of possible $\bolds^{(t'-2)}$
    such that $\bolds^{(t'-1)}\in\mathcal{D}_1^{(k)}(\bolds^{(t'-2)})$,
    and so on.
Note that
    adding back one contextually deleted bit can only increase the number of substrings $0^k$ and $1^k$ by at most one,
    which happens only when the added bit is combined with another run of length at least $k-1$.
Therefore,
    the procedure of adding back the bits one by one leads to the following conclusion:
For each $t'\leq t$,
    the number of sequences that can result in $\mathbf{s}$ after exactly $t'$ contextual deletions is upper-bounded by 
    $\prod_{j=1}^{t'}(\frac{n\log^2 n}{2^{k-1}}+(k-1)t+j-1)\leq (\frac{n\log^2 n}{2^{k-1}}+(k-1)t+t')^{t'}$.
It follows that the number of sequences that can result in $\mathbf{s}$ after up to $t$ contextual deletions is at most
$\sum_{t'=1}^t (\frac{n\log^2 n}{2^{k-1}}+(k-1)t+t')^{t'}\leq t(\frac{n\log^2 n}{2^{k-1}}+kt)^t$.

Also note that $|\mathcal{D}_t^{(k)}(\boldx)|\le\binom{\frac{n\log^2n}{2^{k-1}}}{t}$,
    since $\boldx$ has at most $\frac{n\log^2n}{2^{k-1}}$ runs of length at least $k$.
Therefore, 
    we can upper-bound the number of sequences in $\mathcal{B}_t^{(k)}(\boldx)$ 
    by
\begin{align}
    |\mathcal{B}_t^{(k)}(\boldx)|\leq t\left(\frac{n\log^2 n}{2^{k-1}}+kt\right)^t\binom{\frac{n\log^2 n}{2^{k-1}}}{t}.
    \label{eq:GV_to_bound}
\end{align}
We now upper-bound the right-hand side of \eqref{eq:GV_to_bound}.
First,
    note that
    we have
\begin{align}
    \left(\frac{n\log^2 n}{2^{k-1}}+kt\right)^t
    &=\left(2n^{1-C}\log^2 n+kt\right)^t\nonumber\\
    &\leq \left(2n^{1-C}\log^2n+n^{1-C}\log^2n\right)^t\label{eq:GV_to_bound_1_step}\\
    &=\left(3n^{1-C}\log^2 n\right)^t,
    \label{eq:GV_to_bound_1}
\end{align}
    where in \eqref{eq:GV_to_bound_1_step} we used the fact that $t\leq n^{1-C}$ and $k=C\log n\leq \log^2 n$.
Next,
    using the inequality $\binom{a}{b}\leq a^b$,
    we obtain
\begin{align}
    \binom{\frac{n\log^2 n}{2^{k-1}}}{t}
    &\leq\left(\frac{n\log^2n}{2^{k-1}}\right)^t\nonumber\\
    &=\left(2n^{1-C}\log^2n\right)^t.
    \label{eq:GV_to_bound_2}
\end{align}
Replacing~\eqref{eq:GV_to_bound_1} and \eqref{eq:GV_to_bound_2} into \eqref{eq:GV_to_bound}, and using the bound $t\leq 2^t$,
    we get
\begin{align*}
    |\mathcal{B}_t^{(k)}(\boldx)|
    &\leq \left( 12n^{2-2C}\log^4 n\right)^t\\
    &= 2^{t( (2-2C)\log n + 4\log \log n +\log 12)}\\
    &=2^{(2(1-C)+o(1))t\log n}.
\end{align*}

Since $|\widehat{\mathcal{R}}_k|=(1-o(1))2^n$ by \Cref{lem:runlength}, it follows that using a greedy algorithm to select codewords from  
$\widehat{\mathcal{R}}_k$ one can obtain a $(t,k)$-contextual deletion-correcting code with redundancy at most
$(2(1-C)+o(1))t\log n$.  
\end{proof}

\Cref{thm:redundancy-ub} corresponds to the special case of \Cref{thm:GV} with constant $t$, which is the main result of interest.

\section{Efficient $t$ contextual deletion-correcting codes via variants of Varshamov-Tenengolts codes} \label{sec:vtcodes}

In this section we prove 
\Cref{thm:eff-codes} when $C\in(1/2,1)$.

The proof is split into four parts.
First, in \Cref{sec:VT-small}, we introduce a family of ``VT-type'' codes and show that they can correct a single contextual deletion; the redundancy of these codes depends on the threshold of the contextual deletion. We do not focus on the encoding and decoding procedures for such codes.
Then, in \Cref{subsec:1contexul_del_efficient}, we show how to slightly modify the approach from \Cref{sec:VT-small} to ensure efficient encoding and decoding while only incurring an extra $o(\log n)$ bits of redundancy.
Next, in \Cref{subsec:2contexul_del_efficient}, we discuss how the results from the previous sections can be extended to the setting of $t=2$ contextual deletions.
Lastly,
    in \Cref{subsec:t_contexul_del_efficient} we modify the code in \Cref{subsec:2contexul_del_efficient}
    and construct $t$-contextual deletion-correcting codes.

\subsection{VT-type codes correcting a single contextual deletion}
\label{sec:VT-small}

We present next variants of Varshamov-Tenengolts (VT) codes capable of correcting a single contextual deletion. 
The main result is that for small enough $\varepsilon>0$ and sufficiently large $n$, 
    one can construct a VT-like code for a single contextual deletion
    with redundancy
    $(2(1-C)+\varepsilon)\log n$ and perform encoding and decoding in time polynomial in $n$. 
    Note that the contextual code redundancy is smaller than that of any single-deletion code~\cite[Theorem 2.5]{sloane2002single} whenever $C>1/2$ and $\varepsilon$ is small enough,
    and that it comes arbitrarily close to the Gilbert-Varshamov-based bound of \Cref{thm:GV} with $t=1$ (which did not guarantee efficient encoding/decoding).

The codewords of our code are structured bitstrings that also satisfy a VT-type constraint.
The required structural properties are defined in the next result, which also shows that almost all strings satisfy them.

\begin{lemma}\label{lem:C_eps}
Fix an arbitrary $C\in(1/2,1)$, set $k=C\log n$ and let $\varepsilon\in(0,
C)$ be arbitrary\footnote{If needed, we use the ceiling function to ensure integer values for parameters.}.
Define $\ell\coloneqq (1+\varepsilon/2)\log n-k=(1-C+\varepsilon/2)\log n$ and $w\coloneqq n^{1-C+\varepsilon}$.
Let $\mathcal{C}_{\varepsilon}$ be the set of
    all length-$n$ binary sequences $\mathbf{x}$ with the following properties:
    \begin{enumerate}
        \item \label{ppty_1_Prop_C_eps}
            The number of runs of length at least $k$ in $\mathbf{x}$
            is at most $\frac{n\log n}{2^k}$.
        \item \label{ppty_2_Prop_C_eps}
            $\mathbf{x}$ has no run of length at least $2\log n$.
        \item \label{ppty_3_Prop_C_eps}
            $\mathbf{x}$ does not have $0^{k}1^{\ell}$ or $1^{k}0^{\ell}$ as substrings.
        \item \label{ppty_4_Prop_C_eps}
            Every length-$w$
            substring of $\mathbf{x}$ contains at least one (possibly nonmaximal) run $0^{\ell}$ and at least one (possibly nonmaximal) run $1^{\ell}$
            as subsubstrings. 
    \end{enumerate}
Then, $\mathcal{C}_{\varepsilon}$ has size $(1-o(1))2^n$.
\end{lemma}

Note that the above lemma holds $\forall \, C\in(0,1)$, but we only need it to hold for $C>1/2$ in order for the proof of \Cref{thm:ctxl_del_subseq} to go through (there, we do want to additionally avoid the substrings $0^k1^k$ and $1^k0^k$, which in this case is guaranteed by Property (3) and the fact that $C>1/2$ implies $k>\ell$; on the other hand, if $C<1/2$, we cannot avoid $0^k1^k$ or $1^k0^k$ since they are patterns of length $2C\log n<\log n$.

\begin{IEEEproof}
Suppose that $\mathbf{x}$ is sampled uniformly at random from $\{0,1\}^n$.
It suffices to show that
    the probability that 
    $\mathbf{x}$ satisfies each property is $1-o(1)$.
The desired result then follows from the union bound.
By the proof of \Cref{lem:runlength},
    we already know that
    $\mathbf{x}$ satisfies Properties \ref{ppty_1_Prop_C_eps} and \ref{ppty_2_Prop_C_eps} with probability $1-o(1)$, so that we hence focus on the latter two properties.

The argument showing that $\mathbf{x}$ satisfies Property \ref{ppty_3_Prop_C_eps} with probability $1-o(1)$ is very similar to that used for Property \ref{ppty_2_Prop_C_eps}.
It suffices to note that the probability that the substring $(x_{i},\ldots,x_{i+(1+\varepsilon/2)\log n -1})$ equals either $0^k1^\ell$ or $1^k0^\ell$ is $2n^{-(1+\varepsilon/2)}$.
Then, using the union bound over the at most $n$ choices for $i$ shows that Property~\ref{ppty_3_Prop_C_eps} fails to be satisfied with probability at most $2n^{-\eps/2}=o(1)$.

    To show that
    $\mathbf{x}$ satisfies Property~\ref{ppty_4_Prop_C_eps}
    with probability $1-o(1)$, we first bound the probability that
    a uniformly random length-$w$ binary sequence
    $\mathbf{y}=(y_1,\ldots,y_w)$
    has no run of length at least $\ell$.
For each $i\in[1,w-\ell+1]$,
    define $D_i$ to be the event that
    $y_i,\ldots,y_{i+{\ell}-1}$ is not a $0$-run.
Then,
    the probability that $\mathbf{y}$ has no $0$-run of length at least $\ell$ is simply
    $\Prob{\bigcap_{i=1}^{w-\ell+1}D_i}$,
    which can be upper-bounded as
\begin{align}
    \Prob{\bigcap_{i=1}^{w-\ell+1}D_i}
    &\leq \Prob{\bigcap_{j=1}^{\lfloor\frac{w}{\ell}\rfloor}D_{\ell (j-1)+1}}\nonumber\\
    &=(1-2^{-\ell})^{\lfloor \frac{w}{\ell}\rfloor}\nonumber\\
    &\leq e^{-2^{-\ell}(\frac{w}{\ell}-1)}\nonumber\\
    &=e^{-\frac{n^{\varepsilon/2-o(1)}}{\ell}(1+o(1))}.
    \label{eq:no_all_0_subsubstr}
\end{align}
By \Cref{eq:no_all_0_subsubstr},
    we can apply the union bound over all length-$w$ substrings of $\mathbf{x}$
    and get the following:
The probability that $\mathbf{x}$ has a length-$w$ substring with no $0$-run of length at least $\ell$
    is upper-bounded by
\begin{align*}
ne^{-\frac{n^{\varepsilon/2-o(1)}}{\ell}(1+o(1))}=o(1).
\end{align*}
In other words,
    with probability $1-o(1)$,
    every length-$w$ substring of $\mathbf{x}$
    has a $0$-run of length at least $\ell$.
We can repeat the same argument for $1$-runs. This concludes the proof.
\end{IEEEproof}

From any sequence $\mathbf{x}\in\mathcal{C}_{\varepsilon}$ we can extract a much shorter subsequence,
    denoted as $f(\mathbf{x})$,
    such that
    the $t$-contextual-deletion model on $\mathbf{x}$
    corresponds to
    the $t$-deletion model on $f(\mathbf{x})$.
This observation is formally captured by the following theorem.
\begin{theorem}\label{thm:ctxl_del_subseq}
Fix an arbitrary $C\in(1/2,1)$ and let $\eps\in (0,\min(C,4C-2))$. Let $\ell$, $w$, and $\mathcal{C}_{\varepsilon}$
    be as defined in \Cref{lem:C_eps}.
For any binary sequence $\mathbf{x}$,
    define $f(\mathbf{x})$ to be a subsequence of $\mathbf{x}$
    (including not necessarily consecutive entries of the sequence)
    extracted in the following way:
From left to right,
    for every run $r$ of length at least $k$,
    we put this run and all the following runs into $f(\mathbf{x})$
    until we reach one of the following:
\begin{enumerate}[(I)]
    \item
    \label{crit:f_leqkm1}
        A run with the opposite parity (with respect to $r$)
            whose length is in the range $[\ell,k-1]$.
        In this case we include this opposite-parity run into $f(\mathbf{x})$ as well.
    \item \label{crit:f_geqk}
        A new run of length at least $k$.
        In this case we restart the process with this new run.
    \item \label{crit:f_end}
        The end of $\boldx$.
\end{enumerate}
Then,
    for each $\mathbf{x}\in\mathcal{C}_{\varepsilon}$,
    we have the following properties:
\begin{enumerate}[(1)]
    \item \label{ppty:thm_ctxl_del_subseq_runlength_ub}
        The length of $f(\mathbf{x})$ is at most $\frac{n\log n}{2^k}(2\log n + w
        +k-\ell
        )=n^{2(1-C)+\varepsilon+o(1)}$.
    \item \label{ppty:thm_ctxl_del_subseq_deletion}
        For each $\mathbf{y}$ that is
            obtained from $\mathbf{x}$
            via
         at most $\ell-1$ 
            contextual deletions with threshold $k$,
            $f(\mathbf{y})$ can be obtained from $f(\mathbf{x})$
            via the same number of deletions.
    \item \label{ppty:thm_ctxl_del_subseq_recover}
        Given $\mathbf{y}$, $f(\mathbf{y})$, and $f(\mathbf{x})$,
            we can uniquely recover $\mathbf{x}$.
\end{enumerate}
\end{theorem}
\begin{remark}[A more precise definition of $f(\mathbf{x})$]
\label{rmk:fx_precise}
    We can define $f(\mathbf{x})$ in \Cref{thm:ctxl_del_subseq}
        more formally as follows:
    Let $\mathbf{x}$ be a binary sequence.
    Write $\mathbf{x}=r_1\cdots r_R$,
        where each $r_i$ is a (complete, maximal) run.
    Let $i_1,\ldots,i_K$ be the indices of the runs of length at least $k$,
        where $1\leq i_1<\cdots<i_K\leq R$.
    For each $\tau\in[1,K]$, define
    \begin{align*}
        J_{\tau} \coloneqq \{j\in[i_{\tau}+1,i_{\tau+1}-1]~:~ \ell \leq |r_j| \leq k-1, b(r_j)\neq b(r_{i_{\tau}})\},
    \end{align*}
        where $|r|$ denotes the length of $r$,
        $b(r)$ denotes the parity of $r$ (i.e., the bit constituting this run), and
        $i_{K+1}\coloneqq R+1$.
    Define
    \begin{align*}
        j_{\tau}\coloneqq
        \begin{cases}
            \min(J_{\tau}),\textnormal{ if }J_{\tau} \neq \emptyset,\\
            i_{\tau+1}-1,\textnormal{ otherwise}.
        \end{cases}
    \end{align*}
    Then, $f(\mathbf{x})$ is given by
    \begin{align*}
        f(\mathbf{x})\coloneqq r_{i_1}r_{i_1+1}\dots r_{j_1} r_{i_2}r_{i_2+1}\dots r_{j_2}\dots r_{i_K}r_{i_K+1}\dots r_{j_K}.
    \end{align*}
\end{remark}
\begin{example}
\label{example:fx}
Consider $k=5$, $\ell=3$, and let
\begin{align}
\mathbf{x}=1\underline{000000}100110000\overline{1111}001110\underline{1111111}00\underline{111111}01\overline{000}11,\label{eq:example_fx}
\end{align}
    where the runs of length at least $k$ (corresponding to  $r_{i_{\tau}}$) are underlined,
    and between each pair of such runs $r_{i_{\tau}}$ and $r_{i_{\tau+1}},$
    the first run with the opposite parity from that of  $r_{i_{\tau}}$, and of length between $\ell$ and $k-1$
    (i.e. $r_{j_{\tau}}$ for $J_{\tau}\neq \emptyset$) is overlined.
Then
\begin{align*}
    f(\mathbf{x})=\underline{000000}100110000\overline{1111}\underline{1111111}00\underline{111111}01\overline{000}.
\end{align*}
The parameters in Remark~\ref{rmk:fx_precise}
    can be easily determined, and summarized as follows:
$\mathbf{x}$ consists of $R=17$ runs,
    and $K=3$ of them are of length at least $k=5$.
The parameters/sets $i_s$, $J_s$, and $j_s$ are
\begin{align*}
    i_1 &= 2,       & i_2 &= 11,           & i_3 &= 13,\\
    J_1 &= \{7,9\}, & J_2 &= \emptyset,    & J_3 &= \{16\},\\
    j_1 &= 7,       & j_2 &= 12,           & j_3 &= 16.
\end{align*}
\end{example}

\begin{IEEEproof}[Proof of \Cref{thm:ctxl_del_subseq}]
We first prove Property~\ref{ppty:thm_ctxl_del_subseq_runlength_ub}.
By Property~\ref{ppty_1_Prop_C_eps} of $\mathcal{C}_{\varepsilon}$,
    it suffices to show that each run $r$ of length at least $k$ contributes to at most $(2\log n+w+k-\ell)$ bits in $f(\boldx)$
    in the sense of the definition of $f$ in~\Cref{thm:ctxl_del_subseq}.
Since
    $r$ itself is of length at most $2\log n$ by Property~\ref{ppty_2_Prop_C_eps} of $\mathcal{C}_{\varepsilon}$,
    it remains to show that the run $r$ ``collects'' at most $w+k-\ell$ bits following it.
To clarify, let us once again examine~\Cref{example:fx}.
We follow the definition of $f$ in~\Cref{thm:ctxl_del_subseq} and construct $f(\boldx)$ with $\boldx$ defined in~\eqref{eq:example_fx}:
    The first run we encounter from left to right is the $0$-run of length $6$.
By construction,
    we place that length-$6$ $0$-run and all the following runs into $f(\boldx)$ until the length-$4$ $1$-run,
    and we included this $1$-run into $f(\boldx)$ as well.
This procedure corresponds to the substring $\underline{000000}100110000\overline{1111}$ in $f(\boldx)$.

We say $100110000\overline{1111}$ are the ``follower'' bits
    that the run $\underline{000000}$ collects.
Similarly,
    for the second run of length at least $k$ in $\boldx$,
    which is a $1$-run of length $7$,
    the follower bits the run collects are $00$. Finally, the third run of length at least $k$ in $\boldx$,
    which is a $1$-run of length $6$,
    collects the bits $01\overline{000}$.
    
If
    there are at most $w-1$ bits following $r$ in $\boldx$,
    then $r$ collects at most $w-1\leq w+k-\ell$ bits by Criterion~\ref{crit:f_end}.
Now consider the case where
    there are at least $w$ bits following $r$.
Let
    $r'$ be the first $1$-run of length at least $\ell$ after $r$
    (without loss of generality, assume $r$ is a $0$-run).
Note that
    Property~\ref{ppty_4_Prop_C_eps} of $\mathcal{C}_{\varepsilon}$
    guarantees that $r'$ exists and that
    the first $\ell$ bits of $r'$ lie  within 
    the length-$w$ substring following $r$.
We then split our analysis based on the length of $r'$.
\begin{itemize}
    \item 
        If $|r'|\in[\ell,k-1]$,
    then by Criterion~\ref{crit:f_leqkm1}
    $r$ collects all the bits that follow up to and including $r'$,
    which contributes at most $w+k-\ell$ bits (since at most $|r'|-\ell\leq k-\ell$ bits of $r'$ lie outside of the length-$w$ substring following $r$).
    \item 
        If $|r'|\geq k$,
    then by Criterion~\ref{crit:f_geqk}
    $r$ collects all the bits that follow  but excluding $r'$,
    which contributes at most $w-\ell\leq w+k-\ell$ bits.
\end{itemize}
In all the above cases, $r$ collects at most $w+k-\ell$ ``follower'' bits to be included into $f(\boldx)$. This establishes Property~\ref{ppty:thm_ctxl_del_subseq_runlength_ub} in~\Cref{thm:ctxl_del_subseq}.

Before proving Properties \ref{ppty:thm_ctxl_del_subseq_deletion} and \ref{ppty:thm_ctxl_del_subseq_recover},
    we introduce some auxiliary notation.
Following the definitions in~\Cref{rmk:fx_precise}, for each $\tau\in[1,K],$ let
\begin{align*}
    \boldq^{(\tau)}\coloneqq r_{i_{\tau}}\circ\cdots\circ r_{j_{\tau}}.
\end{align*}
We then have $f(\boldx)=\boldq^{(1)}\circ\cdots\circ \boldq^{(K)}$.
Intuitively, each run $r_{i_{\tau}}$ of length at least $k$ ``contributes'' $\boldq^{(\tau)}$ to the subsequence $f(\boldx)$.
Then,
    define
\begin{align*}
    \mathcal{M}\coloneqq \{\tau\in[1,K]~:~ J_{\tau}\neq\emptyset\},
\end{align*}
    which comprises the indices of the runs $r_{i_{\tau}}$ of length at least $k$
    such that the process of collecting bits following (and including) $r_{i_{\tau}}$ terminates by Criterion~\ref{crit:f_leqkm1}.
Write $M\coloneqq |\mathcal{M}|$
    and order $\mathcal{M}$ as $\mathcal{M}=(\tau_1,\ldots,\tau_M)$,
    where $1\leq \tau_1 < \cdots < \tau_M \leq K$.
Next, for each $m\in[1,M]$ define
\begin{align*}
    \boldu^{(m)}\coloneqq \boldq^{(\tau_{m-1}+1)}\circ\cdots\circ \boldq^{(\tau_m)},
\end{align*}
    where $\tau_0\coloneqq 0$.
In words,
    $\boldu^{(1)}$ represents all the bits $f(\boldx)$ collects until the first time it terminates based on Criterion~\ref{crit:f_leqkm1},
    where the process may have restarted with Criterion~\ref{crit:f_geqk} several times.
Similarly,
    $\boldu^{(2)}$ are all the bits $f(\boldx)$ collects after (but excluding) $\boldu^{(1)}$ until the second time $f(\boldx)$ terminates based on Criterion~\ref{crit:f_leqkm1},
    and so on.
Note that each $\boldu^{(m)}$ is a \emph{substring} of $\boldx$,
    while it is possible that $\boldu^{(m)}$ and $\boldu^{(m+1)}$ are not adjacent.
Additionally,
    we also define
\begin{align*}  \boldu^{(\textnormal{end})}\coloneqq\begin{cases}
        \boldq^{(\tau_M+1)}\circ\cdots\circ \boldq^{(K)},\textnormal{ if }\tau_M < K,\\
    \textnormal{empty string, if }\tau_M=K.
    \end{cases} 
\end{align*}
Note that all the definitions above are also valid even when $M=0$
    (i.e., the process never stops with Criterion~\ref{crit:f_leqkm1},
    and in this case $\boldu^{(\tn{end})}=\boldu^{(1)}$ is simply the substring of $\boldx$ from the first run of length at least $k$ all the way to the end of $\boldx$).
Finally,
    the definitions above allow us to write $\boldx$ and $f(\boldx)$ as
\begin{align}
    \boldx &= \boldw^{(1)}\circ\boldu^{(1)}\circ\boldw^{(2)}\circ\boldu^{(2)}\circ\cdots\circ\boldw^{(M)}\circ\boldu^{(M)}\circ\boldw^{(\tn{end})}\circ\boldu^{(\tn{end})},\label{eq:boldx_rewrite}\\
    f(\boldx) &= \boldu^{(1)}\circ\boldu^{(2)}\circ\cdots\circ\boldu^{(M)}\circ\boldu^{(\tn{end})},\label{eq:fboldx_rewrite}
\end{align}
    where $\boldw^{(1)},\ldots,\boldw^{(M)}$, and $\boldw^{(\tn{end})}$
    are binary strings without runs of length at least $k$ (and which can be possibly empty).
Furthermore,
    each nonempty component in~\eqref{eq:boldx_rewrite}
    contains \emph{complete} runs of $\boldx$.
In other words,
    if $\boldw^{(m)}$ is nonempty,
    then the last bit of $\boldw^{(m)}$ is different from the first bit of $\boldu^{(m)}$,
    and the first bit of $\boldw^{(m)}$ is different from the first bit of $\boldu^{(m-1)}$.
On the other hand,
    if $\boldw^{(m)}$ is empty,
    then the last bit of $\boldu^{(m)}$ is different from the first bit of $\boldu^{(m+1)}$.

We now make the following claims.
\begin{claim}\label{clm:boldy_rewrite}
Suppose $\boldy$ is obtained from $\boldx$ via exactly $t$ contextual deletions with threshold $k$, where $t$ is an arbitrary integer.
Then,
    the contextual deletions can only happen in the $\boldu^{(m)}$ components of $\boldx$. More precisely,
    there exists $M+1$ non-negative integers
    $t_1,\ldots,t_{M},t_{\tn{end}}$ such that
    $t_1+\cdots+t_M+t_{\tn{end}}=t$ and
\begin{align}
    \boldy=\boldw^{(1)}\circ\bolds^{(1)}\circ\boldw^{(2)}\circ\bolds^{(2)}\circ\cdots\circ\boldw^{(M)}\circ\bolds^{(M)}\circ\boldw^{(\tn{end})}\circ\bolds^{(\tn{end})},
    \label{eq:boldy_rewrite}
\end{align}
    where for each $m\in[M]$
    the substring $\bolds^{(m)}$
    is obtained from $\boldu^{(m)}$
    via $t_m$ contextual deletions (with the same threshold $k$),
    and $\bolds^{(\tn{end})}$
    is obtained from $\boldu^{(\tn{end})}$
    via $t_{\tn{end}}$ contextual deletions. 
\end{claim}
\subsubsection*{Proof of Claim \ref{clm:boldy_rewrite}}
By construction,
    all the runs of length at least $k$ in $\boldx$ are contained in $\boldu^{(1)},\ldots,\boldu^{(M)}$ and $\boldu^{(\tn{end})}$.
Furthermore,
    none of these runs of length at least $k$ can be the last run of any $\boldu^{(1)},\ldots,\boldu^{(M)}$.
To see this, note that for each $m\in[M]$,
    the last run in $\boldu^{(m)}$,
    denoted as $r^{(m)}$ (i.e.,  $r^{(m)}=r_{i_{\tau_m}}$),
    satisfies the following properties:
\begin{enumerate}[(i)]
    \item \label{ppty:r_parenthesis_m_1} The length of $r^{(m)}$ is in the range $[\ell,k-1]$,
        since it is the ``stopping pattern'' when collecting $\boldu^{(m)}$ into $f(\boldx)$.
    \item \label{ppty:r_parenthesis_m_2}
        The first encountered run of length at least $k$ when traversing from $r^{(m)}$ to the left,
        denoted as $\widetilde{r}^{(m)}$, has the opposite parity of $r^{(m)}$,
        as otherwise the construction would not stop at $r^{(m)}$.
    \item \label{ppty:r_parenthesis_m_3}
        There are at least two runs between $\widetilde{r}^{(m)}$ and $r^{(m)}$.
        If they were adjacent,
        then $\boldx$ would contain the pattern $0^k1^{\ell}$ or $1^k0^{\ell}$,
        which contradicts Property~\ref{ppty_3_Prop_C_eps} of $\mathcal{C}_{\varepsilon}$.
\end{enumerate}
Even if $\boldu^{(\tn{end})}$ ends with a run of length at least $k$,
    this run is actually the last run of the sequence $\boldx$ and thus cannot contribute to a contextual deletion.
Therefore,
    any possible location for a contextual deletion is within $\boldu^{(m)}$
    for $m\in[M]$
    or within $\boldu^{(\tn{end})}$.
This completes the proof of Claim~\ref{clm:boldy_rewrite}.
\begin{claim}\label{clm:no_create_ell_km1_runs}
Contextual deletions in $\boldx$ cannot lead to runs of length in the range $[\ell,k-1]$.
\end{claim}
\subsubsection*{Proof of Claim \ref{clm:no_create_ell_km1_runs}}
The following are all the possibilities regarding how a contextual deletion can change a runlength in $\boldx$:
\begin{itemize}
    \item
        $0^k1^a0$ becomes $0^k1^{a-1}0,$ for $a\geq 2$.
        By Property \ref{ppty_3_Prop_C_eps} of $\boldx$ we know that $a\leq \ell$.
        Therefore, the newly created runlength $a-1$ is at most $\ell-1$.
    \item 
        $0^k10^{a}1$ becomes $0^{k+a}1$.
        The newly created runlength $k+a$ is at least $k+1$.
\end{itemize}
By the sequential property of contextual deletions described in the proof of Theorem~\Cref{thm:GV},
    we can apply this argument from right to left and
    establish Claim~\ref{clm:no_create_ell_km1_runs}.

\begin{claim}\label{clm:no_adj_k_runs}
Any $\boldx\in\mathcal{C}_{\varepsilon}$ does not have two adjacent runs of length at least $k$.
\end{claim}
\subsubsection*{Proof of Claim \ref{clm:no_adj_k_runs}}
Since $C>1/2$ and $\varepsilon<4C-2$,
    we have $\ell<k$.
Therefore,
    forbidding patters $0^k1^{\ell}$ and $1^k0^{\ell}$ in Property \ref{ppty_3_Prop_C_eps} of $\mathcal{C}_{\varepsilon}$ also implies
    forbidding $0^k1^k$ and $1^k0^k$. Hence, Claim~\ref{clm:no_adj_k_runs} follows.
\begin{claim}\label{clm:fboldy_rewrite}
Assume the same settings as in Claim~\ref{clm:boldy_rewrite}.
We further have
\begin{align}
    f(\boldy)=\bolds^{(1)}\circ\bolds^{(2)}\circ\cdots\circ\bolds^{(M)}\circ\bolds^{(\tn{end})}.
    \label{eq:fboldy_rewrite}
\end{align}
\end{claim}
\subsubsection*{Proof of Claim \ref{clm:fboldy_rewrite}}
First,
    note that each of $\bolds^{(1)},\ldots,\bolds^{(M)},$ and $\bolds^{(\tn{end})}$ starts with a run of length at least $k$.
The reason is that each $\boldu^{(m)}$ and $\bolds^{(\tn{end})}$ starts with a run of length at least $k$,
    and the run right before it cannot be of length at least $k$ by
    Claim \ref{clm:no_adj_k_runs}.
Thus,
    the length of the starting run in $\boldu^{(m)}$ cannot decrease
    (its length
    can possibly increase,
    if one deletes the single-bit run right after it).
In words,
    the substring collection process $f(\boldy)$ will ``initiate'' whenever it encounters the first run in each $\bolds^{(m)}$ or $\bolds^{(\tn{end})}$.
Then,
    it suffices to check the following two conditions:
\begin{enumerate}[(i)]
    \item \label{ppty:last_run} For $m\in[M]$,
        each $\bolds^{(m)}$ ends with a run $\hat{r}^{(m)}$ of length in $[\ell,k-1]$.
        Furthermore,
            when traversing from $\hat{r}^{(m)}$ to the left,
            the first encountered run of length at least $k$ has the opposite parity of  $\hat{r}^{(m)}$.
    \item \label{ppty:middle_runs}
        For
            each run $\hat{r}'$ in $\bolds^{(1)},\ldots,\bolds^{(M)},$ and $\bolds^{(\tn{end})}$ of length in $[\ell,k-1]$
            that is not the last run of $\bolds^{(1)},\ldots,\bolds^{(M)}$,
            when traversing from $\hat{r}'$ to the left,
            the first encountered run of length at least $k$ has the same parity as $\hat{r}'$.
\end{enumerate}
Condition \ref{ppty:middle_runs} ensures that 
    the construction $f(\boldy)$ will not be forced to terminate early, and Condition~\ref{ppty:last_run} guarantees that 
    the construction will stop at the end of $\bolds^{(m)}$ for $m\in[M]$.

We first establish Condition~\ref{ppty:last_run}. We show that for $m\in[M]$,
    each $\bolds^{(m)}$ ends with the same run $r^{(m)}$ as $\boldu^{(m)}$.
First,
    recall that
    the length of $r^{(m)}$ is in  $[\ell,k-1]$,
    as described in Property~\ref{ppty:r_parenthesis_m_1} of $r^{(m)}$.
Without loss of generality, assume that $r^{(m)}$ is a $1$ run, and thus $\widetilde{r}^{(m)}$ is a $0$ run since $r^{(m)}$ and $\widetilde{r}^{(m)}$ have opposite parity. Let the runlength of $r^{(m)}$ and $\widetilde{r}^{(m)}$ be $\ell_m$ and $\widetilde{\ell}_m$, respectively. Then,
    the last few (complete) runs in $\boldu^{(m)}$
    can be summarized as
\begin{align}
    0^{\widetilde{\ell}_m}1^{a_1}0^{a_2}\cdots 1^{a_{\eta-1}}0^{a_{\eta}}1^{\ell_m},
    \label{eq:last_few_runs}
\end{align}
    for some positive even integer $\eta$
    and positive integers $a_1,\ldots,a_{\eta}$ satisfying:
\begin{itemize}
    \item $a_1,a_3,\ldots,a_{\eta-1}\leq \ell-1$ (or otherwise the construction of $f(\boldx)$ will stop before reaching $r^{(m)}$);
    \item $a_2,a_4,\ldots,a_{\eta}\leq k-1$ (since $\widetilde{r}^{(m)}$ is the first encountered run of length at least $k$ when starting to traverse from $r^{(m)}$ to the left).
\end{itemize}
It follows that,
    even after one contextual deletion induced by $\widetilde{r}^{(m)}$,
    the substring in \eqref{eq:last_few_runs}
    becomes either
\begin{align*}
    0^{\widetilde{\ell}_m}1^{a_1-1}0^{a_2}\cdots 1^{a_{\eta-1}}0^{a_{\eta}}1^{\ell_m}, \textnormal{ if }a_1\geq 2,
\end{align*}
    or
\begin{align*}
    0^{\widetilde{\ell}_m+a_2}\cdots 1^{a_{\eta-1}}0^{a_{\eta}}1^{\ell_m}, \textnormal{ if }a_1= 1.
\end{align*}
In either case, the last run in~\eqref{eq:last_few_runs},
    i.e. $r^{(m)}$,
    is still of length $\ell_m$.
Finally,
    by Claim \ref{clm:no_adj_k_runs},
    the length of $\widetilde{r}^{(m)}$
    cannot decrease
    (its length can possibly increase if it merges with other runs).
These arguments prove Condition~\ref{ppty:last_run}.

We now establish Condition~\ref{ppty:middle_runs}.
By Claim~\ref{clm:no_create_ell_km1_runs},
    $\hat{r}'$ 
    must already be included in $\boldu^{(m)}$. Furthermore,
    its previous run of length at least $k$,
    denoted by $\widetilde{\hat{r}}'$,
    has to have the same parity as $\hat{r}'$,
    or otherwise the construction of $f(\boldx)$ will terminate early at $\hat{r}'$. For similar reasons,
    the length of $\widetilde{\hat{r}}'$ cannot decrease. These arguments prove Condition~\ref{ppty:middle_runs}
    and conclude the proof of Claim~\ref{clm:fboldy_rewrite}.

Note that Claims~\ref{clm:boldy_rewrite} and~\ref{clm:fboldy_rewrite} imply Property~\ref{ppty:thm_ctxl_del_subseq_deletion} in~\Cref{thm:ctxl_del_subseq}.

Now we prove Property~\ref{ppty:thm_ctxl_del_subseq_recover} of~\Cref{thm:ctxl_del_subseq}, i.e., we show how to recover $\boldx$ from $\boldy$ in~\eqref{eq:boldy_rewrite}, $f(\boldy)$ in~\eqref{eq:fboldy_rewrite}, and $f(\boldx)$
    when there are at most $\ell$ contextual deletions.
Note that
    by the proof of Claim~\ref{clm:fboldy_rewrite},
    we know that
    the parameter $M$ derived from $\boldy$ is the same as that derived from $\boldx$.
Furthermore,
    $\bolds^{(\tn{end})}$ is empty if and only if $\boldu^{(\tn{end})}$ is empty.
Therefore,
    we can deduce that $f(\boldx)$ and $\boldx$ must take the form in~\eqref{eq:fboldx_rewrite} and~\eqref{eq:boldx_rewrite},
    respectively.
It remains to determine $t_1,\ldots,t_{M},$ and $t_{\tn{end}}$.

We first determine $t_1$.
By the proof of Claim \ref{clm:fboldy_rewrite},
    we know that
    $\bolds^{(1)}$ ends with the same run as
    $\boldu^{(1)}$.
That is,
    if $\bolds^{(1)}$ ends with a (complete) run of length $l_1\in[\ell,k-1]$,
    then $\boldu^{(1)}$ must end with a run having the same parity and the same length $l_1$.
Without loss of generality
    assume $\bolds^{(1)}$ is a $1$ run.
Then,
    since $\bolds^{(1)}$ is obtained from $\boldu^{(1)}$ via $t_1<\ell$ contextual deletions,
    we can determine $t_1$ by 
    examining the last runlength in the first $\left|\bolds^{(1)}\right|$ bits of $f(\boldx)$. More precisely,
    the first $\left|\bolds^{(1)}\right|$ bits of $f(\boldx)$
    will be the length-$\left|\bolds^{(1)}\right|$ prefix of $\boldu^{(1)}$,
    which must end with a $1$ run of length $l_1-t_1>0$.
This procedure uniquely determines $t_1$,
    and consequently, $\boldu^{(1)}$ is uniquely determined by the first $(\left|\bolds^{(1)}\right|+t_1)$ bits of $f(\boldx)$
    (i.e. $\boldu^{(1)}=\bolds^{(1)}\circ 1^{t_1}$).

Next, $t_2$ can be determined in a similar manner:
Compare the next $\left|\bolds^{(2)}\right|$ bits in $f(\boldx)$ with $\bolds^{(2)}$; then,  $t_2$ equals the difference between the length of the last run in each substring.
We can continue with this procedure and determine $t_2,\ldots,t_M,t_{\tn{end}}$ and $\boldu^{(2)},\ldots,\boldu^{(M)},\boldu^{(\tn{end})}$.
Finally,
    we can recover $\boldx$ by placing  $\boldu^{(1)},\ldots,\boldu^{(M)},$ and $\boldu^{(\tn{end})}$ into~\eqref{eq:boldx_rewrite}.
\end{IEEEproof}

\begin{example}
Consider $k=5$, $\ell=3$, and
\begin{align}
\boldx=1\underline{000000}10000\overline{1111}010\underline{11111}00\underline{11111}01\overline{000}1.
\label{eq:ex_x}
\end{align}
Similar to \Cref{example:fx},
    each run $r$ of length at least $k$ in $\boldx$ defined in \eqref{eq:ex_x}
    is underlined.
In addition,
    after each such $r$,
    if a run $r'$ of the opposite parity (with respect to $r$) and of length between $\ell$ and $k-1$
    occurs before the next occurrence of a run of length at least $k$,
    we overline $r'$.
Then, according to the definition of $f$ in \Cref{thm:ctxl_del_subseq},
    the subsequence $f(\boldx)$ is given by
\begin{align}
    f(\mathbf{x})=\underline{000000}10000\overline{1111}\underline{11111}00\underline{11111}01\overline{000}.
    \label{eq:ex_fx}
\end{align}
Suppose two contextual deletions with threshold $k=5$ occur in $\boldx$,
    one after the second run in $\boldx$ and the other after the ninth run in  $\boldx$.
This leads to the output
\begin{align}
    \boldy=1\underline{0000000000}\overline{1111}010\underline{11111}0\underline{11111}01\overline{000}1,
    \label{eq:ex_y}
\end{align}
    where we underlined and overlined the runs in $\boldy$ following the same rule.
Then, 
    by the definition of $f$
    again,
    we can calculate
\begin{align}
    f(\mathbf{y})=\underline{0000000000}\overline{1111}\underline{11111}0\underline{11111}01\overline{000},
    \label{eq:ex_fy}
\end{align}
    which can be obtained from $f(\boldx)$ via two deletions.

We now show how to recover $\boldx$ in \eqref{eq:ex_x} from $\boldy$ in \eqref{eq:ex_y}, $f(\boldy)$ in \eqref{eq:ex_fy}, and $f(\boldx)=0000001000011111111101111101000$
    (note that we now clearly cannot use the underlined and overlined form of $f(\boldx)$ as in \eqref{eq:ex_fx}).
First,
    following the notation in \eqref{eq:boldy_rewrite},
    we
    write $\boldy=1\circ\mathbf{s}^{(1)}\circ(010)\circ\mathbf{s}^{(2)}\circ 1$,
    where
    $\mathbf{s}^{(1)}\coloneqq \underline{0000000000}\overline{1111}$
    and 
    $\mathbf{s}^{(2)}\coloneqq \underline{11111}0\underline{11111}01\overline{000}$.
In words,
    $\mathbf{s}^{(1)}$ and $\mathbf{s}^{(2)}$ 
    are the substrings of $\boldy$ that were included into $f(\boldy)$ based on Criterion~\ref{crit:f_leqkm1}.
In addition,
    $\mathbf{s}^{(\tn{end})}$ is empty, 
    since the construction of $f(\boldy)$ was not required to stop by Criterion~\ref{crit:f_end}.
It follows that $f(\boldy)=\mathbf{s}^{(1)}\circ\mathbf{s}^{(2)}$.
In addition,
    by the derivation of Claim \ref{clm:fboldy_rewrite} in the proof of \Cref{thm:ctxl_del_subseq},
    we know that
    $\boldx$ must take the form
\begin{align}
    \boldx=1\circ\mathbf{u}^{(1)}\circ(010)\circ\mathbf{u}^{(2)}\circ 1,
    \label{eq:ex_x_unknown}
\end{align}
    where $\mathbf{s}^{(i)}$ is obtained from $\mathbf{u}^{(i)}$ via $t_i$ (contextual) deletions
    for some non-negative integers $t_1$ and $t_2$
    such that $t_1+t_2=2$.
In particular,
    we also have $f(\boldx)=\mathbf{u}^{(1)}\circ \mathbf{u}^{(2)}$.

Furthermore,
    since $\mathbf{s}^{(1)}$ ends with $01111$,
    by the derivation of Claim \ref{clm:fboldy_rewrite} again,
    we know that
    that $\mathbf{u}^{(1)}$ must end with $01111$ as well.
It follows that we can determine $t_1$ by
examining
the first $|\bolds^{(1)}|$ bits in
$f(\boldx)$:
\begin{align}
    f(\mathbf{x})&=\underbrace{00000010000111}_{
    \substack{
    \textnormal{the first } \left(|\mathbf{u}^{(1)}|-t_1\right)\\ \textnormal{ bits of }\mathbf{u}^{(1)}
    }
    }
    11111101111101000.
    \label{eq:ex_fx_aligned}
\end{align}
Since the ``underbraced'' part in~\eqref{eq:ex_fx_aligned} ends with $0111$
    while $\mathbf{u}^{(1)}$ ends with $01111$,
    we can deduce that $t_1=1$
    and $\mathbf{u}^{(1)}=000000100001111$.
Proceeding similarly, we can recover $\mathbf{u}^{(2)}=11111001111101000$ from $\mathbf{s}^{(2)}$. By substituting $\mathbf{u}^{(1)}$ and $\mathbf{u}^{(2)}$ into~\eqref{eq:ex_x_unknown},
    we can fully reconstruct $\boldx$ in~\eqref{eq:ex_x}.
\end{example}

\Cref{thm:ctxl_del_subseq}
    motivates defining the following ``VT-like'' single-contextual-deletion-correcting code.
    For an arbitrary integer $a\in[0,\frac{n\log n}{2^k}(2\log n + w+k-\ell)]$,
    we let
    \begin{align}
        \mathcal{VT}_{a,\varepsilon}
        \coloneqq
        \left\{\mathbf{x}\in\mathcal{C}_{\varepsilon}
        ~:~
        \sum_{i=1}^{|f(\mathbf{x})|}i\cdot f(\mathbf{x})_i \equiv a \bmod \left(\frac{n\log n}{2^k}(2\log n + w+k-\ell)+1\right)
        \right\}.
    \end{align}
    One can view $\VT_{a,\eps}$ as, essentially, a standard VT code applied to $f(\mathbf{x})$ (with some additional structural assumptions on $\boldx$).
    The next result states that this code can correct a single contextual deletion with threshold $k$, and gives a bound on the redundancy of the largest such code.
    We postpone the analysis of efficient encoding and decoding procedures for a variant of this code until later.
\begin{corollary}\label{corr:VT_single_ctxl_del}
Fix an arbitrary $C\in(1/2,1)$ and $\eps\in(0,4C-2)$.
Then, $\mathcal{VT}_{a,\varepsilon}$ is a $(t=1,k= C\log n )$-contextual deletion-correcting code.
Furthermore, there is a choice of $a$ such that $\mathcal{VT}_{a,\varepsilon}$ has redundancy at most $(2(1-C)+\eps)\log n+o(\log n)$.
\end{corollary}

\begin{IEEEproof}
Combining the result of \Cref{thm:ctxl_del_subseq}
    with the fact that the standard VT code can correct a single deletion~\cite{sloane2002single},
    each $\mathcal{VT}_{a,\varepsilon}$ is uniquely decodable
    under a single contextual deletion with threshold $k$.
Then,
    note that by \Cref{lem:C_eps},
    \begin{align*}
        \sum_{a=0}^{\frac{n\log n}{2^k}(2\log n + w+k-\ell)} |\mathcal{VT}_{a,\varepsilon}|=|\mathcal{C}_{\varepsilon}| = 2^n(1-o(1)).
    \end{align*}
Therefore,
    there exists some $a^*\in[0,\frac{n\log n}{2^k}(2\log n + w+k-\ell)]$ such that
\begin{align*}
    |\mathcal{VT}_{a^*,\varepsilon}| \geq \frac{2^n(1-o(1))}{\frac{n\log n}{2^k}(2\log n + w+k-\ell)+1} = 2^{n-(2-2C+\varepsilon)\log n +o(\log n)},
\end{align*}
and the redundancy of $\mathcal{VT}_{a^*,\varepsilon}$ is $(2-2C+\varepsilon)\log n +o(\log n)$.
\end{IEEEproof}

\subsection{Correcting a single contextual deletion with efficient encoding and decoding}
\label{subsec:1contexul_del_efficient}

We show next that a variant of the codes from \Cref{corr:VT_single_ctxl_del} supports encoding and decoding procedures running in time $\poly(n)$, while only requiring $O(\log\log n)$ bits of redundancy.
Here, it is convenient to define the code directly through its encoding and decoding procedures.

\subsubsection{Efficient encoding and decoding of $\cC_\eps$}

The encoding and decoding procedures for our code proceed through several steps. The first step in the encoding procedure is to map messages into structured strings from  $\cC_\eps$.
The last step in the decoding procedure is to map strings in $\cC_\eps$ back to messages.
We show that this can be done efficiently.
More precisely, we have the following result.

\begin{lemma}\label{lem:struct}
    There exist injective encoding and decoding maps $\Encstruct:\bits^{n-1}\to\cC_\eps$ and $\Decstruct:\cC_\eps\to\bits^{n-1}$ that are computable in time $\poly(n)$ and satisfy $\Decstruct(\Encstruct(\boldx))=\boldx$ for any $\boldx\in\bits^{n-1}$.
\end{lemma}

We defer the proof of \Cref{lem:struct} to Appendix~\ref{app:DFA}.
The intuition is that the properties defining $\cC_\eps$ can be captured by a deterministic finite automaton (DFA) whose description can be obtained in $\poly(n)$ time.
Then, we can apply known results regarding ``ranking'' and ``unranking'' the set of strings accepted by a DFA to obtain the required encoding and decoding maps running in time $\poly(n)$. Although sometimes used in constrained coding~\cite{ryzhikov2020synchronizing}, we believe this to be the first application of DFA-based methods in the area of deletion error-correction.

\subsubsection{Efficient encoding}

Given a message $\boldx\in\bits^{n-1}$, the encoding function $\Enc(\boldx)$ entails the following:
\begin{itemize}
    \item Compute $\overline{\boldx}=\Encstruct(\boldx)\in\bits^n$.

    \item 
    Compute $h=h_{\textnormal{VT}}(f(\overline{\boldx}))$,
    where for any binary sequence $\mathbf{w}=(w_1,\ldots,w_m)$ we use $h_{\textnormal{VT}}(\mathbf{w})$ to denote the VT syndrome (i.e. $h_{\textnormal{VT}}(\mathbf{w})\coloneqq \sum_{i=1}^m iw_i\bmod m+1$).
    Then,
        represent $h$ as a bit string of length
    \begin{equation*}
        |h|=\left\lceil\log\left(\frac{n\log n}{2^k}(2\log n+w+k-\ell)+1\right)\right\rceil.
    \end{equation*}

    \item Let $b=1-\overline{\boldx}_n$.
    Also, let $(\Encsmall,\Decsmall)$ be the encoding and decoding functions of a binary single deletion-correcting code for messages of length $m=|h|+1$, with redundancy $O(\log m)$. We know many such codes for which the encoding and decoding functions run in time $\poly(m)$ (e.g., the VT code~\cite{Lev65} / systematic VT code~\cite{Abdel98}).
    Then, set $\Enc(\boldx)=\overline{\boldx}\circ bb(1-b)\circ\Encsmall(h)$, where $\circ$ as before denotes string concatenation.
\end{itemize}

It is clear that this encoding procedure runs in $\poly(n)$ time.
Furthermore, going from $\boldx$ to $\overline{\boldx}$ introduces $1$ bit of redundancy, appending $bb(1-b)$ adds $3$ bits of redundancy, and appending $\Encsmall(h)$ adds $|h|+O(\log |h|)$ bits of redundancy by the definition of $\Encsmall$. This leads to a total of
\begin{equation*}
    1 + 3 + |h| + O(\log |h|) = (2(1-C)+\eps)\log n + o(\log n)
\end{equation*}
bits of redundancy.
It remains to see that we can correct one contextual deletion with threshold $k$, which we do next.

\subsubsection{Efficient decoding after a single contextual deletion}
Suppose that we receive $\boldy$ obtained from $\Enc(\boldx)$ via at most one contextual deletion with threshold $k>2$ (recall that this holds for all large enough values of $n$, since $k=C\log n$.
If $|\boldy|=|\Enc(\boldx)|$ then no error was introduced, and so we can easily recover $\boldx$ by computing $\Decstruct(\boldy_1,\dots,\boldy_n)=\Decstruct(\overline{\boldx})=\boldx$ in time $\poly(n)$ by \Cref{lem:struct}.
Therefore, we now assume that one contextual deletion with threshold $k$ has occurred, giving rise to $\boldy$.
We then proceed as follows:
\begin{itemize}
    \item Denote $n'=|\Encsmall(h)|$. Take $\boldy'$ to be the last $n'-1$ bits of $\boldy$.
    Then, compute $\Decsmall(\boldy')$, which equals $h$ since $\boldy'$ is obtained from $\Encsmall(h)$ via at most $1$ deletion.
    \item Use the structure of $\Enc(\boldx)$ to find the prefix $\mathbf{p}$ of $\boldy$ containing exactly those bits coming from $\overline{\boldx}$. 
    There are two cases to consider:
    \begin{itemize}
        \item The run to which $\boldy_{n+1}$ belongs has length at most $2$. Then, this means that $\overline{\boldx}$ ended in a run that was not completely deleted. In this case, we take $\mathbf{p}$ to be the prefix of $\boldy$ up to and excluding the run to which $\boldy_{n+1}$ belongs.

        \item The run to which $\boldy_{n+1}$ belongs has length at least $3$. Then, this means that $\overline{\boldx}$ ended in a run of length $1$ that was deleted, and so the run to which $\boldy_{n+1}$ originally belonged experienced no deletions. Therefore, we take $\mathbf{p}$ to be the prefix of $\boldy$ up to and excluding the last $2$ bits of the run to which $\boldy_{n+1}$ belongs.
    \end{itemize}

    \item Given the prefix $\mathbf{p}$ from the previous step, consider the up to $n+1$ possibilities of adding back the contextual deletion into $\mathbf{p}$ (note that the bit value of the contextual deletion is completely determined by the bit value of the preceding run). Denote by $\mathbf{p}^{(i)}$ the string obtained by inserting the appropriate bit to the left of $\mathbf{p}_i$.  By the analysis from \Cref{corr:VT_single_ctxl_del} and   
    \Cref{thm:ctxl_del_subseq}, we know that there exists a unique $i^\star$ such that $\mathbf{p}^{(i^\star)}\in\cC_\eps$ and $f(\mathbf{p}^{(i^\star)})=h$, and for the unique $\mathbf{p}^{(i^\star)}$ that satisfies this we must have $\mathbf{p}^{(i^\star)}=\overline{\boldx}$.
    Therefore, we can recover $\boldx=\Decstruct(\mathbf{p}^{(i^\star)})$.
\end{itemize}
It is not hard to see that this procedure takes $\poly(n)$ time, since both $\Decsmall(\boldy')$ and $\Decstruct(\mathbf{p}^{(i^\star)})$ run in time $\poly(n)$ and because we can check whether a string $\mathbf{p}^{(i)}\in\cC_\eps$ and $f(\mathbf{p}^{(i)})=h$ in $\poly(n)$ time.

To arrive at the exact statement in~\Cref{thm:eff-codes} 
for $t=1$,
    we can eliminate the $o(\log n)$ term as follows:
First,
    let $\varepsilon\in(0,8C-4)$ be given.
By replacing the role of $\varepsilon$ with $\varepsilon/2$,
    we know that there exists
    an efficient single-contextual-deletion code
    with redundancy at most
    $(2(1-C)+\varepsilon/2+o(1))\log n$.
Then,
    we choose $n$ sufficiently large
    so that the $o(1)$ term is below $\varepsilon/2$.
This leads to a $(1,C\log n)$-contextual deletion-correcting code
    whose redundancy is at most
    $(2(1-C)+\varepsilon)\log n$ for $n$ large enough.

\subsection{Correcting two contextual deletions}
\label{subsec:2contexul_del_efficient}

\Cref{thm:ctxl_del_subseq} implies that
    \emph{any number of contextual deletions} in $\mathbf{x}$
    corresponds to
    the same number of deletions in $f(\mathbf{x})$.
Thus,
    for existing two-deletion-correcting codes \cite{gabrys2018codes,sima2019two,guruswami2021explicit},
    if we can guarantee that
    $f(\mathbf{x})$ satisfies appropriate constraints, we can apply those codes on $f(\mathbf{x})$
    and get two-contextual-deletion-correcting codes.
In the following we choose the code from~\cite{guruswami2021explicit} as the building block for our two-contextual-deletion-correcting code.

The two-deletion-correcting code in~\cite{guruswami2021explicit}
    has redundancy $4\log n + o(\log n)$.
Furthermore,
    the code can be decoded by knowing the value of an efficiently computable hash function.
However,
    it requires the constituent length-$n$ binary sequences
    to satisfy a certain regularity property.
The property of this two-deletion code
    is summarized in the following definition and lemma.
\begin{definition}[{\cite[Definition 5.6]{guruswami2021explicit}}]\label{def:regular_2del_code}
Let $d$ be an absolute constant.
A binary sequence $\mathbf{x}\in\{0,1\}^n$ is said to be $d$-\emph{regular} if 
    every length-$(d\log n)$ substring of $\mathbf{x}$
    contains both a $00$ and a $11$ substring.
\end{definition}
\begin{lemma}[{\cite[Theorem 5.9]{guruswami2021explicit}}]\label{lem:2del_code}
Fix any $d\geq 7$.
There exists a code that can protect every $d$-regular sequence $\mathbf{x}\in\{0,1\}^n$ (as defined in \Cref{def:regular_2del_code})
    against two deletions while introducing $4\log n + o(\log n)$ redundant bits.
More precisely,
    there exists an efficiently computable hash function $\mathrm{hash}_2:\{0,1\}^n\rightarrow\{0,1\}^{4\log n + o(\log n)}$
    with the following property:
    For any $d$-regular sequence $\mathbf{x}\in\{0,1\}^n$,
    knowing $\mathrm{hash}_2(\mathbf{x})$ and a corrupted version of $\mathbf{x}$ after two deletions allows for 
    unique recovery of $\mathbf{x}$.
\end{lemma}

We hence need to add more constraints to $\mathcal{C}_{\varepsilon}$
    to ensure that $f(\mathbf{x})$ satisfies the regularity property in the sense of~\cite{guruswami2021explicit},
    which leads to a two-contextual-deletion-correcting code with redundancy $(8(1-C)+4\varepsilon)\log n + o(\log n)$.
We first show that for a uniformly random length-$n$ binary sequence, every length-$d'$ substring contains both a $00$ and a $11$ string 
    with probability $1-o(1)$, for suitable choices of the parameter $d'$. By [1, Lemma 5.11],
    the probability that 
    a random, length-$m$ binary sequence contains no $00$ or $11$ is at most $(1.62/2)^m=0.81^m=2^{\log(0.81)m}$.
Then, 
    by the union bound,
    for any $d'$ such that
    $d'\log 0.81 + 1 < 0$
    (i.e. $d'>3.29$),
    every $d'\log n$ window in a random, length-$n$ binary sequence contains both a $00$ and a $11$
    with probability at least $1-n^{1+\log(0.81)d'}=1-o(1)$.

We now show that regularity of $\mathbf{x}$
    implies regularity of $f(\mathbf{x})$
    (albeit with different parameters),
    as characterized by the following lemma.

\begin{lemma}\label{lem:window_x_fx}
Let $\mathbf{x}\in\mathcal{C}_{\varepsilon}$
    and let $W$ be a positive number.
Assume that every length-$(W/2)$ window of $\mathbf{x}$ contains both a $00$ and $11$.
Then, 
    every length-$W$ window of $f(\mathbf{x})$ also contains both a $00$ and a $11$.
\end{lemma}
\begin{IEEEproof}
By construction,
    $f(\mathbf{x})$
    consists of nonadjacent substrings of $\mathbf{x}$, say, $f(\mathbf{x})=s_1s_2\cdots s_S$
    for some nonadjacent substrings $s_1,\ldots,s_S$ of $\mathbf{x}$.
Note that
    for $i\in[1,S-1]$,
    by construction, 
    $s_i$ contains a $00$ and a $11$,
    since $s_i$ ends with a run of length at least $\ell$
    and starts (or was restarted) with an opposite-parity run of length at least $k$.

Now, let $\mathbf{w}$ be any length-$W$ window of $f(\mathbf{x})$.
We then split our analysis based on the number of substrings
    that are included in $\mathbf{w}$:
\begin{itemize}
    \item
        If $\mathbf{w}$ includes at least three substrings,
            then it completely contains an $s_i$ for some $i\in[2,S-1]$.
        Since $s_i$ contains both a $00$ and $11$,
            so does $\mathbf{w}$.
    \item 
        If $\mathbf{w}$ includes exactly two substrings,
            one end of it
            corresponds to a window of $\mathbf{x}$ of length at least $W/2$,
            which contains both a $00$ and $11$ by assumption.
    \item
        Lastly,
            if $\mathbf{w}$ lies withing a single substring,
            then it is already a window of length $W>W/2$ of $\mathbf{x}$,
            which contains both a $00$ and $11$ as well.
\end{itemize}
Therefore,
    any length-$W$ window of $f(\mathbf{x})$ contains both a $00$ and a $11$ pattern.
\end{IEEEproof}

We can combine these arguments
    to arrive at the following theorem.
\begin{theorem}\label{thm:two_ctxl_del}
For $C>1/2$ and $\varepsilon>0$ small enough, there exists a two-contextual-deletion-correcting code
    with redundancy
    $(8(1-C)+4\varepsilon+o(1))\log n$.
\end{theorem}
\begin{IEEEproof}
Select $d\geq 7$ such that
    $d(1-C+\varepsilon/2+o(1)) > 3.29$
    (say $d=\max(7,\frac{4}{1-C+\varepsilon/2})$).
Then
    let $\mathcal{C}'_{\varepsilon}$
    be the intersection of $\mathcal{C}_{\varepsilon}$
    and the set of all the sequences
    where
    every $(d(1-C+\varepsilon/2+o(1))\log n)$-window
    has both a $00$ and a $11$.
Note that
    $\mathcal{C}'_{\varepsilon}$ is still of size $2^n(1-o(1))$ since
    $d(1-C+\varepsilon/2+o(1))=4+o(1)>3.29$.
It follows from \Cref{lem:window_x_fx} that
    every length-$(d(2-2C+\varepsilon+o(1))\log n)$ substring of $f(\mathbf{x})$
    contains both a $11$ and $00$.
That is,
    every length-$(d\log n^{2-2C+\varepsilon+o(1)})$ window of $f(\mathbf{x})$
    contains both a $11$ and $00$.
Since $d\geq 7$ and the length of $f(\mathbf{x})$ is at most $n^{2-2C+\varepsilon+o(1)}$,
    the regularity requirement from \cite{guruswami2021explicit}
    is satisfied,
    and thus the two-deletion-correcting code in~\cite{guruswami2021explicit}
    can be applied to $f(\mathbf{x}),$
    for $\mathbf{x}\in\mathcal{C}'_{\varepsilon}$.
This procedure results in a two-contextual-deletion-correcting code
    with redundancy
    $4\log n^{2-2C+\varepsilon+o(1)} + o(\log n)=(8(1-C)+4\varepsilon+o(1))\log n$.
\end{IEEEproof}
Efficient encoding/decoding is also possible for the previously-described two contextual-deletion-correcting codes by adapting the techniques used for single contextual-deletion-correcting codes from Section~\ref{subsec:1contexul_del_efficient}.
It can be seen that
    we only need to check the following two conditions:
\begin{enumerate}
    \item
        We can build a DFA with $\poly(n)$ state space that checks whether a length-$n$ binary sequence belongs to $\mathcal{C}_{\varepsilon}'$ or not,
            where $\mathcal{C}_{\varepsilon}'$ is defined in the proof of \Cref{thm:two_ctxl_del}.
    \item 
        We can compute/employ the hash value $\mathrm{hash}_2(f(\mathrm{unrank}(\mathbf{x})))$
            for any $\mathbf{x}\in\{0,1\}^{n-1}$,
            where $\mathrm{unrank}$ is the unranking function of the DFA described in the first condition (see Appendix~\ref{app:DFA} for the terminology and technical details).
\end{enumerate}
The first condition can be easily satisfied by
    considering the DFA from Section~\ref{subsec:1contexul_del_efficient}, but 
    with two more registers $q_7$ and $q_8$
    that keep track of the last occurrences of $00$ and $11$ (capped at $d(1-C+\varepsilon/2+o(1))\log n$), respectively.
The state space of this revised DFA is still of size poly$(n)$. The unranking function $\mathrm{unrank}$ of this DFA efficiently and uniquely maps any sequence $\mathbf{x}\in\{0,1\}^{n-1}$ to a sequence in $\mathcal{C}_{\varepsilon}'$.

To meet the second condition,
    one possible approach is to modify the construction in \Cref{subsec:1contexul_del_efficient}
    (that is, protect $\mathrm{hash}_2(f(\mathrm{unrank}(\mathbf{x})))$ with some known deletion-correcting  code). Here we demonstrate another approach instead.
Consider an arbitrary input $\mathbf{x}\in\{0,1\}^{n-1}$.
Write 
    $\overline{\mathbf{x}}\coloneqq \mathrm{unrank}(\mathbf{x})$
    and 
    $h_2\coloneqq \mathrm{hash}_2(f(\overline{\mathbf{x}}))$
    for simplicity.
Also write $\overline{\mathbf{x}}=(\overline{x}_1,\ldots,\overline{x}_n)$.
Then, the encoding of $\mathbf{x}$ is of the form
\begin{align}
    \Enc_2(\mathbf{x})\coloneqq E(h_2)\circ (1-\overline{x}_1) \circ \overline{\mathbf{x}},
    \label{eq:enc_two_ctxl_del}
\end{align}
    where $E$ denotes the runlength-limited encoder given in \Cref{thm:rll}.
Since $h_2$ is of length $(8(1-C)+\varepsilon+o(1))\log n$,
    this encoding adds
    $|h_2|+2=(8(1-C)+4\varepsilon+o(1))\log n$ bits of redundancy.
Furthermore,
    this encoding can be computed efficiently
    since $E$ runs in $O(|h_2|)=\poly(\log n)$ time
    and the efficiency of computing $\overline{\mathbf{x}}$ has already been addressed in Section~\ref{subsec:1contexul_del_efficient}.

We claim next that
    any contextual deletion that appears in $ \Enc_2(\mathbf{x})$ must appear in the $\overline{\mathbf{x}}$ part. More precisely,
    if $\mathbf{y}$ is the output of $\Enc_2(\mathbf{x})$
    after up to two contextual deletions,
    then $\mathbf{y}$ must take the form
\begin{align}
    \mathbf{y} = E(h_2)\circ (1-\overline{x}_1) \circ \overline{\mathbf{w}},
    \label{eq:y_two_ctxl_del}
\end{align}
    where $\overline{\mathbf{w}}$ is the output of $\overline{\mathbf{x}}$
    after up to two contextual deletions.
Note that by \Cref{thm:rll},
    the longest runlength in $E(h_2)$ is of length at most $\lceil\log(|h_2|)\rceil+3=O(\log\log n)$,
    which is much smaller than $k-1=(C+o(1))\log n$.
Therefore,
    the prefix
    $E(h_2)\circ 1-\overline{x}_1$ of $ \Enc_2(\mathbf{x})$
    contains no run of length at least $k$.
At the same time,
    it is clear that there cannot be a run that contains $(1-\overline{x}_1,\overline{x}_1)$.
These arguments prove the claim.
    
The decoding process after receiving $\mathbf{y}$ from \Cref{eq:y_two_ctxl_del} is straightforward:
First,
    by \Cref{thm:rll}, we can efficiently recover $h_2$ from $E(h_2)$.
Then,
    since $\overline{\mathbf{x}}$ is in $\mathcal{C}_{\varepsilon}'\subseteq \mathcal{C}_{\varepsilon}$,
    by
    \Cref{thm:ctxl_del_subseq}
    we know that
    $f(\overline{\mathbf{w}})$
    can be obtained from $f(\overline{\mathbf{x}})$
    via up to two deletions.
Next,
    since $\overline{\mathbf{x}}$ is in $\mathcal{C}_{\varepsilon}'$,
    every 
    length-$(d(2-2C+\varepsilon+o(1))\log n)$ substring of $f(\mathbf{x})$
    contains both a $11$ and $00$,
    and thus by \Cref{lem:2del_code} and \Cref{lem:window_x_fx}, 
    we can efficiently recover $f(\overline{\mathbf{x}})$ based on $f(\overline{\mathbf{w}})$ and $h_2$.
Then,
    by
    the proof of
    Property~\ref{ppty:thm_ctxl_del_subseq_recover} of \Cref{thm:ctxl_del_subseq},
    we can efficiently recover $\overline{\mathbf{x}}$ from  $\overline{\mathbf{w}}$, $f(\overline{\mathbf{w}})$, and 
    $f(\overline{\mathbf{x}})$.
Finally, based on the explanation from Appendix~\ref{app:DFA},
    the ranking function $\mathrm{rank}$
    of the associated DFA can efficiently recover  $\mathbf{x}$ from $\overline{\mathbf{x}}$.

Similarly to the arguments described at the end of \Cref{subsec:1contexul_del_efficient},
    we can dispense of  the $o(\log n)$ term as follows:
For any given $\varepsilon\in(0,32C-16)$,
    by applying the code in this subsection with $\varepsilon$ replaced with $\varepsilon/8$,
    we get an efficient $(2,C\log n)$-contextual deletion-correcting code with redundancy at most
    $(8(1-C)+\varepsilon/2+o(1))\log n$,
    which is at most $(8(1-C)+\varepsilon)\log n,$
    for $n$ large enough, and such that
    the $o(1)$ term is smaller than $\varepsilon/2$.
This completes the proof of the $t=2$ part of~\Cref{thm:eff-codes}.

\subsection{Efficient codes correcting any constant number of deletions with logarithmic threshold}
\label{subsec:t_contexul_del_efficient}

In this subsection,
    we prove the $t\geq 3$ part in~\Cref{thm:eff-codes}.
More precisely,
    we show the following.
\begin{theorem}\label{thm:eff-codes-gen-refined}
Let $t\geq 1$ be a constant
    and $C\in(1/2,1)$.
For $\varepsilon$ small enough
    and $n$ large enough,
    there exist a $t$-contextual deletion-correcting code with redundancy
    $(8t(1-C)+\varepsilon)\log n$.
Furthermore,
    both the encoding and decoding
    procedures of this code have time complexity $n^{O(t)}$.
\end{theorem}
Similar to \Cref{subsec:2contexul_del_efficient},
    we will apply the $t$-deletion-correcting codes in \cite{sima2020optimal-systematic} to $f(\boldx)$.
We first describe the properties of the code in \cite{sima2020optimal-systematic}.
\begin{lemma}\label{lem:t_systematic}
\cite[Theorem 1]{sima2020optimal-systematic}
Let $t\geq 1$ be a constant.
There exists a hash function $\mathrm{hash}_t:\{0,1\}^n\rightarrow\{0,1\}^{4t\log n+o(\log n)}$,
    computable in $O(n^{2t+1})$ time,
    such that
$\{(\mathbf{c},\mathrm{hash}_t(\mathbf{c}))~:~\mathbf{c}\in\{0,1\}^n\}$ forms a $t$-deletion correcting code with
    decoding time complexity $O(n^{t+1})$.
\end{lemma}
The construction of our efficient $t$-deletion-correcting code will be almost the same as,
    or even simpler than,
    the one in \Cref{subsec:2contexul_del_efficient}.
Note that the $t$-deletion-correcting code in Lemma \ref{lem:t_systematic}
    applies to \emph{any} input binary string,
    while the two-deletion-correcting code in \Cref{lem:2del_code} requires the input binary string to be regular in the sense of \Cref{def:regular_2del_code}.
As a consequence,
    here we do not need to modify the structured set $\mathcal{C}_{\varepsilon}$ as we did in \Cref{subsec:2contexul_del_efficient}.
\begin{proof}[Proof of Theorem \ref{thm:eff-codes-gen-refined}]
Let $\boldx\in\{0,1\}^{n-1}$ be an arbitrary message.
Define $\overline{\boldx}\coloneqq\Encstruct(\boldx)$,
    where $\Encstruct$ is the efficient encoder from $\{0,1\}^{n-1}$ to $\mathcal{C}_{\varepsilon}$ in \Cref{lem:struct}.
Write $h_t\coloneqq\mathrm{hash}_t(f(\overline{\boldx}))$,
    where $\mathrm{hash}_t$ is the hash function defined in Lemma \ref{lem:t_systematic},
    and $f$ is defined in \Cref{thm:ctxl_del_subseq}.
The overall encoding can be described as
\begin{align}
    \Enc_t(\mathbf{x})\coloneqq E(h_t)\circ (1-\overline{x}_1) \circ \overline{\mathbf{x}},
    \label{eq:enc_t_ctxl_del}
\end{align}
    where $E$ is the runlength-limited encoder in \Cref{thm:rll} and $\overline{x}_1$ is the first bit of $\overline{\boldx}$.
The overall redundancy of the code in \eqref{eq:enc_t_ctxl_del} is
$|E(h_t)|+1=|h_t|+4=(4t+o(1))\log |f(\overline{\boldx})|\leq(8t(1-C)+4t\varepsilon+o(1))\log n$,
    where we used the fact that
    $|f(\overline{\boldx})|\leq n^{2(1-C)+\varepsilon+o(1)}$
    from Property \ref{ppty:thm_ctxl_del_subseq_runlength_ub}
    in \Cref{thm:ctxl_del_subseq}.
The encoding of this code can be computed in $n^{O(t)}$ time since
    $\Encstruct$, $f$, and $E$
    can all be computed in $\poly(n)$ time
    and
    $\mathrm{hash}_t(f(\boldx))$ can be computed in $O(|f(\boldx)|^{2t+1})=n^{O(t)}$ time.

The efficient decoding of the code in \eqref{eq:enc_t_ctxl_del} is almost the same as that for the code in \eqref{eq:enc_two_ctxl_del}.
First,
    the RLL-limited encoder ensures that
    any run in $E(h_t)$ is of length at most  $O(\log (t \log n))<k-1$,
    and thus any contextual deletion in $\Enc_t(\boldx)$ can only happen in the $\overline{\boldx}$ part.
In other words,
    if $\boldy$ is obtained from $\Enc_t(\boldx)$ after at most $t$ contextual deletions,
    then we have
    $\mathbf{y} = E(h_t)\circ 1-\overline{x}_1 \circ \overline{\mathbf{w}}$
    for some $\overline{\mathbf{w}}$ obtained from $\overline{\mathbf{x}}$
    after at most $t$ contextual deletions.
Then,
    we can recover $h_t$ from $E(h_t)$ by \Cref{thm:rll}.
Next,
    similar to \Cref{subsec:2contexul_del_efficient},
    from \Cref{thm:ctxl_del_subseq} we know that
    $f(\overline{\mathbf{w}})$ is obtained from $f(\overline{\boldx})$ via at most $t$ deletions.
Then,
    we apply \Cref{lem:t_systematic}
    to recover $f(\overline{\boldx})$ from $h_t$ and $f(\overline{\mathbf{w}})$
    and then invoke \Cref{thm:ctxl_del_subseq} again 
    to recover $\overline{\boldx}$ from
    $\overline{\mathbf{w}}$,
    $f(\overline{\mathbf{w}})$, and 
    $f(\overline{\boldx})$.
Lastly,
    we recover $\boldx$ by $\boldx=\Decstruct(\overline{\boldx})$,
    where $\Decstruct$ is the efficient decoder from $\mathcal{C}_{\varepsilon}$ to $\{0,1\}^{n-1}$
    defined in \Cref{lem:struct}.
We can check that the overall decoding procedure runs in $n^{O(t)}$ time,
    since
    recovering $f(\overline{\boldx})$ from $h_t$ and $f(\overline{\mathbf{w}})$
    takes $O(|f(\boldx)|^{t+1})=n^{O(t)}$ time
    and all the other steps can be computed in $\poly(n)$ time.

We can arrive at the exact statement in \Cref{thm:eff-codes-gen-refined} in a way similar to the arguments in \Cref{subsec:1contexul_del_efficient,subsec:2contexul_del_efficient}.
First,
    replacing $\varepsilon$ with $\frac{\varepsilon}{8t}$ and repeating the argument lead to an efficient $t$-contextual deletion-correcting code with redundancy $(8t(1-C)+\varepsilon/2+o(1))\log n$.
Then,
    we let $n$ be so large that the $o(1)$ term is below $\varepsilon/2$,
    which establishes \Cref{thm:eff-codes-gen-refined}.
\end{proof}

\section{Efficient codes correcting contextual deletions with logarithmic threshold}
\label{sec:t_ctxl_del}

In this section,
    we prove
    \Cref{thm:eff-codes-gen-poly},
    which is restated here for convenience.
\begin{theorem}[\Cref{thm:eff-codes-gen-poly}, restated]\label{thm:eff-codes-gen-poly-restated}
    Let $k=C\log n$, where $C\in(0,1)$ and $t$ are constants.
    Then, there exist a
    $(t,k)$-contextual deletion-correcting code of block length $n$ with redundancy
\begin{align*}
    18t(1-C)\log n+\left((2C+4)\left\lceil\frac{3}{C}\right\rceil+4
    \right)\log n + o(\log n),
\end{align*}
    where
    the encoding and decoding time complexity 
    is 
   $\poly(n)$ (i.e., the degree of the polynomial in $n$ does not depend on $t$). 
\end{theorem}

At a high level, our proof of~\Cref{thm:eff-codes-gen-poly-restated} proceeds through two main steps.
Fix a threshold $k=C\log n$ for constants $C\in(0,1)$ and $t$.
First, in~\Cref{sec:struct-hash} we carefully design a structured subset of binary strings $\cS_k$ together with a hash function $H$ such that if $\bolds'$ is obtained from $\bolds\in\cS_k$ via at most $t$ contextual deletions with threshold $k$, then the hashes $H(\bolds)$ and $H(\bolds')$, which are vectors over a larger alphabet, are $3t$-close in Hamming distance.
Ignoring some technicalities for now, this allows us to obtain a $(t,k)$-contextual deletion-correcting code by essentially appending to each $\bolds\in\cS_k$ the syndrome of the hash $H(\bolds)$ under an appropriately instantiated Reed-Solomon code correcting $3t$ substitutions, which we show yields the desired redundancy.
Then, in \Cref{sec:eff-gen} we combine bounded independence generators~\cite{alon1992simple} with a slightly modified version of the initial code from \Cref{sec:struct-hash} to obtain $(t,k)$-contextual deletion-correcting codes with efficient encoding and decoding, at the cost of only a negligible increase in redundancy.

\subsection{The structured subset of strings and the hash function}
\label{sec:struct-hash}

Let $k=C\log n$ for a constant $C\in(0,1)$ and fix an arbitrary constant integer $t\geq 1$.
As described above, we begin by describing the relevant structured subset of binary strings $\cS_k$ and hash function $H$.
Then, we show that applying at most $t$ contextual deletions with threshold $k$ to $\bolds\in\cS_k$ corresponds to applying at most $3t$ substitutions to $H(\bolds)$.

We first introduce some relevant notation.
Let $\bolds$ be any binary sequence.
For two nonoverlapping substrings $\mathbf{s}^{(1)}=(s_{\ell_1},s_{\ell_1+1},\ldots,s_{r_1})$ and $\mathbf{s}^{(2)}=(s_{\ell_2},s_{\ell_2+1},\ldots,s_{r_2})$ of a sequence $\mathbf{s}$
    (without loss of generality, we can  assume $\ell_2>r_1$),
    we define the \emph{distance} between $\mathbf{s}^{(1)}$ and $\mathbf{s}^{(2)}$ as $\mathrm{dist}(\mathbf{s}^{(1)},\mathbf{s}^{(2)})\coloneqq\ell_2-r_1$.
For any $a\geq 1$,
    we say $\mathbf{s}^{(2)}$ is \emph{$a$-close} to $\mathbf{s}^{(1)}$ if $\mathrm{dist}(\mathbf{s}^{(1)},\mathbf{s}^{(2)})\leq a$,
    and $\mathbf{s}^{(2)}$ is \emph{$a$-far} from $\mathbf{s}^{(1)}$ if
    $\mathrm{dist}(\mathbf{s}^{(1)},\mathbf{s}^{(2)})> a$.
Define $W\coloneqq 3(\log n-k)$.
Then,
    we can cluster all runs of length at least $k$ in $\bolds$ into $M$ sets of runs,
    each containing $I_m$ runs of length at least $k,$ and for $m\in[1,M]$,
\begin{equation*}
    \mathfrak{C}_m\coloneqq \{(s^m_{i^m_j+1},\ldots,s^m_{i^m_j+\ell^m_j})\}_{m\in [M],j\in [I_m]},
\end{equation*}
such that the following holds:
\begin{enumerate}
    \item $i^m_{j}+\ell^m_j< i^m_{j+1}\leq i^m_{j}+\ell^m_j+W$ for any $m\in[M]$ and $j\in [I_m-1]$  (i.e., within each cluster $\mathfrak{C}_m$ each run of length at least $k$ is $W$-close to the next one).
    \item  $i^m_{I_m}+\ell^m_{I_m}+W<i^{m+1}_1$ for any $m\in [M-1]$ (i.e., the first run in a cluster $\mathfrak{C}_{m+1}$ is $W$-far from the last run in the previous cluster $\mathfrak{C}_m$).
\end{enumerate}
Furthermore,
    for each  $m\in [M]$,
    we encode the information of the cluster $\mathfrak{C}_m$
    into a
    a sequence of three-tuples $V_m(\bolds)\coloneqq(V_m(\mathbf{s})_1,\ldots,V_m(\mathbf{s})_{\lceil\frac{3}{C}\rceil})\in ([0,2\log n-1]\times[0,W]\times\{0,1\})^{\lceil\frac{3}{C}\rceil}$
    as follows:
\begin{itemize}
    \item
     If $\mathfrak{C}_m$ has at most $\lceil\frac{3}{C}\rceil$ runs of length at least $k$
        (i.e. $I_m\leq\lceil\frac{3}{C}\rceil$),
        and if no run in $\mathfrak{C}_m$ is of length at least $2\log n$
        (i.e., $\ell_j^m<2\log n$ for each $j\in[I_m]$),
        we define for each $j\in[\lceil\frac{3}{C}\rceil]$ the $j$th entry of $V_m(\mathbf{s})$ as
    \begin{align*}
    V_m(\bolds)_j = \begin{cases}(\ell^m_j,i^m_{j+1}-i^m_j -~\ell^m_j,\oplus_{i\in\{i^m_j+\ell^m_j+1,\ldots,i^m_{j+1}-1\}}s_i),&\text{if $j<I_m$},\\
    \big(\ell^m_j,(i^{m+1}_{1}-i^m_j-1)\bmod 2,\oplus_{i\in\{i^m_j+\ell^m_j+1,\ldots,i^{m+1}_{1}-1\}}s_i\big),
    &\text{if $j=I_m$},\\
    (0,0,0),&\text{if $j>I_m$,}
    \end{cases}
    \end{align*}
    where $i^{m+1}_1\coloneqq n+1$.
    \item 
       Otherwise, if either $I_m\geq \lceil\frac{3}{C}\rceil+1$ or there exists $j\in[I_m]$ such that $\ell_j^m\geq 2\log n$, let for each $j\in[\lceil\frac{3}{C}\rceil]$
    \begin{align*}
        V_m(\mathbf{s})_j\coloneqq (0,0,0).
    \end{align*}
\end{itemize}
 Each $V_m(\bolds)$ can be uniquely represented by an integer $v_m(\bolds) \in \{0,\ldots,
    N-1\}$,
        where
        $N\coloneqq (4\log n\cdot (W+1))^{\lceil\frac{3}{C}\rceil}$.
Then, 
    for
    each $w\in[0,2^W-1]$ we define
    $H(\bolds,w)\in[0,N]\times [0,2^W]$ as follows:
\begin{itemize}
    \item 
        If there exists a unique $m\in[M]$ such that the binary representation of $w$ is $(s_{i^m_{I_m}+\ell^m_{I_m}+1} \ldots,s_{i^m_{I_m}+\ell^m_{I_m}+W})$,
        then:
    \begin{itemize}
        \item If $m=M$,
            then $H(\bolds,w)\coloneqq (v_m(\bolds),2^W)$.
        \item 
            If $m<M$ and
                the substring
                $\boldw'\coloneqq (s_{i^{m+1}_{I_{m+1}}+\ell^{m+1}_{I_{m+1}}+1} \ldots,s_{i^{m+1}_{I_{m+1}}+\ell^{m+1}_{I_{m+1}}+W})$
                exists, then $H(\bolds,w)\coloneqq (v_m(\bolds),w')$,
                where the binary representation of $w'$ is $\boldw'$. Here, exists is interpreted as not being ``out of bound''.
                For example,
                    if $\bolds=(s_1,\ldots,s_n)$,
                    then $\boldw'=(s_{i^{m+1}_{I_{m+1}}+\ell^{m+1}_{I_{m+1}}+1} \ldots,s_{i^{m+1}_{I_{m+1}}+\ell^{m+1}_{I_{m+1}}+W})$ exists if $i^{m+1}_{I_{m+1}}+\ell^{m+1}_{I_{m+1}}+1\geq 1$ and ${i^{m+1}_{I_{m+1}}+\ell^{m+1}_{I_{m+1}}+W}\leq n$.
    \end{itemize}
    \item 
        If any of the condition above fail to hold, let  $H(\bolds,w)\coloneqq (N,2^W)$.
\end{itemize}
With a slight abuse of notation,
    if $w\in[0,2^W-1]$ has the binary representation $(w_1,\ldots,w_{3(\log n -k)})$,
    we write
    $H(\bolds;w)$
    and
    $H(\bolds;w_1,\ldots,w_{3(\log n-k)})$
    interchangeably.
We can also uniquely express each $H(\bolds,w)$ as an integer in $[0,Q-1]$,
    where $Q\coloneqq (N+1)(2^W+1)=\Theta(n^{3(1-C)}\polylog(n))$.
Furthermore,
    define
\begin{align*}
    H(\mathbf{s};-1)\coloneqq 
    \begin{cases}
        (N,(s_{i^1_{I_1}+\ell^1_{I_1}+1} \ldots,s_{i^1_{I_1}+\ell^1_{I_1}+W})), \textnormal{ if } (s_{i^1_{I_1}+\ell^1_{I_1}+1} \ldots,s_{i^1_{I_1}+\ell^1_{I_1}+W}) \textnormal{ exists},\\
    (N,2^W),\textnormal{ otherwise}.   \end{cases}
\end{align*}
Then,
    for a fixed $\mathbf{s}$,
    we define $H(\mathbf{s})$ as
\begin{align*}
    H(\mathbf{s})\coloneqq (H(\mathbf{s};-1),H(\mathbf{s};0),\ldots,H(\mathbf{s};2^W-1)),
\end{align*}
    which can be viewed as a vector in $[0,Q-1]^{2^W+1}$.

Now we describe the structured set of codewords.
Let $\mathcal{S}_k$ be the set of length-$n$ binary sequences $\bolds$ satisfying all of the following properties:
\begin{enumerate}[(1)]
    \item \label{ppty:t_ctxl_del_no_2logn}
        $\mathbf{s}$ has no runs of length at least $2\log n$.
    \item \label{ppty:t_ctxl_del_distinct_prefix_suffix}
        For every run  $(s_{i+1},\ldots,s_{i+\ell})$ of length at least $k-1$, i.e., $\ell\ge k-1$,
            the length $W-1$ prefixes $$(s_{i+\ell+1},\ldots,s_{i+\ell+W-1})$$ 
            and suffixes 
            $$(s_{i+\ell+2},\ldots,s_{i+\ell+W})$$ 
            of the length-$W$ substring $(s_{i+\ell+1},\ldots,s_{i+\ell+W})$ are all distinct (in fact, we only need distinct length-$W$ substrings following all long runs, and the property that none of these substring is $0^W$ or $1^W$).
    \item \label{ppty:t_ctxl_del_cluster}
        There do not exist $\lceil\frac{3}{C}\rceil+1$ 
            (complete) runs $(s_{i_1+1},\ldots,s_{i_1+\ell_1})$, $(s_{i_2+1},\ldots,s_{i_2+\ell_2}),\ldots,(s_{i_{\lceil\frac{3}{C}\rceil}+1},\ldots,s_{i_{\lceil\frac{3}{C}\rceil}+\ell_{\lceil\frac{3}{C}\rceil}})$,
            each of length at least $\ell_j\ge  k$ for all $j\in [{\lceil\frac{3}{C}\rceil}]$,
            such that $s_{i_j+1}=\ldots=s_{i_j+\ell_j}$
            for $j\in [\lceil\frac{3}{C}\rceil]$
            and $i_{j}+\ell_j<i_{j+1}\le  i_{j}+\ell_j+3(\log n-k)$ for $j\in[\lceil\frac{3}{C}\rceil-1]$. 
\end{enumerate}
Note that for any $\bolds\in\mathcal{S}_k$,
    we have
 $M\leq 2^{W-2}$ by Property~\ref{ppty:t_ctxl_del_distinct_prefix_suffix} and the pigeonhole principle.
In addition,
    we have 
$\ell^m_j< 2\log n$ for each $m\in[M]$ and $j\in[I_m]$ by Property \ref{ppty:t_ctxl_del_no_2logn}, and
$I_m\le \lceil\frac{3}{C}\rceil$ for each $m\in[M]$ by Property \ref{ppty:t_ctxl_del_cluster}.
As a result,
    each $v_m(\bolds)$ will correctly record the information in the cluster $\mathfrak{C}_m$
    (i.e. $v_m(\bolds)\neq 0$).
Furthermore,
    by Property \ref{ppty:t_ctxl_del_distinct_prefix_suffix},
    all the length-$W$ substrings $\boldw^{(m)}\coloneqq(s_{i^m_{I_m}+\ell^m_{I_m}+1} \ldots,s_{i^m_{I_m}+\ell^m_{I_m}+W})$ appearing right after the clusters
    are all distinct for $m\in[1,M]$,
    and thus we have
\begin{align*}
    H(\bolds,-1)&=(N,\boldw^{(1)}),\\
    H(\bolds,\boldw^{(m)})&= (v_m(\bolds),\boldw^{(m+1)})\textnormal{, for }m\in[M-1],\\
    H(\bolds,\boldw^{(M)})&= (v_M(\bolds),2^W).
\end{align*}
We can then think of the $\boldw^{(m)}$s
    as the``signatures'' of the clusters $\mathfrak{C}_m$.
It can be shown that
    $H(\bolds)$ contains all the necessary information to recover $\bolds$ from a contextual deletion-corrupted version $\bolds'$.

The following lemma shows that applying at most $t$ contextual deletions with threshold $k$ to $\bolds\in\cS_k$ corresponds to applying at most $3t$ substitutions to $H(\bolds)$.

\begin{lemma}\label{lem:tsubstitution}
Fix an arbitrary $\bolds\in\cS_k$ and suppose that $\bolds'$ is obtained by performing at most $t$ contextual deletions with threshold $k$ to $\bolds$.
Then, $H(\bolds)$ and $H(\bolds')$ differ in at most $3t$ entries.
\end{lemma}

\begin{proof} 
By the sequential property of contextual deletions described in the proof of \Cref{thm:GV},
    it suffices to show that when applying contextual deletions one by one from right to left,
    one contextual deletion in $\bolds$ can alter at most three entries in $H(\bolds)$.

First,
    consider a contextual deletion happens right after the last run of the cluster $\mathfrak{C}_m,$
    for some $m\in[2,M]$. In particular, prior to this contextual deletion,
    no contextual deletions were present in the clusters $\mathfrak{C}_1,\ldots,\mathfrak{C}_{m}$. Let $\boldw^{(m-1)}$ and $\boldw^{(m)}$
            to be the signature 
            of $\mathfrak{C}_{m-1}$ and $\mathfrak{C}_m$,
            respectively.
There are two possible cases:
\begin{enumerate}
    \item 
        Assume $m\leq M-1$ and that after this contextual deletion,
        $\mathfrak{C}_m$ and $\mathfrak{C}_{m+1}$ merge into a single cluster
        (which happens when the distance between $\mathfrak{C}_m$ and $\mathfrak{C}_{m+1}$ is not ``large  enough'').
        In this case,
            write $\widetilde{\boldw}^{(m+1)}$ for the signature of $\mathfrak{C}_{m+1}$, if it exists.
        Note that $\widetilde{\boldw}^{(m+1)}$ may or may not be the same as the original $\boldw^{(m+1)}$,
            depending on whether there is a contextual deletion in $\mathfrak{C}_{m+1}$ or not.
        Then,
            the following hash values may  possibly change:
        \begin{itemize}
            \item The first entry of $H(\bolds,\boldw^{(m)})$ changes into either:
                \begin{itemize}
                    \item  $v_{m'},$
                    for some $m'>m,$
                        if during prior contextual deletions the substring 
                        $\boldw^{(m)}$
                        becomes the signature of exactly two clusters (and thus the current contextual deletion turn into a unique signature)\footnote{We assumed $m'>m$ since we are applying contextual deletions from right to left, so that the signatures of the prior  clusters still have the correct hash values. At a high level, we argue that $H(\bolds,\boldw^{(m)})$ will change, while $H(\bolds,\boldw^{(m+1)})$ may no longer be a signature.};
                        or
                    \item 
                        $N$ otherwise,
                        since $\boldw^{(m)}$ is no longer the  signature of a cluster.
                \end{itemize}
            \item 
                The value of $H(\bolds,\widetilde{\boldw}^{(m+1)})$
                will also possibly change,
                    since the cluster it correspond to now contains more runs
                (Note that $H(\bolds,\widetilde{\boldw}^{(m+1)})$ will not change if it is already $(N,2^{W-1})$.).
            \item 
                The second entry of $H(\bolds,\boldw^{(m-1)})$ (i.e. the ``pointer'' to the next cluster)
                will change from $\boldw^{(m)}$ into $\widetilde{\boldw}^{(m+1)}$.
        \end{itemize}
    \item 
        Otherwise,
            let $\widetilde{\boldw}^{(m)}$ be the new signature 
            of $\mathfrak{C}_{m}$ after the contextual deletion,
            provided that it exists
            (which may not be the case if $m=M$ and, for example, there are exactly $W-1$ bits after $\mathfrak{C}_M$ before this contextual deletion).
        Then,
            the following hash values will possily change:
        \begin{itemize}
            \item The first entry of $H(\bolds,\boldw^{(m)})$ changes into either:
                \begin{itemize}
                    \item  $v_{m'}$
                    for some $m'>m$, for reasons similar to the above.
                    \item 
                        $N$ otherwise, again for similar reasons.
                \end{itemize}
            \item 
                The value of $H(\bolds,\widetilde{\boldw}^{(m)})$ may also possibly change, as above.
            \item 
                The second entry of $H(\bolds,\boldw^{(m-1)})$
                will change from $\boldw^{(m)}$ to $\widetilde{\boldw}^{(m+1)}$.
        \end{itemize}   
\end{enumerate}
Therefore,
    at most $3$ entries in $H(\bolds)$ are altered.
If any of the substrings defined above does not exists,
    then we have even fewer entries in $H(\bolds)$ changed,
    and thus the arguments still hold.

Next, consider a contextual deletion that happened after a run in the cluster $\mathfrak{C}_m$ that is not its last run. Write $\widetilde{\boldw}^{(m)}$
    for the original signature of $\mathfrak{C}_m$
    before the contextual deletion
    (again,
    $\widetilde{\boldw}^{(m)}$ may or may not be the same as $\boldw^{(m)}$),
    and define $\boldw^{(m-1)}$ as before.
There are three possibilities:
\begin{enumerate}
    \item 
        After the contextual deletion,
        $\mathfrak{C}_m$ splits into two clusters
        (which can happen when $\mathfrak{C}_m$ contains a run $r$ of length exactly $k$ right after another run of length at least $k$,
        the distance between $r,$ the next run of length at least $k$ in $\mathfrak{C}_m,$ is at least $W-k+2$,
        and a contextual deletion reduces the length of $r$ by one).
        In this case,
            let
            $\boldw^{\textnormal{(new)}}$
            be the signature of the newly induced cluster. Then,
            following a similar argument as above,
            $H(\bolds,\widetilde{\boldw}^{(m)})$,
            $H(\bolds,\boldw^{\textnormal{(new)}})$, and
            $H(\bolds,\boldw^{(m-1)})$
            may change.
    \item 
        The cluster $\mathfrak{C}_m$ ends with a run of length exactly $k$
            following a run of length at least $k$,
            and a contextual deletion turns the length of the last run from $k$ into $k-1$.
        In this case,
            the signature of $\mathfrak{C}_m$ changes from $H(\bolds,\widetilde{\boldw}^{(m)})$ into another substring,
            denoted as $\hat{\boldw}^{\textnormal{(new)}}$.
        Then,
            similarly,
            $H(\bolds,\widetilde{\boldw}^{(m)})$,
            $H(\bolds,\hat{\boldw}^{\textnormal{(new)}})$, and
            $H(\bolds,\boldw^{(m-1)})$
            may change.
    \item 
        Otherwise,
            only 
            $H(\bolds,\boldw^{(m)})$ changes.
\end{enumerate}
In this case,
    we also have
    at most $3$ altered entries in $H(\bolds)$.

Finally,
    if $m=1$, we replace $H(\bolds,\boldw^{(m-1)})$ in the discussion above with $H(\bolds,-1)$
    and repeat the arguments.
This concludes the proof of~\Cref{lem:tsubstitution}.
\end{proof}

Motivated by \Cref{lem:tsubstitution}, and looking ahead,
in our final code we will protect $H(\mathbf{s})$ against $3t$ substitution errors, using a Reed-Solomon code over an appropriately large field.
First, we choose $q$ to be the
 smallest prime number larger than 
 $Q$, and recall that 
    $Q=\Theta(n^{3(1-C)}\polylog(n))$  
    is the alphabet size of $H(\mathbf{s})$.
Note that $q$ lies in the interval $[Q+1,2Q]$ by Bertrand's postulate and it can be found in time $\poly(Q)=\poly(n)$ by trial division.
Then,
    define $L\coloneqq 2^W+1$ to be the length of $H(\bolds)$,
    which satisfies $L=\Theta(n^{3(1-C)})$.
Now consider a $[L+6t,L,6t+1]_q$-Reed-Solomon (RS) code over the alphabet $[0,q-1]$. Note that we have 
    $q\geq Q=(N+1)(L+1)\geq L+6t$, with $t$ a constant.
    
It is known that
    the syndrome of such a RS code,
    which can be seen as a function $\mathrm{syn}:[0,q-1]^L\rightarrow[0,q-1]^{6t}$,
    can be computed in $\poly(q,L,6t)=\poly(n)$ time, 
    and that it satisfies
    the following property:
For any sequence $\mathbf{m}\in[0,q-1]^L$,
    if $\mathbf{m}'$ is obtained by substituting at most $3t$ entries in $\mathbf{m}$,
    then $\mathrm{syn}(\mathbf{m})$ and $\mathbf{m}'$ uniquely determine  $\mathbf{m}$ in
    $\poly(q,L,6t)=\poly(n)$ 
    time.
Note that since $Q \leq q$,
    we can treat $H(\bolds)\in [0,Q-1]^L$
    as a sequence in $[0,q-1]^L$ as well.
Therefore,
    for $\mathbf{s}$ and $\mathbf{s}'$ as defined in \Cref{lem:tsubstitution}, given $\mathrm{syn}(H(\mathbf{s}))$ and $H(\mathbf{s'})$
    we can uniquely recover $H(\mathbf{s})$
    in $\poly(n)$ time.
The syndrome $\mathrm{syn}(H(\mathbf{s}))$ can be computed in $\poly(n)$ time
    and can be represented by a binary sequence of length
\begin{equation}\label{eq:red-RS-hash}
   6t\log q=18t(1-C)\log n +o(\log n).
\end{equation}
Looking ahead, this will essentially correspond to the redundancy of our final efficient code.

Recall that $\mathcal{D}_t^{(k)}(\bolds)$ is the set of sequences obtained from $\bolds$ after at most $t$ contextual deletions.
It remains to show that $\bolds$ can be recovered from $\bolds'\in\mathcal{D}_t^{(k)}(s)$ and $H(\bolds)$.
This is guaranteed by the following lemma.
\begin{lemma}\label{lem:recovers}
For any $\bolds\in \mathcal{S}_k$ and $\bolds'\in\mathcal{D}_t^{(k)}(\bolds)$, we can efficiently and uniquely recover $\bolds$ from $\bolds'$ and $H(\bolds)$.
\end{lemma}

\begin{proof}
First,
    by looking at the second entry of $H(\bolds,-1)$,
    we can recover $\boldw^{(1)}$,
    the true signature of the first cluster $\mathfrak{C}_1$.
Then,
    by reading the second entry of $H(\bolds,\boldw^{(1)})$,
    we can retrieve $\boldw^{(2)}$,
    which allows us to find   $H(\bolds,\boldw^{(2)})$.
By repeating this process,
    we can obtain all the hash values $H(\bolds,\boldw^{(3)}),\ldots,H(\bolds,\boldw^{(M)})$.

Let $r^{(1)},\ldots,r^{(I_1)}$ be the (correct) runs of length at least $k$ in $\mathfrak{C}_1$,
    and let $\ell_1\ldots,\ell_{I_1}$ be their respective lengths.
Note that $I_1$ and $\ell_1,\ldots,\ell_{I_1}$ are reconstructable from the hash value $H(\bolds,\boldw^{(1)})$.
Then,
    from left to right, we locate the first run of length at least $k$ in $\bolds'$,
    denoted as $\widetilde{r}^{(1)}$,
    and write $L_1$ for the length of $\widetilde{r}^{(1)}$.
Without loss of generality, we assume $\widetilde{r}^{(1)}$ is a $0$-run,
    and thus write $\widetilde{r}^{(1)}=0^{L_1}$. 
Note that the starting position and the parity of $\widetilde{r}^{(1)}$
    is necessarily the same as that of the true $r^{(1)}$,
    but
    $\widetilde{r}^{(1)}$ now may consist of several runs of length at least $k$ in $\mathfrak{C}_1$ (when they ``merge together'' by absorbing single-bit runs).
Note that
    if $L_1\geq \ell_1+1$,
    then necessarily $r^{(1)}$ had a deleted a single-bit run following it.
We hence recover $r^{(1)}$ by adding a $1$ after the first $\ell_1$ bits of $\widetilde{r}^{(1)}$.

We now consider $L_1=\ell_1$. In this case we know that $\widetilde{r}^{(1)}$ is the true $r^{(1)}$.
To determine whether or not to add a $1$ after $r^{(1)}$,
    we proceed according to:
\begin{itemize}
    \item 
        If $r^{(1)}$ is the last run in this cluster (i.e. $I_1=1$),
            the length-$W$ substring after $r^{(1)}$ must be the true signature $\boldw^{(1)}$.
        Let the length-$W$ substring after $r^{(1)}$ (before adding a $1$) be $\boldu$.
        It follows that by adding a $1$ after it, the length-$W$ substring following  $r^{(1)}$ becomes $1\boldu_{1:W-1}$.
        The only possibility that $\boldu=1\boldu_{1:W-1}$ is that both are $1^W$,
            but we have forbidden the pattern $0^k1^W$ in $\bolds$ by Property~\ref{ppty:t_ctxl_del_distinct_prefix_suffix} of $\mathcal{S}_k$.
        Therefore,
            $\boldu$ and $1\boldu_{1:W-1}$ are not equal,
            and exactly one of them is the true signature $\boldw^{(1)}$.
        We can thus determine whether we need to add a $1$ after $r^{(1)}$.
    \item 
        If there is another run of length at least $k$ after $r^{(1)}$ in this cluster (i.e., $I_1\geq 2$),
            then we let $\delta_1$ be the (true) distance between $r^{(1)}$ and $r^{(2)}$,
            which is also available from $H(\bolds,\boldw^{(1)})$.
        \begin{itemize}
            \item 
                If $\delta_1\geq 1$,
                    then we examine the distance between $r^{(1)}$ and the next run of length at least $k$ (before adding a $1$).
                Let this quantity be $\hat{\delta}$.
                It is necessary that $\hat{\delta}\geq 1$ as well,
                    since if $r^{(1)}$ is adjacent to the next run of length at least $k$,
                    the same statement holds after adding a $1$ after $r^{(1)}$.
                It follows that the next run of $r^{(1)}$ is of length at most $k-1$.
                Then,
                    note that adding a $1$ after $r^{(1)}$ can only either decrease $\hat{\delta}$ to $0$ if the next run of $r^{(1)}$ is of length exactly $k-1$;
                    or increase $\hat{\delta}$ by one otherwise.
                We can thus uniquely determine which is the correct case.
            \item 
                Consider $\delta_1=0$. Let  $\widetilde{r}^{(2)}$ be 
                    the run after $r^{(1)}$,
                    and let $L_2$ be the length of $\widetilde{r}^{(2)}$.
                Note that adding a $1$ after $r^{(1)}$ is equivalent to  increasing $L_2$ by one.
                And thus the task of determining whether to add a $1$ after $r^{(1)}$ or not is the same as determining whether $L_2$ or $L_2+1$ is the ``correct'' length of $\widetilde{r}^{(2)}$.
                
                Notice that it is necessary that $L_2\geq \ell_2-1$,
                    since $r^{(2)}$ has to ``fit in $\widetilde{r}^{(2)}$'' to agree with the assumption $\delta_1=0$.
                Furthermore,
                    if $L_2=\ell_2-1$,
                        we deduce that we have to append a $1$ to $r^{(1)}$,
                        which is the only way $r^{(2)}$ can ``fit in $\widetilde{r}^{(2)}$''.
                \begin{remark}\label{rmk:r2}
                    In the case above case where $L_2=\ell_2-1$ and $r^{(2)}$ ``fits'' at the end of $\widetilde{r}^{(2)}$ after appending a $1$ to $r^{(1)}$, 
                        to determine whether we have to append a bit to $r^{(2)}$ or not,
                        we need to repeat the argument
                        with $r^{(1)}$ replaced by $r^{(2)}$.
                    All other similar cases 
                    can be handled in the same manner. 
                \end{remark}
            
                Now consider $L_2\geq \ell_2$.
                Note that $\widetilde{r}^{(2)}$ may consists of several runs of length at least $k$, merged  by contextually deleting single-bit runs.
                We thus repeatedly apply the argument used on $L_1\geq \ell_1+1$.
                In words,
                    we first place $r^{(2)}$ at the start of $\widetilde{r}^{(2)}$ and then, if $L_2-\ell_2>0$, we add a single-bit run after $r^{(2)}$
                    to match the correct length of $r^{(2)}$.
                Then, for the remaining length-$(L_2-\ell_2)$ suffix of $\widetilde{r}^{(2)}$,
                    we can determine if we should place $r^{(3)}$ at its start by checking whether $L_2-\ell_2$ is larger than or equal to $\ell_3$ or not. If so,
                    we repeat the process to determine if $r^{(4)}$ fits in the remaining length-$(L_2-\ell_2-\ell_3)$ suffix of $\widetilde{r}^{(2)}$.
                Let this process end at some index $i\geq 2$.
                More precisely, define
                \begin{align*}
                    i\coloneqq \max\{i'\in[2,I_1]~:~\ell_2+\cdots+\ell_{i'}\leq L_2\},
                \end{align*}
                    and then let $\xi\coloneqq L_2-(\ell_2+\cdots+\ell_{i})$ be the ``remaining length'' after putting $r^{(2)},\ldots,r^{(i)}$ into $\widetilde{r}^{(2)}$.
                In particular,
                    we have $\xi\geq 0$,
                    and if $I_m\geq i+1$,
                    we further have $\xi\leq \ell_{i+1}-1$
                    (or otherwise, we contradict the maximality of $i$).
                Then, 
                    we can successfully recover the runs $r^{(2)},\ldots,r^{(i-1)}$.
                More precisely,
                    without adding a $1$ after $r^{(1)}$,
                    we get
                \begin{align}
                    0^{\ell_1}1^{\ell_2}01^{\ell_3}0\cdots 1^{\ell_{i-1}}01^{\ell_i+\xi}.
                    \label{eq:r1noadd}
                \end{align}
                Note that we do not directly write~\eqref{eq:r1noadd} as $0^{\ell_1}\cdots 1^{\ell_i}01^{\xi}$ since
                    $\xi$ may be $0$ and thus $r^{(i)}$ may not have an added $0$ following it.
                On the other hand, after appending a $1$ to $r^{(1)}$,
                    the ``remaining length'' $\xi+1$ is at least one,
                    and thus we have to necessarily append a $0$ to $r^{(i)}$. That is,
                    we get
                \begin{align}
                    0^{\ell_1}1^{\ell_2}01^{\ell_3}0\cdots 1^{\ell_{i-1}}01^{\ell_i}01^{\xi+1}.
                    \label{eq:r1add}
                \end{align}
                Our goal then becomes to determine which of~\eqref{eq:r1noadd} and~\eqref{eq:r1add} is correct.
                
                We first assume $\xi=0$ and consider the following:
                \begin{itemize}
                    \item 
                        If $I_1\geq i+1$,
                            then we examine the next run right after $r^{(i)}$,
                            denoted as $\widetilde{r}^{(3)}$.
                        Let the length of $\widetilde{r}^{(3)}$ be $L_3$.
                        The task becomes to determine which of the following three cases is correct:
                        \begin{enumerate}[(A)]
                            \item \label{case:r1NoAdd_riNoAdd} $0^{\ell_1}\cdots 1^{\ell_i}0^{L_3}$ (neither $r^{(1)}$ nor $r^{(i)}$ have an added bit back).
                            \item \label{case:r1NoAdd_riAdd} $0^{\ell_1}\cdots 1^{\ell_i}0^{L_3+1}$ ($r^{(1)}$ does not have an added bit back but $r^{(i)}$ does).
                            \item \label{case:r1Add_riAdd} $0^{\ell_1}\cdots 1^{\ell_i}010^{L_3}$
                            (both $r^{(1)}$ and $r^{(i)}$ have an added bit back).
                        \end{enumerate}
                        \begin{itemize}
                            \item
                                If $L_3\leq k-2$,
                                    then in any of the three cases the next run of length at least $k$ following  $r^{(i)}$ happens  after $\widetilde{r}^{(3)}$.
                                Let $\hat{\delta}'$ be the calculated distance between $r^{(i)}$ and the next run of length at least $k$ in Case \ref{case:r1NoAdd_riNoAdd}.
                                Then the associated distance in Case \ref{case:r1NoAdd_riAdd} and Case \ref{case:r1Add_riAdd} are $\hat{\delta}'+1$ and $\hat{\delta}'+2$, respectively.
                                Only one of them can be the true distance between $r^{(i)}$ and $r^{(i+1)}$,
                                    denoted as $\delta_i$.
                            \item 
                                If $L_3=k-1$,
                                    then we examine $\delta_i$.
                                    If $\delta_i=0$,
                                        then necessarily  Case ~\ref{case:r1NoAdd_riAdd} is the correct one.
                                        On the other hand, if $\delta_i\geq 1$,
                                            then only Cases \ref{case:r1NoAdd_riNoAdd}
                                            and \ref{case:r1Add_riAdd}
                                            are possible.
                                        As calculated before,
                                            the distances to the next run of length at least $k$ in these two cases differ by two,
                                        and only one of them can be the true $\delta_2$.
                            \item 
                                If $L_3\geq k$,
                                    then necessarily $\delta_2$ is either $0$ or $2$.
                                If $\delta_2=0$ then only Cases~\ref{case:r1NoAdd_riNoAdd}
                                and~\ref{case:r1NoAdd_riAdd} are possible,
                                and thus $r^{(1)}$ does not need adding back a bit.
                                \begin{remark}\label{rmk:ri}
                                    Similar to~\Cref{rmk:r2},
                                    to determine whether we need to add a bit back to $r^{(i)}$ or not,
                                        we repeat the above analysis with $r^{(1)}$ replaced by $r^{(i)}$.
                                \end{remark}
                        \end{itemize}
                    \item
                        If $I_1=i$,
                            then following the definition above,
                            it is necessary that $L_3\leq k-1$.
                        We then look at the length-$W$ substring after $r^{(i)}$ in each case,
                            i.e. the candidate signatures.
                        Let the length-$W$ after $r^{(i)}$ in Case \ref{case:r1NoAdd_riNoAdd} be $\boldu$.
                        In particular,
                            $\boldu$ starts with $0^{L_3}1$
                            since $L_3\leq k-1$ and we have assumed $W\geq k$.
                        Then the signatures in the three cases are $\boldu$ for Case \ref{case:r1NoAdd_riNoAdd},
                        $0\boldu_{1:W-1}$ for Case \ref{case:r1NoAdd_riAdd},
                        and $01\boldu_{1:W-2}$ for Case \ref{case:r1Add_riAdd}.
                        Similar to the cases above,
                            $\boldu=0\boldu_{1:W-1}$ only when both of them are $0^W$,
                            which has been forbidden in $\bolds$.
                        Thus they are unequal.
                        At the same time,
                            since $\boldu$ starts with $0$ (recall that $L_3$ is the length of $\widetilde{r}^{(3)}$,
                            which is necessarily positive),
                            we have 
                            $0\boldu_{1:W-1}\neq 01\boldu_{1:W-2}$.
                        Therefore,
                            the only possible ``collision''
                            is the case $\boldu=01\boldu_{1:W-2}$,
                            which happens only when both are $01010101\cdots$.
                        In this case,
                            we calculate the XOR of all the bits between $r^{(i)}$ and the next run of length at least $k$ for both cases,
                            which are different since there is an extra $01$ in Case \ref{case:r1Add_riAdd}
                            compared with \ref{case:r1NoAdd_riNoAdd}.
                        Note that the true value is available in $H(\bolds,\boldw^{(1)})$,
                            and thus we can discern which of Case \ref{case:r1NoAdd_riNoAdd}
                            and Case \ref{case:r1Add_riAdd}
                            is correct.
                \end{itemize}
            We now assume $\xi\geq 1$.
            Then, we must append a $0$ to $r^{(i)}$.
            It follows that
                to determine whether $r^{(1)}$ has to be appended by a $1$ or not
                is the same as to
                discriminate between the following two scenarios:
            \begin{align}
                0^{\ell_1}\cdots 1^{\ell_i}01^{\xi}
                \label{eq:r1noadd_xigeq1}
            \end{align}
            and
            \begin{align}
                0^{\ell_1}\cdots 1^{\ell_i}01^{\xi+1}.
                \label{eq:r1add_xigeq1}
            \end{align}
            Note that if $\xi\geq k$,
                then necessarily $I_1\geq i+1$,
                $\xi=\ell_{i+1}-1$,
                and we have to append $1$ to $r^{(1)}$.
            At the same time,
                if $\xi=k-1$,
                then $r^{(1)}$ we add back a bit if and only if $I_1\geq i+1$, $\ell_{i+1}=k$, and $\mathrm{dist}(r^{(i)},r^{(i+1)})=1$.
            \begin{remark}\label{rmk:rip1}
            Similar to~\Cref{rmk:r2,rmk:ri},
                in the cases above where we append a $1$ to $r^{(1)}$ to ``make room'' for $r^{(i+1)}$ to fit within $\widetilde{r}^{(3)}$,
                to determine whether to append a $0$ to $r^{(i+1)}$ we repeat the same arguments with $r^{(1)}$ replaced by $r^{(i+1)}$.
            \end{remark}
            In the following we assume $1\leq\xi\leq k-2$.
            \begin{itemize}
                \item
                    If $I_1=i$,
                        we calculate the distance between $r^{(i)}$ and the next run of length at least $k$ for the case~\eqref{eq:r1noadd_xigeq1},
                        and denote this distance by $\hat{\delta}''$.
                    For the same rule pertaining to the case in~\eqref{eq:r1add_xigeq1},
                        we denote the calculated distance by $\hat{\delta}''+1$.
                    Note that
                        $\hat{\delta}''\not\equiv \hat{\delta}''+1\bmod 2$
                        and the value of the true distance modulo $2$ is available from $H(\bolds,\boldw^{(1)})$.
                    Therefore,
                        we can discern which of~\eqref{eq:r1noadd_xigeq1} and~\eqref{eq:r1add_xigeq1} is correct, i.e., 
                        whether to append a $1$ to $r^{(1)}$ or not.
                \item 
                    If $I_1\geq i+1$,
                        then the true distance between $r^{(i)}$ and $r^{(i+1)}$,
                        which can be found from $H(\bolds,\boldw^{(1)})$,
                        can uniquely determine whether to append a $1$ to $r^{(1)}$.
                    The reason is that the calculated distance between $1^{\ell_i}$ and the next run of length at least $k$ in the two cases, ~\eqref{eq:r1noadd_xigeq1} and~\eqref{eq:r1add_xigeq1},
                    differ by one.
                \end{itemize}
        \end{itemize}
\end{itemize}
For the case described in~\Cref{rmk:r2},
    we can similarly determine whether to append a bit to $r^{(1)}$, $r^{(2)},\dots$ or not.
The same is true for~\Cref{rmk:ri}, as based on $r^{(1)},\ldots,r^{(i-1)}$ we can determine if we should append a bit to $r^{(i)}$ or not; and, for~\Cref{rmk:rip1} as well.
For all the other cases discussed above, we can recover the runs involved and determine if to to add back a bit of the opposite parity or not. This argument shows that we can perform decoding to recover all the runs of length at least $k$ in $\bolds$ in time upper-bounded by a polynomial in $n$ whose degree does not depend on $t$.
\end{proof}
\begin{example}
Consider $n=64$ (i.e., $\log n =6$) and $C=\frac{2}{3}$ (i.e., $k=4$ and $W=3(\log n- k)=6$). Examine the sequence
\begin{align}
    \bolds = \underbrace{000000111101111111}_{\textnormal{First cluster }\mathfrak{C}_1}\overline{010101}0\underbrace{1111111111}_{\substack{\textnormal{Second cluster}\\\mathfrak{C}_2}}\overline{001100}1010\cdots,
    \label{eq:example_t_ctxl_del_s}
\end{align}
    where the remaining bits on the right alternate between $1$s and $0$. It can be seen that $\bolds$ has two clusters, as underbraced in~\eqref{eq:example_t_ctxl_del_s}.
We also overlined the length-$W$ signature after each cluster in~\eqref{eq:example_t_ctxl_del_s}
    (i.e., $\boldw^{(1)}=010101$ and $\boldw^{(2)}=001100$).
The first cluster $\mathfrak{C}_1$ has $I_1=3$ runs of length at least $k$,
    which is below the threshold $\lceil\frac{3}{C}\rceil=5$.
Furthermore,
    $\mathfrak{C}_1$ has no runs of length at least $2\log n=12$.
Therefore,
    the information in $\mathfrak{C}_1$ 
    can be expressed in terms of the sequence of three-tuples
    $V_1(\bolds)=(V_1(\bolds)_1,V_1(\bolds)_2,V_1(\bolds)_3,V_1(\bolds)_4,V_1(\bolds)_5)\in\left([0,11],[0,6],\{0,1\}\right)^5$,
    where
\begin{align*}
    V_1(\bolds)_1&=(6,0,0),\\
    V_1(\bolds)_2&=(4,1,0),\\
    V_1(\bolds)_3&=(7,1,1),\\
    V_1(\bolds)_4&=(0,0,0),\\
    V_1(\bolds)_5&=(0,0,0).
\end{align*}
Similarly,
    the information in $\mathfrak{C}_2$ 
    is contained in
    $V_2(\bolds)=(V_2(\bolds)_1,V_2(\bolds)_2,V_2(\bolds)_3,V_2(\bolds)_4,V_2(\bolds)_5)$,
    where
\begin{align*}
    V_2(\bolds)_1&=(10,0,0),\\
    V_2(\bolds)_2&=(0,0,0),\\
    V_2(\bolds)_3&=(0,0,0),\\
    V_2(\bolds)_4&=(0,0,0),\\
    V_2(\bolds)_5&=(0,0,0).
\end{align*}
We can uniquely express $V_1(\bolds)$ and $V_2(\bolds)$
    as two integers $v_1(\bolds)$ and $v_2(\bolds)$ in the range $[0,N-1]$,
    where $N=(12\times 7\times 2)^5=168^5$.
Then,
    the hash functions for $\bolds$ is
\begin{align*}
    H(\bolds,-1)&=(168^5,010101),\\
    H(\bolds,010101)&=(v_1(\bolds),001100),\\
    H(\bolds,001100)&=(v_2(\bolds),64).
\end{align*}
For any other binary string $\boldw\in\{0,1\}^6\setminus\{010101,001100\}$,
    we have
\begin{align*}
    H(\bolds,\boldw)=(168^5,64).
\end{align*}
The overall hash function is
    $H(\bolds)=(H(\bolds,-1),H(\bolds,000000),\ldots,H(\bolds,111111))$,
    which can be represented as a length-$65$ sequence over an alphabet of size $(168^5+1)\times 65$.
Note that
    even though in this specific example $N=168^5$ is much larger than $2^W=64$, one should be reminded that $N=\polylog(n)$ while $2^W=n^{3(1-C)}$.

Suppose next that we receive
    a corrupted sequence as below,
\begin{align}
    \bolds' = 0000001111111111101010101111111111011001010\cdots,
    \label{eq:example_t_ctxl_del_s_prime}
\end{align}
    which is obtained from $\bolds$ via two contextual deletion.
To be more precise, we deleted bits following 
    the second run of length at least $k$ in $\mathfrak{C}_1$
    and the only run in $\mathfrak{C}_2$.
We first identify the clusters and signatures in $\bolds'$,
    which leads to
\begin{align*}
    \bolds' = \underbrace{00000011111111111}_{\tn{First cluster }\mathfrak{C}_1'}\overline{010101}0\underbrace{1111111111}_{\substack{\tn{Second cluster}\\\mathfrak{C}_2'}}\overline{011001}010\cdots.
\end{align*}
It can be seen that $\bolds'$ also has two clusters,
    which is the same as $\bolds$.
The information 
    in $\mathfrak{C}_1'$
    is encapsulated in
    the sequence of three-tuples $V_1(\bolds')=(V_1(\bolds')_1,V_1(\bolds')_2,V_1(\bolds')_3,V_1(\bolds')_4,V_1(\bolds')_5)$, where
\begin{align*}
    V_1(\bolds')_1&=(6,0,0),\\
    V_1(\bolds')_2&=(11,1,1),\\
    V_1(\bolds')_3&=(0,0,0),\\
    V_1(\bolds')_4&=(0,0,0),\\
    V_1(\bolds')_5&=(0,0,0),
\end{align*}
    which is different from $V_1(\bolds)$.
It follows that the integer $v_1(\bolds')$ that represents $V_1(\bolds')$ is also different from $v_1(\bolds)$.
On the other hand,
    the second cluster is the same in both $\bolds$ and $\bolds'$.
    (only their signatures are different).
Therefore,
    we have
\begin{align*}
    V_2(\bolds')=V_2(\bolds),
\end{align*}
    and the integer representation $v_2(\bolds')$ of $V_2(\bolds')$
    is also the same as $v_2(\bolds)$.
Then,
    the overall hash function $H(\bolds')$ is
\begin{align*}
    H(\bolds',-1)&=(168^5,010101),\\
    H(\bolds',010101)&=(v_1(\bolds)',011001),\\
    H(\bolds',011001)&=(v_2(\bolds)',64),\\
    H(\bolds',\boldw)&=(168^5,64)\textnormal{ for }\boldw\in\{0,1\}^6\setminus\{010101,011001\}.
\end{align*}
Comparing $H(\bolds)$ and $H(\bolds')$,
    we deduce that the following entries are different:
\begin{align*}
    H(\bolds,010101)=(v_1(\bolds),001100)&\neq H(\bolds',010101)=(v_1(\bolds'),011001),\\
    H(\bolds,001100)=(v_2(\bolds),64)&\neq H(\bolds',001100)=(168^5,64),\\
    H(\bolds,011001)=(168^5,64)&\neq H(\bolds',011001)=(v_2(\bolds)',64).
\end{align*}
At the same time,
    $H(\bolds,-1)=H(\bolds',-1)=(168^5,010101)$
    and $H(\bolds,\boldw)=H(\bolds',\boldw)$ for $\boldw\in\{0,1\}^6\setminus\{010101,001100,011001\}$.
Therefore,
    $H(\bolds')$,
    when viewed as a length-$65$ sequence over an alphabet of size $(168^5+1)\times 65$,
    differs from $H(\bolds)$ in $3$ entries.
This result agrees with \Cref{lem:tsubstitution},
    which asserts that $H(\bolds)$ and $H(\bolds')$ differ in at most $3\times 2=6$ entries.

We now show how to recover $\bolds$ from $\bolds'$ in \eqref{eq:example_t_ctxl_del_s_prime} and $H(\bolds)$ using the decoding process described in the proof of \Cref{lem:recovers}.
First,
    by reading the second entry of $H(\bolds,-1)$, we retrieve $\boldw^{(1)}=010101$.
Then,
    the first entry of $H(\bolds,010101)$ is
    $v_1(\bolds)$,
    which contains the information in the first cluster $\mathfrak{C}_1$.
We then convert $v_1(\bolds)$ into $V_1(\bolds)_1,\ldots,V_1(\bolds)_5$.
The first entry of these three-tuples are $6,4,7,0,0$,
    and thus we deduce that $\mathfrak{C}_1$ has three runs of length at least $k=4$ (i.e. $I_1=3$).
Denote them as $r^{(1)}$ and $r^{(2)}$ and $r^{(3)}$,
    respectively.
Furthermore,
    their lengths are $\ell_1=6$, $\ell_2=4$, and $\ell_3=7$, respectively.

We start scanning $\bolds'$ from left to right,
    and we examine the first run of length at least $k=4$,
\begin{align*}
    \underline{000000}1111111111101010101111111111011001010\cdots,
\end{align*}
    as underlined. Observing that $\ell_1=6$,
    we deduce that the underlined part is the correct $r^{(1)}$.
Since there are still two runs in this cluster,
    we examine the distance between $r^{(1)}$ and $r^{(2)}$,
    i.e., the second entry of $V_1(\bolds)_1$,
    which is $0$.
Therefore,
    we record the ``$\delta_1=0$ case'' in the decoding algorithm.
We now have to examine the next run.
To be more precise,
    we have to consider two cases. In the first case, we have to consider appending a $1$ to  the first $000000$. If we do not append the bit, we get
\begin{align}
    \bolds'_{r^{(1)}\tn{ no add}}=000000\underline{11111111111}01010101111111111011001010\cdots,
\end{align}
    where the underlined run has length $11$.
On the other hand,
    if we append back $1$ to $000000$,
    we obtain
\begin{align*}
    \bolds'_{r^{(1)}\tn{ add}}=000000\underline{111111111111}01010101111111111011001010\cdots,
\end{align*}
    where the length of the underlined run equals $12$.
Note that either case,
    comparing with the information $\ell_2=4$,
    we deduce that we have to insert a $0$ to create the correct runlength of $4$,
    leading to the following decoding (partial) results:
\begin{align*}
    \bolds''_{r^{(1)} \tn{ no add}}&=00000011110\underline{1111111}01010101111111111011001010\cdots,\\
    \bolds''_{r^{(1)} \tn{ add}}&=00000011110\underline{11111111}01010101111111111011001010\cdots.
\end{align*}
Then,
    using the fact that $\ell_3=7$,
    we can deduce that there are three possibilities:
\begin{align*}
    \bolds'''_{r^{(1)} \tn{ no add } r^{(3)} \tn{ no add}}&=00000011110\underline{1111111}\overline{010101}01111111111011001010\cdots,\\
    \bolds'''_{r^{(1)} \tn{ no add } r^{(3)} \tn{ add}}&=00000011110\underline{1111111}\overline{001010}101111111111011001010\cdots,\\
    \bolds'''_{r^{(1)}\tn{ add }r^{(3)}\tn{ add}}&=00000011110\underline{1111111}\overline{010101}0101111111111011001010\cdots,
\end{align*}
    where the underlined parts in all three cases correspond to $r^{(3)}$. Note that for only two of the three cases,
    the overlined parts agree with the true signature $\boldw^{(1)}=010101$.
To discern which one is correct, we read the third entry of $V_1(\bolds)_3$,
    which entails that the XOR of all the bits between $r^{(3)}$ and the start of the next run of length at least $k$ should be $1$.
However,
    this value is $1$ for the ``$r^{(1)} \tn{ no add } r^{(3)} \tn{ no add}$'' case and $0$ for the ``$r^{(1)}\tn{ add }r^{(3)}\tn{ add}$'' case,
    which can be calculated by XORing the underlined parts for both cases below:
\begin{align*}
    \bolds'''_{r^{(1)} \tn{ no add } r^{(3)} \tn{ no add}}&=000000111101111111\underline{0101010}1111111111011001010\cdots,\\
    \bolds'''_{r^{(1)}\tn{ add }r^{(3)}\tn{ add}}&=000000111101111111\underline{010101010}1111111111011001010\cdots,
\end{align*}
We can thus deduce that the  ``$r^{(1)} \tn{ no add } r^{(3)} \tn{ no add}$'' case is correct.

In summary,
    so far we have successfully recovered $\mathfrak{C}_1$ as
\begin{align*}
    \underbrace{000000111101111111}_{\tn{Recovered } \mathfrak{C}_1}0101010\underline{1111111111}011001010\cdots.
\end{align*}

We continue the decoding process by moving to the next run of length at least $k$,
    as underlined in the expression above.
At the same time,
    by reading the second entry of $H(\bolds,010101)$,
    we know that $\boldw^{(2)}=001100$,
    and thus we can obtain $v_2(\bolds)$ from
    the first entry of $H(\bolds,001100)$.
Then,
    by looking at $V_2(\bolds)$,
    we can deduce that $I_2=1$ and that the only run is of length $10$.
Compared with the length of the underlined run,
    we deduce that the underlined part is already the only run in $\mathfrak{C}_2$.
To determine
    whether we should append a $0$ to that run, we notice that
    the signature without adding is $011001$,
    and the signature becomes $001100$ after adding a $0$.
Only the latter case agrees with the true signature $\boldw^{(2)}$,
    and thus we add a $0$ after the only run in $\mathfrak{C}_2$.
We have thus completed the decoding process and recovered $\bolds$ as
\begin{align*}
    \bolds=\underbrace{000000111101111111}_{\tn{Recovered } \mathfrak{C}_1}0101010\underbrace{1111111111}_{\tn{Recovered }\mathfrak{C}_2}0011001010\cdots,
\end{align*}
    which is the same as~\eqref{eq:example_t_ctxl_del_s}.
\end{example}

\subsection{The efficiently encodable and decodable codes}\label{sec:eff-gen}

We will now take the remaining steps to turn the results obtained in \Cref{sec:struct-hash} into an efficiently encodable and decodable $(t,k)$-contextual deletion-correcting code.

\paragraph{Efficiently encoding into $\cS_k$}
We begin by giving an efficient algorithm that injectively encodes a message $\boldx$ into a structured string $\bolds\in\cS_k$ with little redundancy.
To this end, we use \emph{almost $\kappa$-wise independent random variables}~\cite{alon1992simple} to ``mask'' the message and ensure it satisfies the desired structural properties.
This high-level approach has been previously used in the context of deletion-correcting codes in, e.g.,~\cite{CJLW22,CGMR20}.

\begin{definition}[Almost $\kappa$-wise independent random variable]\label{def:almost_indep}
Let $\kappa$ and $n$ be positive integers and let $\varepsilon>0$.
A random variable $X=(X_1,X_2,\ldots,X_n)\in\{0,1\}^n$ is said to be
    \emph{$\varepsilon$-almost $\kappa$-wise independent}
    if
    for all indices $1\leq i_1<i_2<\cdots<i_{\kappa}\leq n$ and
    any $(x_1,\ldots,x_k)\in\{0,1\}^\kappa$
    it holds that
\begin{align*}
    |\Prob{X_{i_1}=x_1,X_{i_2}=x_2,\ldots,X_{i_{\kappa}}=x_{\kappa}}-2^{-\kappa}|\leq \varepsilon.
\end{align*}
\end{definition}

The following theorem shows that almost $\kappa$-wise independent random variables can be constructed efficiently from few independent and uniformly at random bits.
\begin{theorem}[{\cite[Theorem 2]{alon1992simple}}]\label{thm:almost_ind_generator}
Let $\kappa$ and $n$ be positive integers and $\varepsilon>0$.
There exists a function (generator) $g:\{0,1\}^d\rightarrow \{0,1\}^n$,
    where $d=(2+o(1))\log(\frac{\kappa\log n}{2\varepsilon})$,
    such that $g(U_d)$ is $\varepsilon$-almost $\kappa$-wise independent,
    where $U_d$ denotes the uniform distribution over $\{0,1\}^d$.
    Furthermore, $g$ is computable in time $\poly(n)$.
\end{theorem}
A simple but important \emph{masking} property that we will exploit below is that for any fixed string $\mathbf{x}\in\{0,1\}^n$ the random variable $g(U_d)+\mathbf{x}$ is also $\varepsilon$-almost $\kappa$-wise independent, where $g(U_d)+\mathbf{x}$ denotes the bit-wise XOR of $g(U_d)$ and $\mathbf{x}$ and $g$ is the function from \Cref{thm:almost_ind_generator}.

The following lemma states that bounded independence suffices to satisfy all but one properties defining $\cS_k$.
Because the proof of this lemma is long, we defer it to \Cref{sec:t_contextual_almost_ind} to avoid breaking the exposition here.
\begin{lemma}\label{lem:t_contextual_almost_ind}
Let $\kappa=(2\lceil\frac{3}{C}\rceil+2)\log n$ and $\varepsilon=n^{-((C+2)\lceil\frac{3}{C}\rceil+2)}$
(and thus $d$ in \Cref{thm:almost_ind_generator} is $((2C+4)\lceil\frac{3}{C}\rceil+4+o(1))\log n$).
Then, an $\varepsilon$-almost $\kappa$-wise independent random vector $X=(X_1,\ldots,X_n)$
    satisfies Properties \ref{ppty:t_ctxl_del_no_2logn}, \ref{ppty:t_ctxl_del_distinct_prefix_suffix}, and \ref{ppty:t_ctxl_del_cluster} of $\mathcal{S}_k$
    with probability $1-o(1)$.
\end{lemma}

To injectively encode an arbitrary $\mathbf{x}\in\{0,1\}^n$ into $\mathbf{x}'\in\mathcal{S}_k$ in polynomial time we proceed as follows.
Let $U_d$ be uniformly distributed over $\{0,1\}^d$. Then,
    consider 
    $\mathbf{x}'\coloneqq \mathbf{x}+g(U_d)$,
    where the addition operator stands for bit-wise XOR.
Note that the random vector 
    $\boldx'$
    is also $\varepsilon$-almost $\kappa$-wise independent.
Therefore,
    \Cref{lem:t_contextual_almost_ind} implies that $\boldx'\in\mathcal{S}_k$
    with probability $1-o(1)$.

It follows that for each $\boldx$
    there is at least one realization of $U_d$, which we denote by $\mathbf{u}^\star$, 
    such that $\mathbf{x}'=
    \mathbf{x}+g(\mathbf{u}^\star)
    $
    is in $\mathcal{S}_k$.
We can then simply perform brute-force search of all possible $\mathbf{u}\in\{0,1\}^d$ to find the desired $\mathbf{u}^\star$, which takes $2^d\cdot\poly(n)=\poly(n)$ time.

\paragraph{The encoding and decoding procedures}
We are now ready to describe our polynomial-time
    encoding and decoding algorithm for an arbitrary message $\mathbf{x}\in\{0,1\}^n$.
Write $\mathbf{x}'=(x'_1,\ldots,x'_n)$, and recall that $\mathbf{x}'= \mathbf{x}+g(\mathbf{u}^{\star})$.
From our construction,
    we know that if $\mathbf{w}'$ is obtained from $\mathbf{x}'$ via $t$ contextual deletions, then
    the receiver can uniquely recover $\mathbf{x}$ based on $\mathbf{w}'$ as long as the receiver also knows
    $\mathbf{u}^{\star}$, and $\mathrm{syn}(H(\mathbf{x}'))$.
To this end,
    we define $\mathbf{s}_{\textnormal{info}}\coloneqq \mathbf{u}^{\star}\circ\mathrm{syn}(H(\mathbf{x}'))$,
    which is of length $d+6t\log q = 18t(1-C)\log n+((2C+4)\lceil\frac{3}{C}\rceil+4)\log n + o(\log n)$,
    and then add error-correcting redundancy to $\mathbf{s}_{\textnormal{info}}$.
Following the exposition regarding the efficient two-contextual-deletion-correcting code at the end of \Cref{subsec:2contexul_del_efficient}, encoding reduces to:
\begin{align}
    \Enc(\mathbf{x})\coloneqq E(\mathbf{s}_{\textnormal{info}})\circ (1-x'_1) \circ \mathbf{x}',
    \label{eq:t_del_enc}
\end{align}
    where $E$ is the runlength-limited encoder given in \Cref{thm:rll}.
Recalling~\Cref{eq:red-RS-hash}, this encoding adds $|\mathbf{s}_{\textnormal{info}}|+2=18t(1-C)\log n+((2C+4)\lceil\frac{3}{C}\rceil+4)\log n + o(\log n)$ bits of redundancy.
At the same time, the encoding is  efficient, since $E$ can be computed in $O(|\mathbf{s}_{\textnormal{info}}|)=O(t\log n)$ 
time (recall that we have assumed that $t$ is a constant) and
    the generator $g$ is efficient by its definition in~\Cref{thm:almost_ind_generator}.
Furthermore,
    by~\Cref{thm:rll},
    the longest run in $E(\mathbf{s}_{\textnormal{info}})$
    is of length at most
    $\lceil\log |\mathbf{s}_{\textnormal{info}}|\rceil+3=O(\log\log n+\log t)$,
    which is significantly smaller than $k-1=(C+o(1))\log n$.
Therefore, using a similar argument as in \Cref{subsec:2contexul_del_efficient},
    we deduce that
    any contextual deletion in $\Enc(\mathbf{x})$ defined in \Cref{eq:t_del_enc} can only arise in the $\mathbf{x}'$ component.
That is,
    if $\mathbf{y}$ is obtained from $\Enc(\mathbf{x})$ via $t$ contextual deletions,
    then
    $\mathbf{y}$ must take the following form
\begin{align}    \mathbf{y}=E(\mathbf{s}_{\textnormal{info}})\circ (1-x'_1) \circ \mathbf{w}',
    \label{eq:y_t_ctxl_del}
\end{align}
    where $\mathbf{w}'$ is obtained from $\mathbf{x}'$ via $t$ contextual deletions.

Similarly to \Cref{subsec:2contexul_del_efficient},
    the decoding process upon recovering $\mathbf{y}$ in \Cref{eq:y_t_ctxl_del} is straightforward.
We first determine $\mathbf{s}_{\textnormal{info}}=\mathbf{u}^{\star}\circ\mathrm{syn}(H(\mathbf{x}'))$ from $E(\mathbf{s}_{\textnormal{info}})$.
Then,
    by \Cref{lem:tsubstitution},
    we recover $H(\mathbf{x'})$ from $H(\mathbf{w}')$ and $\mathrm{syn}(H(\mathbf{x'}))$.
Next,
    by \Cref{lem:recovers},
    we recover $\mathbf{x}'$ from $\mathbf{w}'$ and $H(\mathbf{x}')$.
Finally,
    we compute $\mathbf{x}'+g(\mathbf{u}^{\star})$
    to recover the original message $\mathbf{x}$.
It can be easily checked that 
    all the decoding steps can be performed in $\poly(n)$ time.

\begin{remark}\label{rmk:non_const_t}
    It is possible to extend the codes in~\Cref{thm:eff-codes-gen-poly-restated} to apply to setting where $t$ grows with $n$ (i.e. $t=\omega(1)$).
In the proofs above,
    we need the condition that $t$ is a constant only in the following steps:
\begin{itemize}
    \item The alphabet size $Q$ of $H(\bolds)$ should be at least $L+6t$.
    \item The $[L+6t,L,6t+1]_q$-RS code has encoding and decoding time complexity $\poly(n)$.
    \item $\lceil \log |\mathbf{s}_{\textnormal{info}}|\rceil+3\leq k-1$,
        where $|\bolds_{\textnormal{info}}|=18t(1-C+O(1))\log n$, so that $\bolds_{\tn{info}}$ can be protected by the RLL encoder against contextual deletions.
\end{itemize}
By carefully examining for which values of $t$ the above still hold, we can extend the parameter range of $(t,k=C\log n)$-contextual deletion-correcting codes for nonconstant $t$. These results will be presented elsewhere.
\end{remark}

\subsection{Proof of \Cref{lem:t_contextual_almost_ind}}\label{sec:t_contextual_almost_ind}

Let $B_2$ be the event that $X$ violates Property~\ref{ppty:t_ctxl_del_no_2logn} (again, think of $B$ describing ``bad events''), 
    and define $B_3$ and $B_4$ similarly for Properties~\ref{ppty:t_ctxl_del_distinct_prefix_suffix} and~\ref{ppty:t_ctxl_del_cluster}, respectively.

We make use of the following property of $\varepsilon$-almost $\kappa$-wise independent random variables.
\begin{proposition}\label{prop:almost_indep_to_indep}
    If $A\subseteq\bits^n$ is an event that only depends on at most $\kappa$ indices,
    then
\begin{align*}
    \mathbb{P}(X \in A )\leq \mathbb{P}(Y\in A)+2^{\kappa}\varepsilon,
\end{align*}
    where $Y\coloneq (Y_1,\ldots,Y_n)$ follows a uniform distribution over $\{0,1\}^n$ and is independent of $X$.
    Here the assumption that $A$ depends on at most $\kappa$ indices means that
    there exist $\tau\in [\kappa]$,
    indices $1\leq i_1<\cdots<i_{\tau}\leq n$,
    and a subset $A'\subseteq \{0,1\}^{\tau}$ such that $X\in A$ if and only if
    $(X_{i_1},\ldots,X_{i_{\tau}}) \in A'$.
\end{proposition}
Although this is well known, for completeness we also provide a proof of the above result.
\begin{proof}[Proof of \Cref{prop:almost_indep_to_indep}]
First,
    extend the index set to $\{i_1,\ldots,i_{\tau},i_{\tau+1},\ldots,i_{\kappa}\}$ for some $i_{\tau+1},\ldots,i_{\kappa}\notin \{i_1,\ldots,i_{\tau}\}$.
Then,
    write the probability of interest as follows:
\begin{align}
    \mathbb{P}(X\in A)
    &=\sum_{\mathbf{a}\in A'}\Prob{(X_{i_1},\ldots,X_{i_{\tau}})=\mathbf{a}}\nonumber\\
    &=\sum_{\mathbf{a}\in A'}\sum_{\mathbf{b}\in\{0,1\}^{\kappa-\tau}}\Prob{(X_{i_1},\ldots,X_{i_{\tau}},X_{i_{\tau+1}},\ldots,X_{i_{\kappa}})=\mathbf{a}\circ\mathbf{b}}\nonumber\\
    &\leq \sum_{\mathbf{a}\in A'}\sum_{\mathbf{b}\in\{0,1\}^{\kappa-\tau}}(\Prob{(Y_{i_1},\ldots,Y_{i_{\tau}},Y_{i_{\tau+1}},\ldots,Y_{i_{\kappa}})=\mathbf{a}\circ\mathbf{b}}+\varepsilon)\label{eq:prop_almost_indep}\\
    &=\sum_{\mathbf{a}\in A'}(\Prob{(Y_{i_1},\ldots,Y_{i_{\tau}})=\mathbf{a}}+2^{\kappa-\tau}\varepsilon)\nonumber\\
    &=\mathbb{P}(Y\in A)+2^{\kappa-\tau}|A'|\varepsilon\nonumber\\
    &\leq \mathbb{P}(Y\in A)+2^{\kappa}\varepsilon,\label{eq:prop_almost_indep_ub}
\end{align}
    where \Cref{eq:prop_almost_indep} follows from the definition of $\varepsilon$-almost $\kappa$-wise independence and the definition of $Y$,
    and \Cref{eq:prop_almost_indep_ub} follows from the trivial bound $|A'|\leq 2^{\tau}$.
\end{proof}
In particular,
    with the choice of $\kappa$ and $\varepsilon$ specified in this lemma,
    Proposition \ref{prop:almost_indep_to_indep} implies that
\begin{align}
    \mathbb{P}(X\in A)
    &\leq \mathbb{P}(Y\in A)+n^{-C\lceil\frac{3}{C}\rceil}\nonumber\\
    &\leq \mathbb{P}(Y\in A)+\frac{1}{n^3}.
    \label{eq:almost_indep_to_indep}
\end{align}

We first claim that for any binary sequence $\mathbf{s}$ whose length is $(1+\delta)\log n,$ for $\delta\in(0,2\lceil\frac{3}{C}\rceil+1),$
    the probability that $X$ contains $\mathbf{s}$ as a substring is $o(1)$.
Fix any length-$((1+\delta)\log n)$ substring of $X$,
    say $(X_i,X_{i+1},\ldots,X_{i+((1+\delta)\log n)-1})$.
The event that this substring equals $\mathbf{s}$ only depends on $(1+\delta)\log n$ coordinates.
Since $(1+\delta)\log n\leq \kappa$,
    by 
    \Cref{eq:almost_indep_to_indep}
    we have
\begin{align*}
    \Prob{(X_i,X_{i+1},\ldots,X_{i+((1+\delta)\log n)-1})=\mathbf{s}}
    &\leq \Prob{(Y_i,Y_{i+1},\ldots,Y_{i+((1+\delta)\log n)-1})=\mathbf{s}}+\frac{1}{n^3}\\
    &=\frac{1}{n^{1+\delta}}+\frac{1}{n^3}.
\end{align*}
By the union bound over all length-$((1+\delta)\log n)$ substrings,
    we have
\begin{equation}
    \Prob{X\textnormal{ contains }\mathbf{s}}
    \leq \frac{1}{n^{\delta}}+\frac{1}{n^2}= o(1).
    \label{eq:almost_indep_no_long_pattern}
\end{equation}

In particular,
    setting $\mathbf{s}=0^{2\log n}$ in \Cref{eq:almost_indep_no_long_pattern}
    gives that
    $X$ has a $0$-run of length at least $2\log n$ with probability $o(1)$.
A similar result holds for $\mathbf{s}=1^{2\log n}$.
This shows that
    $B_2$ happens with probability $o(1)$.

Next we prove that $B_3$ occurs with probability $o(1)$.
Before that,
    we first show that,
    with high probability,
    for any run of length at least $k-1$
    (say $(s_{i+1},\ldots,s_{i+\ell}),$ for some $\ell \geq k-1$),
    we can indeed extract the prefix $(s_{i+\ell+1},\ldots,s_{i+\ell+3(\log n - k)-1})$ and suffix $(s_{i+\ell+2},\ldots,s_{i+\ell+3(\log n - k)})$.
It suffices to show that the last $3\log n - 2k - 1$ bits of $X$ have no $0^{k-1}$ or $1^{k-1}$ runs.
Since this event only depends on $3\log n-2k-1$ bits and $3\log n-2k-1\leq \kappa$,
    by \Cref{eq:almost_indep_to_indep}
    we have that
\begin{align*}
    \Prob{(X_{n-3\log n +2k+2},\ldots,X_n) \textnormal{ contains }0^{k-1} \textnormal{ or }1^{k-1}}
    &\leq \frac{2(3\log n -2k-1)}{2^{k-1}}+\frac{1}{n^3}\\
    &\leq\frac{12\log n}{n^{C+o(1)}}+\frac{1}{n^3}\\
    &=o(1).
\end{align*}

For any run of length at least $k-1$,
    say $(s_{i+1},\ldots,s_{i+\ell})$ for some $\ell \geq k-1$,
    the prefix $(s_{i+\ell+1},\ldots,s_{i+\ell+3(\log n - k)-1})$ and suffix $(s_{i+\ell+2},\ldots,s_{i+\ell+3(\log n - k)})$
    are the same with probability $o(1)$.
Note that this can happen only when $s_{i+\ell+1}=s_{i+\ell+2}=\cdots=s_{i+\ell+3(\log n - k)}$ and $s_{i+\ell}\neq s_{i+\ell+1}$.
In other words, the length-$(3(\log n-k))$ substring following a run of length at least $k-1$ in $X$ has the same length-$(3(\log n-k)-1)$ prefix and suffix 
    if and only if $X$ contains the pattern $0^{k-1}1^{3(\log n-k)}$ or $1^{k-1}0^{3(\log n-k)}$.
However,
    the pattern $0^{k-1}1^{3(\log n-k)}$
    is of length $3\log n - 2k - 1 = (3-2C+o(1))\log n$,
    and 
    we have $3-2C+o(1)\in(1,3)\subseteq(1,2\lceil\frac{3}{C}\rceil)$
    since $C\in(0,1)$.
Therefore,
    \Cref{eq:almost_indep_no_long_pattern} implies
    that $X$ contains $0^{k-1}1^{3(\log n-k)}$ with probability $o(1)$.
Similarly,
    $X$ contains $1^{k-1}0^{3(\log n-k)}$ with probability $o(1)$.
This shows that a run of length at least $k-1$ in $X$ has the same prefix and suffix with probability $o(1)$.

Next, we show that
    the two length-$(3(\log n-k))$ substrings following two different runs of length least $k-1$ have distinct length-$((3\log n-k)-1)$ prefixes and suffixes
    with high probability. Fix any two indices $a$ and $b$ such that $1\leq a < b \leq n$.
Define $B_{a,b}$ to be the event that all three conditions below hold:
\begin{enumerate}
    \item $X_{a-k+2}=\cdots=X_a\neq X_{a+1}$.
    \item $X_{b-k+2}=\cdots=X_b\neq X_{b+1}$.
    \item At least one of the following holds:
    \begin{enumerate}
        \item[(i)] $(X_{a+1},\ldots,X_{a+3(\log n-k)-1})=(X_{b+1},\ldots,X_{b+3(\log n-k)-1})$.
        \item[(ii)] $(X_{a+1},\ldots,X_{a+3(\log n-k)-1})=(X_{b+2},\ldots,X_{b+3(\log n-k)})$.
        \item[(iii)] $(X_{a+2},\ldots,X_{a+3(\log n-k)})=(X_{b+1},\ldots,X_{b+3(\log n-k)-1})$.
        \item[(iv)] $(X_{a+2},\ldots,X_{a+3(\log n-k)})=(X_{b+2},\ldots,X_{b+3(\log n-k)})$.
    \end{enumerate}
\end{enumerate}
Then two length-$(3(\log n-k))$ substrings following two different runs of length least $k-1$ share the same length-$((3\log n-k)-1)$ prefix or suffix
    if and only if there exist $a<b$ such that $B_{a,b}$ holds true.

Note that
    the event $\{X\in B_{a,b}\}$ depends on at most
    $2(3\log n-2k-1)\leq \kappa$ bits.
        Thus, by \Cref{prop:almost_indep_to_indep}
            we have
        \begin{align}
            \Prob{X\in B_{a,b}}\leq \Prob{Y\in B_{a,b}}+ \frac{1}{n^3}.
            \label{eq:almost_indep_ppty3_X_Y}
        \end{align}
To bound $\Prob{Y\in B_{a,b}}$,
    we split our analysis based on the value of $b-a$.
\begin{enumerate}
    \item
        Consider the case $b-a>3\log n - 2k-2$.
        Since $a+3(\log n -k)<b-k+2$,
            it follows that
        \begin{align}
        \Prob{Y\in B_{a,b}}
        &=
        \Prob{Y_{a-k+2}=\cdots=Y_a\neq Y_{a+1} \wedge Y_{b-k+2}=\cdots=Y_b\neq Y_{b+1} \wedge \textnormal{Condition} (3)}\nonumber\\
        &\leq \Prob{Y_{a-k+2}=\cdots=Y_a \wedge Y_{b-k+2}=\cdots=Y_b \wedge \textnormal{Condition} (3)}\nonumber\\
        &=\Prob{Y_{a-k+2}=\cdots=Y_a} \Prob{Y_{b-k+2}=\cdots=Y_b} \Prob{\textnormal{Condition} (3)}\nonumber\\
        &\leq 4\cdot 2^{-2(k-2)}2^{-3(\log n-k)+1}\nonumber\\
        &=2^{(-3+C+o(1))\log n}\nonumber\\
        &=\frac{1}{n^{3-C-o(1)}}.
        \label{eq:almost_indep_ppty3_Y_ab_far}
        \end{align}
        For $b-a>3\log n -2k-2$, combining \Cref{eq:almost_indep_ppty3_X_Y,eq:almost_indep_ppty3_Y_ab_far} yields
        \begin{align}
            \Prob{X\in B_{a,b}}&\leq \frac{1}{n^{3-C-o(1)}}+\frac{1}{n^3}\nonumber\\
            &=\frac{1}{n^{3-C-o(1)}},
            \label{eq:almost_indep_X_ab_far}
        \end{align}
            where we used the fact that $\frac{1}{n^3}$ is negligible compared to $\frac{1}{n^{3-C-o(1)}}.$
    \item 
        Consider the case $k-1\leq b-a\leq 3\log n - 2k-2$.
        We can simply write 
        \begin{align}
        \Prob{Y\in B_{a,b}}
        &=
        \Prob{Y_{a-k+2}=\cdots=Y_a\neq Y_{a+1} \wedge Y_{b-k+2}=\cdots=Y_b\neq Y_{b+1} \wedge \textnormal{Condition} (3)}\nonumber\\
        &\leq \Prob{Y_{a-k+2}=\cdots=Y_a \wedge \textnormal{Condition} (3)}\nonumber\\
        &=\Prob{Y_{a-k+2}=\cdots=Y_a} \Prob{\textnormal{Condition} (3)}\nonumber\\
        &\leq 4\cdot 2^{-(k-2)}2^{-3(\log n-k)+1}\label{eq:equal_substring_overlap}\\
        &=2^{(-3+2C+o(1))\log n}\nonumber\\
        &=\frac{1}{n^{3-2C-o(1)}},
        \label{eq:almost_indep_ppty3_Y_ab_close}
        \end{align}
            where in \Cref{eq:equal_substring_overlap}
            we used the fact that
            for any two substrings of length $l$ in $Y$,
            the probability that they are the same is $2^{-l}$ regardless of whether they overlap or not (see, e.g., the proof of~\cite[Theorem 14]{CJLW22}).
        For $b-a\leq 3\log n -2k-1$, combining \Cref{eq:almost_indep_ppty3_X_Y} and \Cref{eq:almost_indep_ppty3_Y_ab_close} yields
        \begin{align}
            \Prob{X\in B_{a,b}}&\leq \frac{1}{n^{3-2C-o(1)}}+\frac{1}{n^3}\nonumber\\
            &=\frac{1}{n^{3-2C-o(1)}}.
            \label{eq:almost_indep_X_ab_close}
        \end{align}
\end{enumerate}
Note that since we require that $a$ and $b$ are the ends of two different runs of length at least $k-1$,
    we only have to consider $b-a\geq k-1$.
Therefore,
    the above two cases include all possibilities of $(a,b)$. 
Now,
    by \Cref{eq:almost_indep_X_ab_far},
    \Cref{eq:almost_indep_X_ab_close}, and
    the union bound,
    the probability that $X$ has two different runs of length at least $k-1$
    followed by two substrings with the same prefix of suffix 
    is upper bounded by
\begin{align}
    \sum_{1\leq a<b\leq n}\Prob{X\in B_{a,b}}
    &=\sum_{b-a > 3\log n - 2k - 2}\Prob{X\in B_{a,b}}
    +\sum_{k-1\leq b-a \leq 3\log n - 2k - 2}\Prob{X\in B_{a,b}}\nonumber\\
    &\leq  \sum_{b-a > 3\log n - 2k - 2} \frac{1}{n^{3-C-o(1)}} + \sum_{k-1\leq b-a \leq 3\log n - 2k - 2}\frac{1}{n^{3-2C-o(1)}}\nonumber\\
    &\leq \frac{n^2}{n^{3-C-o(1)}}+\frac{n(3\log n -2k -2)}{n^{3-2C-o(1)}},
    \label{eq:almost_indep_X_ab_all}\\
    &=o(1).\nonumber
\end{align}
    where in \Cref{eq:almost_indep_X_ab_all} we used the fact that there are at most $n^2$ pairs $(a,b)$ such that $b-a > 3\log n - 2k - 2$
    and there are at most $n(3\log n -2k -2)$ pairs $(a,b)$ such that $k-1\leq b-a \leq 3\log n -2k -2$.
These arguments show that $X\in B_3$ with probability $o(1)$.

Instead of showing that $X\in B_4$ occurs with probability $o(1)$,
    we show that $X\in B_2^C\cap B_4$ with probability $o(1)$,
    where $B_2^C$ denotes the complement of $B_2$.
The intuition behind this approach is that 
    $B_4$ is not necessarily a ``local'' event (since $\ell_j$ may be too large),
    so we intersect it with $B_2^C$
    to make the event become ``local''.
To this end, define $B'$ to be the following event:
\begin{itemize}
    \item There exist $\lceil\frac{3}{C}\rceil+1$ substrings $(X_{i_1+1},\ldots, X_{i_1+\ell_1})$, $(X_{i_2+1},\ldots,X_{i_2+\ell_2}),\ldots,(X_{i_{\lceil\frac{3}{C}\rceil+1}+1},\ldots,X_{i_{\lceil\frac{3}{C}\rceil+1}+\ell_{\lceil\frac{3}{C}\rceil+1}})$, each of length $\ell_j\in [k,2\log n-1]$ for $j\in [\lceil\frac{3}{C}\rceil+1]$, such that
    both of the following holds:
    \begin{enumerate}
        \item $i_{j}+\ell_j\leq i_{j+1}\le  i_{j}+\ell_j+3(\log n -k)$ for each $j\in[\lceil\frac{3}{C}\rceil]$.
        \item $X_{i_j+1}=\cdots X_{i_j+\ell_j}\neq X_{i_j+\ell_j+1}$ for each $j\in[\lceil\frac{3}{C}\rceil+1]$.
    \end{enumerate}
\end{itemize}
Note that we relaxed the requirement that $(X_{i_1+1},\ldots, X_{i_1+\ell_1})$ is a complete run,
    since we now do not require that $X_{i_1}\neq X_{i_1+1}$.
It turns out that this relaxation can simplify the subsequent computations.
At this point,
    it is clear that $B_2^C\cap B_4\subseteq B$,
    and thus we have
\begin{align}    \Prob{X\in B_2^C\cap B_4}\leq \Prob{X\in B'}.
    \label{eq:P_X_B2C_B4_ub}
\end{align}

We further decompose $B'$ as follows:
For each vector $\mathbf{d}\coloneqq(d_1,d_2,\ldots,d_{\lceil\frac{3}{C}\rceil})\in[0,3(\log n-k)]^{\lceil\frac{3}{C}\rceil}$ define $B'_{\mathbf{d}}$ to be the following event:
\begin{itemize}
   \item There exist $\lceil\frac{3}{C}\rceil+1$ substrings $(X_{i_1+1},\ldots, X_{i_1+\ell_1})$, $(X_{i_2+1},\ldots,X_{i_2+\ell_2}),\ldots,(X_{i_{\lceil\frac{3}{C}\rceil+1}+1},\ldots,X_{i_{\lceil\frac{3}{C}\rceil+1}+\ell_{\lceil\frac{3}{C}\rceil+1}})$, each of length $\ell_j\in [k,2\log n-1]$ for $j\in [\lceil\frac{3}{C}\rceil+1]$, such that
    both of the following holds:
    \begin{enumerate}
        \item $i_{j+1}=i_j+\ell_j+d_j$ for each $j\in[\lceil\frac{3}{C}\rceil]$.
        \item $X_{i_j+1}=\cdots X_{i_j+\ell_j}\neq X_{i_j+\ell_j+1}$ for each $j\in[\lceil\frac{3}{C}\rceil+1]$.
    \end{enumerate}
\end{itemize}
Then we have 
\begin{align*}    B'=\bigcup_{\mathbf{d}\in[0,3(\log n-k)]^{\lceil\frac{3}{C}\rceil}} B'_{\mathbf{d}},
\end{align*}
    and thus by the union bound we have
\begin{align}
    \Prob{X\in B'}\leq\sum_{\mathbf{d}\in[0,3(\log n-k)]^{\lceil\frac{3}{C}\rceil}}\Prob{X\in B'_{\mathbf{d}}}.
    \label{eq:P_X_Bp_ub}
\end{align}

We now bound $\Prob{X\in B'_{\mathbf{d}}}$ for an arbitrary $\mathbf{d}\in[0,3(\log n-k)]^{\lceil\frac{3}{C}\rceil}$.
Fix any starting index $m\in [n]$ and consider
$B'_{\mathbf{d};m}$ to be the intersection of $B'_{\mathbf{d}}$ and $\{i_1=m\}$.
To be more precise,
    $B'_{\mathbf{d};m}$ is the event such that both of the following hold:
\begin{itemize}
    \item $X_{i_1+1}=\cdots=X_{i_1+\ell_1}\neq X_{i_1+\ell_1+1},$ for some $\ell_1\in[k, 2\log n-1]$,
        where $i_1\coloneqq m$.
    \item For each $j\in\{2,\dots,\lceil 3/C\rceil +1\}$, we have
        $X_{i_j+1}=\cdots=X_{i_j+\ell_j}\neq X_{i_j+\ell_j}$ for some $\ell_j\in[k,2\log n-1]$,
        where $i_j\coloneqq i_{j-1}+\ell_{j-1}+d_{j-1}$.
\end{itemize}
The probability of $B'_{\mathbf{d};m}$ can be upper bounded 
    as follows:
Note that for any given $\mathbf{v}\coloneq(v_1,\ldots,v_{\lceil\frac{3}{C}\rceil+1})\in[k,2\log n-1]^{\lceil\frac{3}{C}\rceil+1}$,
    the event $B'_{\mathbf{d};m}\cap \{\ell_1=v_1\}\cap\cdots\cap\{\ell_{\lceil\frac{3}{C}\rceil+1}=v_{\lceil\frac{3}{C}\rceil+1}\}$
    (i.e., the event that the $j$-th run is of length exactly $v_j$)
    depends only on at most $\sum_{j=1}^{\lceil\frac{3}{C}\rceil+1}(v_j+1)$ bits,
    and
    $\sum_{j=1}^{\lceil\frac{3}{C}\rceil+1}(v_j+1)\leq 2(\lceil\frac{3}{C}\rceil+1)\log n\leq \kappa$.
Therefore,
    \Cref{eq:almost_indep_to_indep} yields
\begin{align}
    \Prob{X\in B'_{\mathbf{d};m}\cap \{\ell_1=v_1\}\cap\cdots\cap\{\ell_{\lceil\frac{3}{C}\rceil+1}=v_{\lceil\frac{3}{C}\rceil+1}\}}
    &\leq \Prob{Y\in B'_{\mathbf{d};m}\cap \{\ell_1=v_1\}\cap\cdots\cap\{\ell_{\lceil\frac{3}{C}\rceil+1}=v_{\lceil\frac{3}{C}\rceil+1}\}}+\frac{1}{n^3}.
    \label{eq:cluster_runlength_given_almost_indep}
\end{align}
Then,
    note that we have
\begin{align}
    \Prob{Y\in B'_{\mathbf{d};m}\cap \{\ell_1=v_1\}\cap\cdots\cap\{\ell_{\lceil\frac{3}{C}\rceil+1}=v_{\lceil\frac{3}{C}\rceil+1}\}}
    \leq 2^{-v_1}2^{-v_2}\cdots 2^{-v_{\lceil\frac{3}{C}\rceil+1}},\label{eq:cluster_runlength_given_Y_ub}
\end{align}
    which is a consequence of the following argument. 
If $\mathbf{d}$, $m$ and $\mathbf{v}$
    are selected such that 
    $m+\sum_j d_j + \sum_j v_j > n$,
    then the probability on the left-hand side of \Cref{eq:cluster_runlength_given_Y_ub} equals zero (since it depends on some ``out-of-bound'' random variables $Y_{n+1},Y_{n+2},\ldots$).
Otherwise,
    we actually have an equality in \Cref{eq:cluster_runlength_given_Y_ub}:
If $d_j\geq 1$ for all $j$,
    then every (incomplete) run of interest is separated by at least one bit,
    and thus equality in \Cref{eq:cluster_runlength_given_Y_ub} holds by independence of the bits of $Y$. Even if $d_j=0$ for some $j$ (i.e., two runs are adjacent), equality still holds.
For the sake of exposition, we demonstrate this idea using the simple case where there are only two runs,
    and the general result are a consequence of a straightforward extension.
The probability that
$Y_{m+1}=\cdots=Y_{m+v_1}\neq Y_{m+v_1+1}$ and $Y_{m+v_1+1}=\cdots Y_{m+v_1+v_2}\neq Y_{m+v_1+v_2+1}$ can be expressed as
\begin{align*}
    \Prob{Y_{m+1}=\cdots=Y_{m+v_1}=1-Y_{m+v_1+1}=\cdots=1-Y_{m+v_1+v_2}=Y_{m+v_1+v_2+1}},
\end{align*}
    which is $2\cdot 2^{-(v_1+v_2+1)}=2^{-(v_1+v_2)},$ obtained 
    by considering the two cases $Y_{m+1}=0$ and $Y_{m+1}=1$.
It now follows from \Cref{eq:cluster_runlength_given_almost_indep} and \Cref{eq:cluster_runlength_given_Y_ub} that
\begin{align}
    \Prob{X\in B'_{\mathbf{d};m}}
    &\leq \sum_{v_1,\ldots,v_{\lceil\frac{3}{C}\rceil+1}=k}^{2\log n-1} (2^{-v_1}\cdots 2^{-v_{\lceil\frac{3}{C}\rceil+1}}+\frac{1}{n^3})\nonumber\\
    &=(2 \log n - k )^{\lceil\frac{3}{C}\rceil+1}\frac{1}{n^3}+\sum_{v_1,\ldots,v_{\lceil\frac{3}{C}\rceil+1}=k}^{2\log n-1} 2^{-v_1}\cdots 2^{-v_{\lceil\frac{3}{C}\rceil+1}}\nonumber\\
    &\leq (2 \log n)^{\lceil\frac{3}{C}\rceil+1}\frac{1}{n^3}+\sum_{v_1,\ldots,v_{\lceil\frac{3}{C}\rceil+1}=k}^{\infty} 2^{-v_1}\cdots 2^{-v_{\lceil\frac{3}{C}\rceil+1}}\nonumber\\
    &=(2 \log n)^{\lceil\frac{3}{C}\rceil+1}\frac{1}{n^3}+(2^{-k+1})^{\lceil\frac{3}{C}\rceil+1}\nonumber\\
    &=(2 \log n)^{\lceil\frac{3}{C}\rceil+1}\frac{1}{n^3}+\frac{1}{n^{C\lceil\frac{3}{C}\rceil+C+o(1)}}\nonumber\\
    &\leq 2(2 \log n)^{\lceil\frac{3}{C}\rceil+1}\frac{1}{n^3},
    \label{eq:P_X_Bdm_ub}
\end{align}
    where we used the fact that $\frac{1}{n^{C\lceil\frac{3}{C}\rceil+C+o(1)}}$ is negligible compared to $\frac{(2 \log n)^{\lceil\frac{3}{C}\rceil+1}}{n^3}$. From \Cref{eq:P_X_Bdm_ub} and a union bound over $m\in[n]$
    we have
\begin{equation}
    \Prob{X\in B'_{\mathbf{d}}}
    \leq 2n(2 \log n)^{\lceil\frac{3}{C}\rceil+1}\frac{1}{n^3}=2(2 \log n)^{\lceil\frac{3}{C}\rceil+1}\frac{1}{n^2}.
    \label{eq:P_X_Bd_ub}
\end{equation}

Finally,
    using \Cref{eq:P_X_B2C_B4_ub,eq:P_X_Bp_ub,eq:P_X_Bd_ub},
    we arrive at
\begin{equation}
    \Prob{X\in B_2^C\cap B_4}
    \leq 2(3(\log n -k)+1)^{\lceil\frac{3}{C}\rceil} (2 \log n)^{\lceil\frac{3}{C}\rceil+1}\frac{1}{n^2}=o(1).\label{eq:P_X_B2C_B4_ub_result}
\end{equation}
We are now ready to conclude the argument.
From \Cref{eq:P_X_B2C_B4_ub_result} and our previous arguments,
    the probability that $X$ fails at least one of the Properties~\ref{ppty:t_ctxl_del_no_2logn},~\ref{ppty:t_ctxl_del_distinct_prefix_suffix}, or~\ref{ppty:t_ctxl_del_cluster} of $\mathcal{S}_k$
    is at most
\begin{align}
    \Prob{X\in B_2\cup B_3\cup B_4}
    &\leq\Prob{X\in B_2\cup B_4}+\Prob{B_3}\nonumber\\
    &=\Prob{X\in B_2}+\Prob{X\in B_2^C\cap B_4}+o(1)\label{eq:set_identity}\\
    &=o(1)+o(1)+o(1)\nonumber\\
    &=o(1)\nonumber,
\end{align}
    where \Cref{eq:set_identity} follows from the fact that $B_2\cup B_4$ is the disjoint union of $B_2$ and $B_2^C\cap B_4$.

    This concludes the proof of \Cref{lem:t_contextual_almost_ind}.

\section{Capacity bounds for the extremal contextual deletion channel} \label{sec:capacity}

We now turn our attention to studying the coding capacity of the extremal contextual deletion channel (i.e., the asymptotic rate of the largest zero-error code for this channel). Our results are  summarized in \Cref{thm:capacity_p1}.
We prove the lower bound in \Cref{sec:cap-lb} and the upper bound in \Cref{sec:cap-ub}.

Recall that we denote the extremal contextual deletion channel with threshold $k$ by $\fD_{k,1}$.
Furthermore, for a set $S\subseteq \bits^n$, we denote by $\fD_{k,1}(S)$ the set of all strings obtained by sending strings of $S$ through $\fD_{k,1}$.
Then, because the behavior of the channel is deterministic given the input, it is not hard to see that the size of the largest code of block length $n$ with vanishing decoding error probability on $\fD_{k,1}$ is $|\fD_{k,1}(\bits^n)|$, and in fact the resulting code is zero-error.
In other words, the rate of the largest zero-error code of block length $n$ is $\frac{1}{n}\log |\fD_{k,1}(\bits^n)|$, and the coding capacity is
\begin{equation*}
    C_k = \lim_{n\to\infty} \frac{1}{n}\log |\fD_{k,1}(\bits^n)|.
\end{equation*}

\subsection{Capacity lower bound}\label{sec:cap-lb}

Fix an arbitrary threshold $k\geq 2$.
To obtain a lower bound on $C_k$ it suffices to find a sequence of subsets $\cA_1, \cA_2,\dots$ with $\cA_n\subseteq\bits^n$ for which we can compute (or, at least, lower bound)
\begin{equation*}
    \liminf_{n\to\infty}\frac{1}{n}\log|\cA_n|,
\end{equation*}
and such that $\fD_{k,1}$ (seen as a map) is injective on $\cA_n$.
This ensures that $\cA_n$ is a zero-error code for $\fD_{k,1}$ (in particular, $|\fD_{k,1}(\cA_n)|=|\cA_n|$).

For convenience, we first recall the definition of the sets $\cH_n$ from \Cref{thm:capacity_p1}, which we will henceforth focus on.
Define
    \begin{equation}\label{eq:def-E0}
    \mathcal{E}_0 \coloneqq \{0^k100,0^k1010,\ldots,0^k101^{k-2}0,0^k101^k\},
\end{equation}
    and let $\mathcal{E}_1$ denote the sets of bit-wise complements of strings in $\mathcal{E}_0$.
Then define
\begin{equation}\label{eq:def-E}
\mathcal{E}=\mathcal{E}_0\cup \mathcal{E}_1.
\end{equation}
It follows that $\mathcal{H}_n$ is the collection of length-$n$ binary sequences
    that contain no substrings from $\mathcal{E}\cup\{0^{k+1}1^k00,1^{k+1}0^k11\}$.

We will show that $\fD_{k,1}$ is injective on subsets $\cH'_n\subseteq \cH_n,$ satisfying $\log \cH'_n=(1-o(1))\log\cH_n$. Also, it is clear that
\begin{equation*}
    C_k \geq \liminf_{n\to\infty} \frac{1}{n}\log|\cH_n|.
\end{equation*}
As mentioned before, the right-hand side quantity can be computed based on enumeration techniques presented in Appendix~\ref{app:0k10}.

We now define the relevant subsets $\cH'_n$.
\begin{definition}[Structured subset of $\cH_n$]\label{def:H_n}
We take $\cH'_n$
    to be the collection of sequences in $\mathcal{H}_n$
    for which every run of length at least $k$ is followed by either:
\begin{enumerate}
    \item A run of length at least $2$, or
    \item $101^{k-1}0$ if it is a $0$-run, or $010^{k-1}1$ if it is a $1$-run.
\end{enumerate}
\end{definition}
The main idea is that for every sequence in $\mathcal{H}_n,$
    each substring $0^k10$
    is followed by either $1^{k-1}0$
    or some prefix of $1^{k-1}0$,
    the latter of which is only allowed when that prefix is at the end of the sequence.
On the other hand,
    the additional requirements imposed on $\cH'_n$
    make sure that every occurrence of $0^k10$
    is followed by $1^{k-1}0$.
We will show that
    these additional constraints ensure that $\fD_{k,1}$ is injective on $\cH'_n$.
But before that,
    we first show that the sizes of $\mathcal{H}_n$ and $\cH'_n$ are close, as characterized by the following lemma.
\begin{lemma}\label{lem:H_o_n_size}
We have
\begin{align}
    |\mathcal{H}_{n-k-2}| \leq |\cH'_n| \leq |\mathcal{H}_n|.
    \label{eq:H_o_n_size}
\end{align}
In particular, this means that $\log |\cH'_n|=(1-o(1))\log |\cH_n|$.
\end{lemma}
\begin{IEEEproof}
The right-hand side inequality in \Cref{eq:H_o_n_size} easily follows from the fact that
$\cH'_n\subseteq\mathcal{H}_n$.

To establish the inequality on the left-hand side,
    we fix an arbitrary string 
     $\mathbf{x}\in\mathcal{H}_{n-k-2}$ and argue that
    we can pad the sequence with $k+2$ bits at its end to get a sequence in $\cH'_{n}$.
More precisely,
    we first perform padding with at most $k+2$ bits to satisfy one of the constraints,
    and then simply add an alternating string   
    ($0101\dots$ if the sequence ends with $1$
    or
    $1010\dots$ if the sequence ends with $0$)
    to ensure that the length equals $n$.
   
More precisely, the construction proceeds as follows:
Write $\mathbf{x}=r_1\ldots r_S$,
    where $r_1,\ldots,r_S$ are (complete) runs.
Let $r_j$ be the last run of length at least $k$.
If $r_j$ is followed by at least four runs
    (i.e. $j\leq S-4$),
    then the forbidden patterns in $\mathcal{H}_n$ already imply that either Constraint (1) or (2) has to hold.
Now we split our analysis based on the value of $j$ and assume without loss of generality that $r_j$ is a $0$-run:
\begin{enumerate}
    \item 
        If $j=S$,
            then we pad $101^{k-1}0$.
        We check that this padding will not introduce any forbidden pattern from $\mathcal{E}_0\cup\mathcal{E}_1\cup\{0^{k+1}1^k00,1^{k+1}0^k11\}$.
        Note that a direct comparison shows that
            $101^{k-1}0\notin\mathcal{E}_0\cup\mathcal{E}_1\cup\{0^{k+1}1^k00,1^{k+1}0^k11\}$.
        Therefore,
            if padding with $101^{k-1}0$ introduces a forbidden pattern,
            then that pattern must lie across $\mathbf{x}$ and $101^{k-1}0$. More precisely, that forbidden pattern $\mathbf{s}$ must be decomposable as $\mathbf{s}=\mathbf{s}_1\circ\mathbf{s}_2,$ where $\mathbf{s}_1$ is a nonempty suffix of $\mathbf{x}$ and $\mathbf{s}_2$ is a nonempty prefix of $101^{k-1}0$.
        We then proceed with the next steps:
        \begin{enumerate}
            \item [(i)] We check that this padding does not introduce any pattern $\mathbf{s}$ from $\mathcal{E}_0$:
            The only possible way for $\mathbf{s}$ to lie across $\mathbf{x}$ and $101^{k-1}0$ is that the $0^k$ prefix of $\mathbf{s}$ aligns with the last $k$ bits of $\mathbf{x}$.
            However,
                the possible suffixes of $\mathbf{s}$ are then $100,1010,\ldots,101^{k-1}0$, and $101^k$,
                none of which is a prefix of $101^{k-1}0$.
            \item [(ii)]
                We check that this padding will not introduce any pattern $\mathbf{s}$ from $\mathcal{E}_1$:
                Since $\mathbf{x}$ ends with a $0$-run of length at least $k$,
                    if a forbidden pattern $\mathbf{s}\in\mathcal{E}$ lies across $\mathbf{x}$ and $101^{k-1}0$, it must contain $0^k$ in the middle.
                However,
                    none of the strings in $\mathcal{E}_1$ have this property.
            \item [(iii)]
                 We check that this padding will not introduce  $\mathbf{s}=0^{k+1}1^k00$:
                 The only possible way for $\mathbf{s}$ to lie across $\mathbf{x}$ and $101^{k-1}0$ is that
                 $|r_{j-1}|\geq k+1$
                 and for the prefix $0^{k+1}$ to align with the last $k+1$ bits of $\mathbf{x}$.
                 However,
                    the suffix $1^k00$ of $\mathbf{s}$ is not a prefix of $101^{k-1}$.
            \item [(iv)]
                 We check that this padding will not introduce  $\mathbf{s}=1^{k+1}0^k11$:
                 The only possible way for $\mathbf{s}$ to lie across $\mathbf{x}$ and $101^{k-1}0$ is that
                 $|r_{j-1}|\geq k+1$, $|r_j|=k$,
                 and for the prefix $1^{k+1}0^k$ to align with the last $2k+1$ bits of $\mathbf{x}$.
                 However,
                    the suffix $11$ of $\mathbf{s}$ is not a prefix of $101^{k-1}$.
        \end{enumerate}
        These arguments show that padding with $101^{k-1}0$ will not introduce any forbidden patterns in  the set $\mathcal{E}_0\cup\mathcal{E}_1\cup\{0^{k+1}1^k00,1^{k+1}0^k11\}$. Similar approaches may be used to establish the remaining cases -- the details are omitted for simplicity of exposition.       
    \item 
        If $j=S-1$,
            then $|r_S|\in[1,k-1]$
            (since $r_{S-1}$ is the last run of length at least $k$).
        We then consider the following two cases:
        \begin{enumerate}
            \item [(i)]  If $|r_S|=1$ (i.e. $r_S=1$),
                    then we pad
                    $01^{k-1}0$.
                Again,
                    even if $|r_{S-2}|\geq k+1$
                    and $|r_{S-1}|=k$,
                    we will not introduce a forbidden pattern
                    $1^{k+1}0^k11$.
            \item [(ii)]
                If $|r_S|\in [2,k-1]$,
                    then we do not perform bit-padding at this stage.
        \end{enumerate}
    \item 
        If $j=S-2$,
            then we must have $|r_{S-1}|,|r_S|\in[1,k-1]$ and we have to consider two cases:
        \begin{enumerate}
            \item [(i)]
                If $|r_{S-1}|=1$,
                    then by the forbidden-pattern constraint ($0^k100$)
                    we must have $|r_S|=1$.
                That is,
                    $r_{S-2}r_{S-1}r_S$
                    end with
                    $0^k10$.
                We then pad $1^{k-1}0$.
            \item [(ii)]
                If $|r_{S-1}|\in [2,k-1]$,
                    then we do not have to perform any checks.
        \end{enumerate}
    \item 
        If $j=S-3$,
            then we must have $|r_{S-2}|,|r_{S-1}|,|r_S|\in[1,k-1]$, and once again, we consider two cases:
        \begin{enumerate}
            \item [(i)]
                If $|r_{S-2}|=1$,
                    then similarly
                    we must have $|r_{S-1}|=1$.
                Then $r_{S-3}r_{S-2}r_{S-1}r_S$
                    ends with
                    $0^k101^{|r_S|}$.
                We then pad $1^{k-1-|r_S|}0$.
            \item [(ii)]
                If $|r_{S-2}|\in [2,k-1]$, then we do not have to perform any checks.
        \end{enumerate}
\end{enumerate}

Finally, note that trimming the last $k+2$ bits of any string $\boldx\in\cH_n$ yields a string in $\cH'_{n-k-2}$.
Since there are $2^{k+2}$ possible trimmed suffixes, we have $|\cH'_{n-k-2}|\geq \frac{|\cH_n|}{2^{k+2}}$, and so $\log |\cH'_n|\geq (1-o(1))\log|\cH_n|$ because $k$ is constant.
\end{IEEEproof}
\begin{remark}
For the case $k=2$,
    the proof of \Cref{lem:H_o_n_size} can be simplified as follows:
For each sequence in $\mathcal{H}_{n-4}$,
    we pad four alternating bits at its end
    to get a sequence in $\cH'_n$.
\end{remark}

We now prove injectivity of $\fD_{k,1}$ on $\cH'_n$.
\begin{lemma}\label{lem:H_o_n_UD}
$\fD_{k,1}$ is injective on $\cH'_n$.
That is,
    for each $\mathbf{x}\in\cH'_n$,
    we can uniquely recover $\mathbf{x}$ from $\mathbf{y}=\fD_{k,1}(\mathbf{x})$.
\end{lemma}
\begin{IEEEproof}[Proof of \Cref{lem:H_o_n_UD}]
We proceed to prove an actually stronger statement that
$\fD_{k,1}$ is injective on $\cH'\coloneqq \bigcup_{n=1}^{\infty}\cH'_n$.

Note that decoding can be performed by scanning the output from left to right,
    since contextual deletions can be seen as being applied sequentially from right to left. See the proof of \Cref{thm:GV}
    for a rigorous characterization for this property. 
Therefore,
    it suffices to show that when scanning an output $\mathbf{y}\in\fD_{k,1}(\cH')$ from left to right and encountering a run of length at least $k$, there is only one way to add back the deleted bit following that run.

Write $\mathbf{y}=r_1\cdots r_S$,
    where $r_1,\ldots,r_S$ denote complete runs. Let $r_i$ be the first (leftmost) run of length at least $k$ in $\mathbf{y}$.
Without loss of generality, assume $r_i=0^{\gamma}$.
There must have been at least one contextual deletion caused by $r_i$,
    so we must add at least one $1$ bit back somewhere after the first $k$ bits of $r_i$.
As mentioned before,
    we can do the decoding from left to right,
    so we first have to decide where is the first place to add back a $1$-bit.
We split our discussion based on the value of 
    $\gamma$,
    which is the length of $r_i$:
\begin{enumerate}
    \item 
        If $|r_i|=k$,
            then the only option we have is to add a $1$ at the end of $r_i$.
    \item 
        If $|r_i|\geq k+1$, we first show that
            the only place to add a $1$ back 
            is either at the end of $r_i$
            (i.e., replace $r_i=0^{\gamma}$ with $0^{\gamma}1$), 
            or, 
            one bit away from it's end
            (i.e., replace $r_i=0^{\gamma}$ with $0^{\gamma-1}10$).
        The reason is that
            if we
            replace $r_i=0^{\gamma}$ by 
            $0^{\gamma-\gamma'}10^{\gamma'}$
            for some $\gamma'\in[2,\gamma-k]$,
            then we introduce the forbidden pattern $0^k100$
            no matter how the subsequence following (and including) the added bit is decoded
            (this subsequence starts with $100$ and there cannot be a deletion for these three bits).
        We also claim that  
            $r_i$ cannot be the last run of $\mathbf{y}$.
        Assume for contradiction that $r_i$ is the last run of $\mathbf{y}$.
        Then there are only two possible ways in which we can add a $1$ back:
        \begin{enumerate}
            \item [(i)]
                If we add a $1$ at the end of $r_i$,
                    then $r_i$ is followed by a single-bit run.
                At the same time, $r_i$ is not followed by $101^{k-1}0$.
                This contradicts the additional condition imposed on $\cH'_n$.
            \item [(ii)]
                If we replace $r_i=0^{\gamma}$ by 
                    $0^{\gamma-1}10$,
                    then $r_i$ is also followed by a single-bit run.
                At the same time we still have that
                    $r_i$ is not followed by
                    $101^{k-1}0$,
                    which also leads to a contradiction.
        \end{enumerate}
        These arguments show that $r_i$ cannot be the last run of $\mathbf{y}$.
        We therefore proceed to analyze the next possible scenario, this time based on the length of $r_{i+1}$:
        \begin{enumerate}
            \item 
                If $|r_{i+1}|=k-1$,
                    then we can only replace $r_i=0^{\gamma}$ by  $0^{\gamma-1}10$.
                If we replace $r_i=0^{\gamma}$ by $0^{\gamma}1$,
                    then $r_{i+1}$ becomes $1^k$,
                    which will introduce $1^k0$ in the next step of decoding.
                We further divide our discussion based on whether $r_{i+1}$ is the last run of $\mathbf{y}$ or not:
            \begin{enumerate}
                \item [(i)]  
                    If $r_{i+1}$ is the last run of $\mathbf{y}$,
                        then the decoded output will end with
                        $1^k0$,
                        which contradicts the definition of $\cH'_n$.
                \item [(ii)]
                    If $r_{i+1}$ is not the last run of $\mathbf{y}$,
                        then the decoded output will contain
                        $0^{k+1}1^k00$,
                        which is a forbidden pattern in $\mathcal{H}_n$.
            \end{enumerate}
            \item 
                If $|r_{i+1}|\geq k$,
                    the only option is to add a $1$ at the end of $r_i$.
                In any other case, we introduce the forbidden pattern
                    $0^k101^k$.
            \item 
                If $|r_{i+1}|\in[1,k-2]$
                    (say, $|r_{i+1}|=\zeta$),
                    then we still can only add $1$ at the end of $r_i$.
                Otherwise we either introduce the forbidden pattern
                    $0^k101^\zeta0$
                    when $r_{i+1}$ is not the last run,
                    or we
                    end with an incomplete pattern
                    $0^k101^{\zeta}$
                    (in terms of Constraint (2)
                    in the definition of $\cH'_n$),
                    when $r_{i+1}$ is the last run.
        \end{enumerate}
\end{enumerate}
The same arguments apply every time we encounter a run of length at least $k$ in $\boldy$ during the decoding process.
\end{IEEEproof}

Combining \Cref{lem:H_o_n_size,lem:H_o_n_UD} yields to the lower bound of \Cref{thm:capacity_p1}.

\subsection{Capacity upper bound}
\label{sec:cap-ub}

We now present an upper bound
    by considering another forbidden pattern set, thereby completing the proof of \Cref{thm:capacity_p1}.
For convenience, we recall the definition of the sets $\cJ_n$ from \Cref{thm:capacity_p1}. Let
\begin{equation*}
    \mathcal{F}_0\coloneqq \{0^{k+1}1^k001,0^{k+1}1^k0001,\ldots,0^{k+1}1^k0^{k-1}1,0^{k+1}1^k0^{k+1}\},
\end{equation*}
    and let $\mathcal{F}_1$ denote the set of bitwise complements of strings in $\mathcal{F}_0$.
    Define $\mathcal{F}=\mathcal{F}_0\cup \mathcal{F}_1$. 
    Then, $\cJ_n$ is the set of length-$n$ binary strings that do not have substrings from $\mathcal{E}\cup\mathcal{F}$, where $\mathcal{E}$ is described in \Cref{eq:def-E}. 

We will show that for any input $\boldx\in\bits^n$ there is some $\boldx'\in\bigcup_{i=1}^n\cJ_i$ such that $\fD_{k,1}(\boldx)=\fD_{k,1}(\boldx')$.
This means that
\begin{equation*}
    \fD_{k,1}\left(\bigcup_{i=1}^n\cJ_i\right) = \fD_{k,1}(\bits^n).
\end{equation*}
Since $|\fD_{k,1}(\bits^n)|=\left|\fD_{k,1}\left(\bigcup_{i=1}^n\cJ_i\right)\right|\leq \left|\bigcup_{i=1}^n\cJ_i\right|$, it follows that
\begin{equation*}
    C_k\leq \limsup_{n\to \infty} \frac{1}{n}\log \left|\bigcup_{i=1}^n\cJ_i\right|.
\end{equation*}

We are now ready to prove the key lemma.
\begin{lemma}\label{lem:J_n_onto}
For any $\mathbf{x}\in\{0,1\}^n$
    there exists a sequence $\mathbf{x}'\in\bigcup_{i=1}^n \mathcal{J}_{i}$
    such that $\fD_{k,1}(\boldx)=\fD_{k,1}(\boldx')$.
\end{lemma}
\begin{IEEEproof}
Fix an arbitrary $\mathbf{x}\in\{0,1\}^n$.
We apply a sequence of transformations to $\mathbf{x}$ to arrive at some
    $\mathbf{x}'\in\bigcup_{i=1}^n \mathcal{J}_{i}$
    which has the same output as $\mathbf{x}$ under $\fD_{k,1}$.

The transformations are as follows:
\begin{enumerate}
    \item 
        We first remove the ``stray'' single-bit runs.
        That is,
            from left to right,
            whenever we see $0^k10^k$ (resp.\ $1^k01^k$),
            we replace it with $0^{2k}$ (resp.\ $1^{2k}$).
        This leaves the output unchanged and only shortens the sequence or keeps it of the same length.
    \item 
        We then ``push'' each remaining single-bit run to the right.
        More precisely,
            from left to right,
            whenever we see 
            a length-one run after a run of length at least $k$
            (assuming the length-one run is $1$ and the run of length at least $k$ is $0^k$, and vice versa), 
            either:
        \begin{enumerate}
            \item [(i)]
                This $0^k1$ pattern is at the end of the sequence.
                In this case no further action is needed.
            \item [(ii)]
                This $0^k1$ is followed by a $0$-run.
                Note that, since Step 1 has been completed, the length of the following $0$-run is at most $k-1$, i.e., the run is $0^{\gamma}$ for some $\gamma\in[1,k-1]$.
                Then:
                \begin{enumerate}
                    \item [(a)] 
                        If the next $1$-run exists and is of length exactly $k-1$,
                            we replace $0^k10^{\gamma}$
                            with
                            $0^{k+\gamma-1}10$.
                    \item  [(b)]
                        Otherwise,
                            we replace $0^k10^{\gamma}$ by 
                            $0^{k+\gamma}1$.
                \end{enumerate}
        \end{enumerate}
            Note that this step does not affect the length of the sequence,
                and 
                the output remains unchanged
                since we avoid turning any length-$(k-1)$ run into a length-$k$ run and also avoid turning a run of length at least $k$ into a run of length strictly less than $k$.
    \item   
            From left to right,
                whenever we see $0^{k+1}1^k0^{\gamma}1$
                for some $\gamma\in[2,k-1]$,
                we replace it by 
                the shorter substring $0^k101^{k-1}0^{\gamma-1}1$.
            At the same time,
                whenever we see $0^{k+1}1^k0^{k+1}$,
                we replace it by the string 
                $0^k101^{k-1}0^k$.
\end{enumerate}
It can be checked that
    after these steps, 
    the sequence contains 
    no forbidden pattern from $\mathcal{E}\cup\mathcal{F}_1$.
It can also be verified that
    these steps
    will not increase the length of the sequence.
Therefore,
    the resulting sequence $\mathbf{x}'$ has to be in  $\bigcup_{n'=1}^n \mathcal{J}_{n'}$
    and has the same output as $\mathbf{x}$
    under $\fD_{k,1}$.
\end{IEEEproof}

\subsection{Concrete capacity bounds from \Cref{thm:capacity_p1}}

Recall the two limits from \Cref{eq:xi_k,eq:nu_k},
\begin{equation*}
    \xi_k = \liminf_{n\to\infty} |\cH_n|^{1/n}
\end{equation*}
and
\begin{equation*}
    \nu_k = \limsup_{n\to\infty} \left|\bigcup_{i=1}^n \cJ_i\right|^{1/n}.
\end{equation*}
Based on \Cref{thm:capacity_p1}, we know that
\begin{equation*}
    \log\xi_k \leq C_k\leq \log \nu_k,
\end{equation*}
for all $k\geq 2$.

\Cref{tab:capacity_bounds_numerical} presents some numerically computed values of 
    $\log \xi_k$ and $\log \nu_k$.
To compute $\log \xi_k$
    for each $k\in[2,10],$ we follow the process outlined below:
\begin{enumerate}[(1)]
    \item \label{step:capacity_numerical_pattern}
        Set $\mathcal{P}=\mathcal{E}\cup\{0^{k+1}1^k00,1^{k+1}0^k11\}$ as defined in \Cref{thm:capacity_p1}.
    \item \label{step:capacity_numerical_corr}
        Calculate the correlation polynomial $AB_z$ for all $A,B\in\mathcal{P}$,
            as described in \Cref{thm:forbiddern_patterns}.
    \item \label{step:capacity_numerical_solve}
        Solve the system of equations in \Cref{thm:forbiddern_patterns}
            to obtain the generating function of $|\mathcal{H}_n|$,
            denoted as $F(z)$,
            which is guaranteed to be a rational function of $z$.
    \item \label{step:capacity_numerical_fraction}
        Write $F(z)=\frac{n(z)}{d(z)}$
            for some coprime polynomials $n(z)$ and $d(z)$.
    \item \label{step:capacity_numerical_roots}
        Compute all roots of $d(z)$ numerically 
            and check for a simple largest-magnitude real root.
        Denote the root by $\xi$.
    \item \label{step:capacity_numerical_limit}
        By \Cref{lemma:analytics},
            we have 
            $\lim_{n\rightarrow\infty}\frac{1}{n}\log|\mathcal{H}_n|=\log \xi$.
        For more details regarding the existence of this  limit, see Appendix~\ref{app:0k10} and~\cite{odlyzko1985enumeration}.
\end{enumerate}
We observe that all the $d(z)$ for these $k$ in Step \ref{step:capacity_numerical_fraction} are of the form $z^{2k+2}-2z^{2k+1}+z^{k+1}-1$,
    and we numerically computed all the roots of $d(z)$ and $d'(z)$ to verify that the largest-magnitude real roots of the $d(z)$'s are indeed simple, as established in the Appendix.
It then follows that for $k\in[2,10]$ we have that
    $\xi_k$ is the largest-magnitude real root of $z^{2k+2}-2z^{2k+1}+z^{k+1}-1$.
We follow a similar procedures to compute $\log \nu_k$ 
    for $k\in[2,5]$,
    where the forbidden pattern set in Step~\ref{step:capacity_numerical_pattern}
    is replaced with $\mathcal{E}\cup\mathcal{F}$ as defined in \Cref{thm:capacity_p1}.
The denominators $d(z)$ from Step~\ref{step:capacity_numerical_fraction} are summarized  in \Cref{tab:poly_J_n},
    and we numerically verify that each of these polynomials has a simple largest-magnitude real root,
    which is the desired $\nu_k$.

\begin{table}
    \centering
    \begin{tabular}{c|c}
         $k$ & The denominator (polynomial) of the generating function of $|\mathcal{J}_n|$ \\
         \hline
2  & $z^{10} - 2z^9 + z^7 - z^4 - z^3 + z^2 - 1$ \\
3  & $z^{13} - 3z^{12} + 3z^{11} - 3z^{10} + 4z^9 - 4z^8 + 4z^7 - 4z^6 + 3z^5 - 3z^4 + 2z^3 - z^2 + z - 1$ \\
4  & $z^{18} - 2z^{17} + z^{13} - z^8 - z^5 + z^4 - 1$  \\
5  & $z^{21} - 3z^{20} + 3z^{19} - 3z^{18} + 3z^{17} - 3z^{16} + 4z^{15} - 4z^{14} 
       + 4z^{13} - 4z^{12} + 4z^{11} - 4z^{10} + 3z^9 - 3z^8 + 3z^7 - 3z^6
       + 2z^5 - z^4 + z^3 - z^2 + z - 1$  \\
    \end{tabular}
    \caption{The denominator of the generating function of $|\mathcal{J}_n|$ for selected $k$}
    \label{tab:poly_J_n}
\end{table}

\section*{Acknowledgment} The authors gratefully acknowledge useful discussions with Roni Con and Elena Grigorescu.

\bibliographystyle{IEEEtran}
\bibliography{refs}

@article{brakensiek2017efficient,
  title={Efficient low-redundancy codes for correcting multiple deletions},
  author={Brakensiek, Joshua and Guruswami, Venkatesan and Zbarsky, Samuel},
  journal={IEEE Transactions on Information Theory},
  volume={64},
  number={5},
  pages={3403--3410},
  year={2017},
  publisher={IEEE}
}

@book{stein2010complex,
  title={Complex Analysis},
  author={Stein, Elias M and Shakarchi, Rami},
  volume={2},
  year={2010},
  publisher={Princeton University Press}
}

@incollection{odlyzko1985enumeration,
  title={Enumeration of strings},
  author={Odlyzko, Andrew M},
  booktitle={Combinatorial Algorithms on Words},
  pages={205--228},
  year={1985},
  publisher={Springer}
}

@article{guibas1981string,
  title={String overlaps, pattern matching, and nontransitive games},
  author={Guibas, Leonidas J and Odlyzko, Andrew M},
  journal={Journal of Combinatorial Theory, Series A},
  volume={30},
  number={2},
  pages={183--208},
  year={1981},
  publisher={Elsevier}
}

@article{Lev65,
  title={Binary codes capable of correcting deletions, insertions, and reversals},
  author={Levenshtein, Vladimir I.},
  journal={Doklady Akademii Nauk},
  volume={163},
  number={4},
  pages={845--848},
  year={1965},
    url={https://www.mathnet.ru/eng/dan31411}
}

@article{guruswami2021explicit,
  title={Explicit two-deletion codes with redundancy matching the existential bound},
  author={Guruswami, Venkatesan and H{\aa}stad, Johan},
  journal={IEEE Transactions on Information Theory},
  volume={67},
  number={10},
  pages={6384--6394},
  year={2021},
  publisher={IEEE}
}

@article{sima2020optimal,
  title={On optimal $k$-deletion correcting codes},
  author={Sima, Jin and Bruck, Jehoshua},
  journal={IEEE Transactions on Information Theory},
  volume={67},
  number={6},
  pages={3360--3375},
  year={2020},
  publisher={IEEE}
}

@article{gabrys2018codes,
  title={Codes correcting two deletions},
  author={Gabrys, Ryan and Sala, Frederic},
  journal={IEEE Transactions on Information Theory},
  volume={65},
  number={2},
  pages={965--974},
  year={2018},
  publisher={IEEE}
}

@article{sima2019two,
  author={Sima, Jin and Raviv, Netanel and Bruck, Jehoshua},
  journal={IEEE Transactions on Information Theory}, 
  title={Two Deletion Correcting Codes From Indicator Vectors}, 
  year={2020},
  volume={66},
  number={4},
  pages={2375-2391},
  doi={10.1109/TIT.2019.2950290}}

@article{varshamov1965codes,
  title={Codes which correct single asymmetric errors (in Russian)},
  author={Varshamov, RR and Tenengolts, GM},
  journal={Automatika i Telemkhanika},
  volume={161},
  number={3},
  pages={288--292},
  year={1965}
}

@inproceedings{sima2020optimal-systematic,
  title={Optimal systematic $t$-deletion correcting codes},
  author={Sima, Jin and Gabrys, Ryan and Bruck, Jehoshua},
  booktitle={2020 IEEE International Symposium on Information Theory (ISIT)},
  pages={769--774},
  year={2020},
  organization={IEEE}
}

@inproceedings{Mitzenmacher08survey,
  author       = {Michael Mitzenmacher},
  editor       = {Joachim Gudmundsson},
  title        = {A Survey of Results for Deletion Channels and Related Synchronization
                  Channels},
  booktitle    = {Algorithm Theory - {SWAT} 2008, 11th Scandinavian Workshop on Algorithm
                  Theory, Gothenburg, Sweden, July 2-4, 2008, Proceedings},
  series       = {Lecture Notes in Computer Science},
  volume       = {5124},
  pages        = {1--3},
  publisher    = {Springer},
  year         = {2008},
  url          = {https://doi.org/10.1007/978-3-540-69903-3\_1},
  doi          = {10.1007/978-3-540-69903-3\_1},
  bibsource    = {dblp computer science bibliography, https://dblp.org}
}

@article{sloane2002single,
  title={On single-deletion-correcting codes},
  author={Sloane, Neil JA},
  journal={arXiv preprint math/0207197},
  year={2002}
}

@article{mercier2010survey,
  title={A survey of error-correcting codes for channels with symbol synchronization errors},
  author={Mercier, Hugues and Bhargava, Vijay K and Tarokh, Vahid},
  journal={IEEE Communications Surveys \& Tutorials},
  volume={12},
  number={1},
  pages={87--96},
  year={2010},
  publisher={IEEE}
}

@article{cheraghchi2020overview,
  title={An overview of capacity results for synchronization channels},
  author={Cheraghchi, Mahdi and Ribeiro, Jo{\~a}o},
  journal={IEEE Transactions on Information Theory},
  volume={67},
  number={6},
  pages={3207--3232},
  year={2021},
  publisher={IEEE}
}

@phdthesis{ryzhikov2020synchronizing,
  title={Synchronizing automata and coding theory},
  author={Ryzhikov, Andrew},
  year={2020},
  school={Universit{\'e} Paris-Est}
}

@article{heckel2019characterization,
  title={A characterization of the {DNA} data storage channel},
  author={Heckel, Reinhard and Mikutis, Gediminas and Grass, Robert N},
  journal={Scientific reports},
  volume={9},
  number={1},
  pages={9663},
  year={2019},
  publisher={Nature Publishing Group UK London}
}

@inproceedings{weindel2023embracing,
  title={Embracing errors is more effective than avoiding them through constrained coding for {DNA} data storage},
  author={Weindel, Franziska and Gimpel, Andreas L and Grass, Robert N and Heckel, Reinhard},
  booktitle={2023 59th Annual Allerton Conference on Communication, Control, and Computing (Allerton)},
  pages={1--8},
  year={2023},
  organization={IEEE}
}

@article{tabatabaei2022expanding,
  title={Expanding the molecular alphabet of {DNA}-based data storage systems with neural network nanopore readout processing},
  author={Tabatabaei, S Kasra and Pham, Bach and Pan, Chao and Liu, Jingqian and Chandak, Shubham and Shorkey, Spencer A and Hernandez, Alvaro G and Aksimentiev, Aleksei and Chen, Min and Schroeder, Charles M and others},
  journal={Nano letters},
  volume={22},
  number={5},
  pages={1905--1914},
  year={2022},
  publisher={ACS Publications}
}

@article{pan2022rewritable,
  title={Rewritable two-dimensional {DNA}-based data storage with machine learning reconstruction},
  author={Pan, Chao and Tabatabaei, S Kasra and Tabatabaei Yazdi, SM Hossein and Hernandez, Alvaro G and Schroeder, Charles M and Milenkovic, Olgica},
  journal={Nature communications},
  volume={13},
  number={1},
  pages={2984},
  year={2022},
  publisher={Nature Publishing Group UK London}
}

@article{tabatabaei2015rewritable,
  title={A rewritable, random-access {DNA}-based storage system},
  author={Tabatabaei Yazdi, SM Hossein and Yuan, Yongbo and Ma, Jian and Zhao, Huimin and Milenkovic, Olgica},
  journal={Scientific reports},
  volume={5},
  number={1},
  pages={14138},
  year={2015},
  publisher={Nature Publishing Group UK London}
}

@article{milenkovic2024dna,
  title={{DNA}-based data storage systems: A review of implementations and code constructions},
  author={Milenkovic, Olgica and Pan, Chao},
  journal={IEEE Transactions on Communications},
  volume={72},
  number={7},
  pages={3803--3828},
  year={2024},
  publisher={IEEE}
}

@article{bancroft2001long,
  title={Long-term storage of information in {DNA}},
  author={Bancroft, Catherine and Bowler, Tamsin and Bloom, Brent and Clelland, Catherine T.},
  journal={Science},
  volume={293},
  number={5536},
  pages={1763--1765},
  year={2001},
  publisher={American Association for the Advancement of Science}
}

@article{church2012next,
  title={Next-generation digital information storage in {DNA}},
  author={Church, George M. and Gao, Yuan and Kosuri, Sriram},
  journal={Science},
  volume={337},
  number={6102},
  pages={1628--1628},
  year={2012},
  publisher={American Association for the Advancement of Science}
}

@article{goldman2013towards,
  title={Towards practical, high-capacity, low-maintenance information storage in synthesized {DNA}},
  author={Goldman, Nick and Bertone, Paul and Chen, Siyuan and Dessimoz, Christophe and LeProust, Emily M. and Sipos, Botond and Birney, Ewan},
  journal={Nature},
  volume={494},
  number={7435},
  pages={77--80},
  year={2013},
  publisher={Nature Publishing Group}
}

@article{grass2015robust,
  title={Robust chemical preservation of digital information on {DNA} in silica with error-correcting codes},
  author={Grass, Robert N. and Heckel, Reinhard and Puddu, Matteo and Paunescu, Dan and Stark, Wendelin J.},
  journal={Angewandte Chemie International Edition},
  volume={54},
  number={8},
  pages={2552--2555},
  year={2015},
  publisher={Wiley Online Library}
}

@article{sima2023error,
  title={Error correction for {DNA} storage},
  author={Sima, Jin and Raviv, Netanel and Schwartz, Moshe and Bruck, Jehoshua},
  journal={IEEE BITS the Information Theory Magazine},
  volume={3},
  number={3},
  pages={78--94},
  year={2023},
  publisher={IEEE}
}

@article{press2020hedges,
  title={HEDGES error-correcting code for {DNA} storage corrects indels and allows sequence constraints},
  author={Press, William H. and Hawkins, Jay A. and Jones, Simon K. and Schaub, Jonathan M. and Finkelstein, Ilana J.},
  journal={Proceedings of the National Academy of Sciences},
  volume={117},
  number={35},
  pages={18489--18496},
  year={2020},
  publisher={National Academy of Sciences}
}

@article{lee2020photon,
  title={Photon-directed multiplexed enzymatic {DNA} synthesis for molecular digital data storage},
  author={Lee, Hyun-Ho and Kalhor, Reza and Goela, Neelkant and Bolot, Jonathan and Church, George M.},
  journal={Nature Communications},
  volume={11},
  pages={5246},
  year={2020},
  doi={10.1038/s41467-020-19010-8}
}

@article{doricchi2022emerging,
  title={Emerging approaches to {DNA} data storage: challenges and prospects},
  author={Doricchi, Andrea and Platnich, Casey M and Gimpel, Andreas and Horn, Friederikee and Earle, Max and Lanzavecchia, German and Cortajarena, Aitziber L and Liz-Marz{\'a}n, Luis M and Liu, Na and Heckel, Reinhard and others},
  journal={ACS nano},
  volume={16},
  number={11},
  pages={17552--17571},
  year={2022},
  publisher={ACS Publications}
}

@article{yazdi2017portable,
  title={Portable and error-free {DNA}-based data storage},
  author={Yazdi, SM Hossein Tabatabaei and Gabrys, Ryan and Milenkovic, Olgica},
  journal={Scientific reports},
  volume={7},
  number={1},
  pages={5011},
  year={2017},
  publisher={Nature Publishing Group UK London}
}

@article{lopez2019dna,
  title={{DNA} assembly for nanopore data storage readout},
  author={Lopez, Randolph and Chen, Yuan-Jyue and Dumas Ang, Siena and Yekhanin, Sergey and Makarychev, Konstantin and Racz, Miklos Z and Seelig, Georg and Strauss, Karin and Ceze, Luis},
  journal={Nature communications},
  volume={10},
  number={1},
  pages={2933},
  year={2019},
  publisher={Nature Publishing Group UK London}
}

@article{tabatabaei2020dna,
  title={{DNA} punch cards for storing data on native {DNA} sequences via enzymatic nicking},
  author={Tabatabaei, S Kasra and Wang, Boya and Athreya, Nagendra Bala Murali and Enghiad, Behnam and Hernandez, Alvaro Gonzalo and Fields, Christopher J and Leburton, Jean-Pierre and Soloveichik, David and Zhao, Huimin and Milenkovic, Olgica},
  journal={Nature communications},
  volume={11},
  number={1},
  pages={1742},
  year={2020},
  publisher={Nature Publishing Group UK London}
}

@article{antkowiak2020low,
  title={Low cost {DNA} data storage using photolithographic synthesis and advanced information reconstruction and error correction},
  author={Antkowiak, Philipp L and Lietard, Jory and Darestani, Mohammad Zalbagi and Somoza, Mark M and Stark, Wendelin J and Heckel, Reinhard and Grass, Robert N},
  journal={Nature communications},
  volume={11},
  number={1},
  pages={5345},
  year={2020},
  publisher={Nature Publishing Group UK London}
}

@misc{con2025channels,
      title={Channels with Input-Correlated Synchronization Errors}, 
      author={Roni Con and João Ribeiro},
      year={2025},
      eprint={2504.14087},
      archivePrefix={arXiv},
      primaryClass={cs.IT},
      url={https://arxiv.org/abs/2504.14087}, 
    note={Preliminary version in ISIT 2025}
}

@ARTICLE{MDK18,
  author={Mao, Wei and Diggavi, Suhas N. and Kannan, Sreeram},
  journal={IEEE Transactions on Information Theory}, 
  title={Models and Information-Theoretic Bounds for Nanopore Sequencing}, 
  year={2018},
  volume={64},
  number={4},
  pages={3216-3236},
  doi={10.1109/TIT.2018.2809001}}

@ARTICLE{ABGHK24,
  author={Alon, Noga and Bourla, Gabriela and Graham, Ben and He, Xiaoyu and Kravitz, Noah},
  journal={IEEE Transactions on Information Theory}, 
  title={Logarithmically Larger Deletion Codes of All Distances}, 
  year={2024},
  volume={70},
  number={1},
  pages={125-130},
  doi={10.1109/TIT.2023.3304565}}

@ARTICLE{CK14,
  author={Cullina, Daniel and Kiyavash, Negar},
  journal={IEEE Transactions on Information Theory}, 
  title={An Improvement to {Levenshtein's} Upper Bound on the Cardinality of Deletion Correcting Codes}, 
  year={2014},
  volume={60},
  number={7},
  pages={3862-3870},
  doi={10.1109/TIT.2014.2317698}}

@ARTICLE{KK13,
  author={Kulkarni, Ankur A. and Kiyavash, Negar},
  journal={IEEE Transactions on Information Theory}, 
  title={Nonasymptotic Upper Bounds for Deletion Correcting Codes}, 
  year={2013},
  volume={59},
  number={8},
  pages={5115-5130},
  doi={10.1109/TIT.2013.2257917}}

@article{Dob67,
  title={Shannon's theorems for channels with synchronization errors},
  author={Dobrushin, Roland L.},
  journal={Problemy Peredachi Informatsii},
  volume={3},
  number={4},
  pages={18--36},
  year={1967},
  publisher={Russian Academy of Sciences, Branch of Informatics, Computer Equipment and~…},
    url={https://www.mathnet.ru/eng/ppi1919}
}

@ARTICLE{LT21,
  author={Li, Yonglong and Tan, Vincent Y. F.},
  journal={IEEE Transactions on Information Theory}, 
  title={On the Capacity of Channels With Deletions and States}, 
  year={2021},
  volume={67},
  number={5},
  pages={2663-2679},
  doi={10.1109/TIT.2020.3043117}}

@INPROCEEDINGS{Hae19,  author={Bernhard {Haeupler}},  booktitle={2019 IEEE 60th Annual Symposium on Foundations of Computer Science (FOCS)},   title={Optimal Document Exchange and New Codes for Insertions and Deletions},   year={2019},  volume={},  number={},  pages={334-347},}

@article{CJLW22,
author = {Cheng, Kuan and Jin, Zhengzhong and Li, Xin and Wu, Ke},
title = {Deterministic Document Exchange Protocols and Almost Optimal Binary Codes for Edit Errors},
year = {2022},
issue_date = {December 2022},
publisher = {Association for Computing Machinery},
address = {New York, NY, USA},
volume = {69},
number = {6},
issn = {0004-5411},
url = {https://doi.org/10.1145/3561046},
doi = {10.1145/3561046},
journal = {J. ACM},
month = nov,
articleno = {44},
numpages = {39}
}

@ARTICLE{schoeny17codes,
  author={Schoeny, Clayton and Wachter-Zeh, Antonia and Gabrys, Ryan and Yaakobi, Eitan},
  journal={IEEE Transactions on Information Theory}, 
  title={Codes Correcting a Burst of Deletions or Insertions}, 
  year={2017},
  volume={63},
  number={4},
  pages={1971-1985},
  doi={10.1109/TIT.2017.2661747}}

@InProceedings{bellare09format,
author="Bellare, Mihir
and Ristenpart, Thomas
and Rogaway, Phillip
and Stegers, Till",
editor="Jacobson, Michael J.
and Rijmen, Vincent
and Safavi-Naini, Reihaneh",
title="Format-Preserving Encryption",
booktitle="Selected Areas in Cryptography",
year="2009",
publisher="Springer Berlin Heidelberg",
address="Berlin, Heidelberg",
pages="295--312",
}

@inproceedings{goldberg85compression,
author = {Goldberg, A and Sipser, M},
title = {Compression and ranking},
year = {1985},
isbn = {0897911512},
publisher = {Association for Computing Machinery},
address = {New York, NY, USA},
doi = {10.1145/22145.22194},
booktitle = {Proceedings of the Seventeenth Annual ACM Symposium on Theory of Computing (STOC 1985)},
pages = {440-448},
numpages = {9},
location = {Providence, Rhode Island, USA},
}

@article{immink2022innovation,
  title={Innovation in constrained codes},
  author={Immink, Kees A Schouhamer},
  journal={IEEE Communications Magazine},
  volume={60},
  number={10},
  pages={20--24},
  year={2022},
  publisher={IEEE}
}

@article{alon1992simple,
  title={Simple constructions of almost $k$-wise independent random variables},
  author={Alon, Noga and Goldreich, Oded and H{\aa}stad, Johan and Peralta, Ren{\'e}},
  journal={Random Structures \& Algorithms},
  volume={3},
  number={3},
  pages={289--304},
  year={1992},
  publisher={Wiley Online Library}
}

@ARTICLE{CGMR20,
  author={Cheraghchi, Mahdi and Gabrys, Ryan and Milenkovic, Olgica and Ribeiro, João},
  journal={IEEE Transactions on Information Theory}, 
  title={Coded Trace Reconstruction}, 
  year={2020},
  volume={66},
  number={10},
  pages={6084-6103},
  doi={10.1109/TIT.2020.2996377}}

@ARTICLE{HS21,
  author={Haeupler, Bernhard and Shahrasbi, Amirbehshad},
  journal={IEEE Transactions on Information Theory}, 
  title={Synchronization Strings and Codes for Insertions and Deletions—A Survey}, 
  year={2021},
  volume={67},
  number={6},
  pages={3190-3206},
  doi={10.1109/TIT.2021.3056317}}

@ARTICLE{Abdel98,
  author={Abdel-Ghaffar, K. A. S. and Ferreira, H. C.},
  journal={IEEE Trans. Inf. Theory}, 
  title={Systematic encoding of the {V}arshamov-{T}enengol'ts codes and the {C}onstantin-{R}ao codes}, 
  year={1998},
  volume={44},
  number={1},
  pages={340-345},
  doi={10.1109/18.651063}}

\appendices

\section{Enumerating strings that avoid the substrings $0^k10$ and $1^k01$}\label{app:0k10}

This appendix explains
    how to count the number of binary sequences that
    do not contain substrings from a set of forbidden patterns.
For simplicity,
    we only show the details for the case where the forbidden set is $\{0^k10,1^k01\}$.
Similar derivations can be performed for any other pattern sets.
\begin{remark}
The enumeration results to follow also lead to a (worse) lower bound on
    the capacity of the extremal contextual deletion channel compared to \Cref{thm:capacity_p1}.
In fact, note that
    this channel
    is injective on 
    the set of length-$n$ binary sequences forbidding $\{0^k10,1^k01\}$,
    since, during the decoding process,
        after each run of length at least $k$
    the only place where we can add the deleted bit back
    is right at its end.
\end{remark}

Let $\mathbf{x},\mathbf{y}\in\Sigma^n$ be two strings over an alphabet $\Sigma$ of cardinality $q$. 
The \emph{correlation vector} of the two strings, denoted by $\mathbf{x} \circ  \mathbf{y},$ is a length-$n$ binary vector whose $i$-th coordinate from the right is $1$ if and only if the length-$i$ suffix of $\mathbf{x}$ equals the length-$i$ prefix of $\mathbf{y}$. As an example,
for $n=5$ and $\mathbf{x}=11010$, $\mathbf{y}=01011$, we have
\begin{equation}
\begin{array}{cccccc}
    \mathbf{x}= &1 & 1 & 0 & 1 & 0 \\
    \mathbf{y}= &0 & 1 & 0 & 1 & 1 \\
    \underline{\mathbf{y}_1}= &  & 0 & 1 & 0 & 1 \\
    \underline{\mathbf{y}_2}= & &  & 0 & 1 & 0 \\
    \underline{\mathbf{y}_3}= &   &  &  & 0 & 1 \\
    \underline{\mathbf{y}_4}= & &  &  &  & 0 \\
\end{array} \notag
\end{equation}
where $\underline{\mathbf{y}_j}$ denotes a right shift of $\mathbf{y}$ by $j$ positions. Hence, $\mathbf{x} \circ \mathbf{y}=00101.$ It is convenient to represent the correlation vector as a polynomials, which for our example equals
$1+z^2.$ Note that henceforth, for two strings $\mathbf{x},\mathbf{y}$ we use $XY_z$ to denote their correlation in polynomial form, so that for the above case, $XY_z=1+z^2.$ When $\mathbf{x}=\mathbf{y}$, we refer to the correlation vector as the \emph{autocorrelation vector}, and the correlation polynomial as the \emph{autocorrelation polynomial}, denoted by $XX_z$. 
Furthermore, it is clear that in general, $\mathbf{x} \circ \mathbf{y}\neq \mathbf{y} \circ \mathbf{x}$.

For $X=0^k10$ and $Y=1^k01$, it is straightforward to see that $XX_z=YY_z=1+z^{k+1}$, and that $XY_z=YX_z=0$ unless $k=1.$

We say that a set of $t$ strings $\{{A,B,\ldots,T\}}$ over an alphabet of size $q$ is \emph{reduced} if for no string in the set is a proper substring of another string in the set. Clearly, $\{{0^k10,1^k01\}}$ is a reduced set of two strings over a binary alphabet.
Also, we let $f_{X}(n)$ denote the number of strings of length $n$ that
 end with $X \in \{{A,B,\ldots,T\}}$ and have no other occurrence of $A,B,\ldots, T$,
and we use $f(n)$ to denote the number of strings of length $n$ that avoid \emph{all} strings in $\{{A,B,\ldots,T\}}$. The ordinary generating functions of the counting numbers of the above described strings, $f_{X}(n),$ $X \in \{{A,B,\ldots,T\}},$ and $f(n)$ (e.g., $F(z)=\sum_{n=1}^{\infty}\, f(n)\,z^{-n}$) are denoted by $F_{X}(z)$ and $F(z),$ respectively. 

We find the following result from~\cite{guibas1981string,odlyzko1985enumeration} useful for our subsequent derivations.

\begin{theorem}[{\cite[Theorem 4.1]{odlyzko1985enumeration}}]\label{thm:forbiddern_patterns}Let $f(n)$ be the number of strings of length $n$ over an alphabet of cardinality $q$ that avoid all strings in a  reduced set of $t$ strings $\{{A,B,\ldots,T\}}$. Then, the generating functions of the number of strings
that avoid all strings in the reduced set, $F(z)$, and that avoid all strings in the reduced set except for a single occurrence of $X$ at the end,
$F_{X}(z)$,
satisfy the following system of equations:
\begin{align*}
(z-q)F(z)+zF_A(z)+zF_B(z)+\cdots+zF_T(z)&=z,\\
F(z)-zAA_zF_A(z)-zBA_zF_B(z)-\cdots-zTA_zF_T(z)&=0,\\
&\vdots\\
F(z)-zAT_zF_A(z)-zBT_zF_B(z)-\ldots-zTT_zF_T(z)&=0.
\end{align*}
In particular,
    all the generating functions
    $F(z)$, $F_A(z)$, \ldots, and $F_T(z)$
    are rational functions of $z$.
\end{theorem}

For the reduced set $\{{A=0^k10,B=1^k01\}}$, and $q=2$, the system of equations for the generating functions takes the form
\begin{align*}
(z-2)F(z)+zF_A(z)+zF_B(z)&=z\\
   F(z)-z(1+z^{k+1})F_A(z)&=0,\\
   F(z)-z(1+z^{k+1})F_B(z)&=0.
\end{align*}

Hence, $F_A(z)=F_B(z)=\frac{F(z)}{z(1+z^{k+1})},$ and
\begin{equation}
(z-2)F(z)+\frac{2 F(z)}{1+z^{k+1}}-z=0,
\end{equation}
i.e.,
\begin{equation}    
F(z)=\frac{1+z^{k+1}}{1-2z^k+z^{k+1}}.
\end{equation}
We can estimate the asymptotic behavior of $f(n)$ using well-established techniques from analytic combinatorics, and in particular, the following results.

Let $\rho$ be the largest magnitude real root of the denominator of $F(z).$ Then,
$$f(n)\leq c\rho^n,$$
where $c$ is a positive constant. We also have the more quantitative estimate from~\cite{odlyzko1985enumeration}, which is formally stated below.

\begin{lemma}[{\cite[Lemma 2.1]{odlyzko1985enumeration}}] \label{lemma:analytics}
Suppose that $F(z)$ is a generating function which is analytic for $|z|\geq r>0,$ for some given $r$, with the possible exception of a simple pole $z=\rho,$ for which we have $|\rho|>r$ and a residue equal to $\alpha.$ If in addition
$$|F(z)| \leq C \text{ for } |z|=r,$$
where $C$ is some constant, then
\begin{equation*}
    |f(n)-\alpha\,\rho^{n-1}|\leq r^n (C+|\alpha| (|\rho|-r))^{-1}, \quad \textrm{for all $n\geq 1$.}
\end{equation*}
\end{lemma}

To apply \Cref{lemma:analytics}, we need the following result for the denominator polynomial
$1-2z^k+z^{k+1}$ of the generating function.

\begin{proposition}\label{prop:denom_0k10} For every integer $k\geq2,$ the polynomial $g_k(z)=1-2z^k+z^{k+1}$ has no multiple roots (i.e., all the roots are simple). Furthermore, exactly one root of the polynomial has absolute value $>1$.\end{proposition}

\begin{proof} 
For $k=1$, we have $1-2z+z^2=(1-z)^2$, i.e., we have a double root at $z=1$. For the case $k=2$, it is easy to solve the cubic equation and see that the roots equal $z=1$, $z=\frac{1-\sqrt{5}}{2}$ and $z=\frac{1+\sqrt{5}}{2}.$ Hence, we focus on the setting $k>2$.

It is straightforward to compute the derivative of the polynomial as
$$
g_k'(z) = -2k z^{k-1}+(k+1)z^k = z^{k-1}\big((k+1)z-2k\big).
$$
If $r$ were a multiple root then $g_k(r)=0$ and $g_k'(r)=0$. From $g_k'(r)=0$ we have either $r=0$ or $(k+1)r-2k=0$. The case $r=0$ is impossible since $g_k(0)=1\neq0$. Hence
$$
r=\frac{2k}{k+1}.
$$
Substituting the above $r$ into the polynomial gives $g_k(r)=1-2r^k+r^{k+1}=0$. Upon factorization we get
$$
g_k(r)=1+r^k(r-2)=1-\frac{2}{k+1}r^k.
$$
Thus, a repeated root would have to satisfy
$$
1-\frac{2}{k+1}r^k=0\quad\Longrightarrow\quad r^k=\frac{k+1}{2}.
$$

But $r=\dfrac{2k}{k+1}=2-\dfrac{2}{k+1}\ge \tfrac{3}{2}$ for every integer $k\ge3$. Hence

$$
r^k \ge\Big(\tfrac{3}{2}\Big)^k,
$$

and therefore

$$
2r^k \ge 2\Big(\tfrac{3}{2}\Big)^k.
$$

For $k=3$ we have $2(3/2)^3=6.75>4$, and for $k\ge3$ the quantity $2(3/2)^k$ grows exponentially; in particular, by induction, we have
$$
2\Big(\tfrac{3}{2}\Big)^k>k+1,\quad \forall \, k\ge3.
$$

Therefore $2r^k>k+1$, contradicting the equality $2r^k=k+1$. Hence no $r$ can satisfy both $g_k(r)=0$ and $g_k'(r)=0$, so $g_k$ has no multiple roots for any integer $k> 2$ either. 

To prove the second claim, we let $g(z)=z^{k+1}-2z^k+1$ and set $w=1/z$. Then $g(z)=0$ if and only if
$$
h(w):=w^{k+1}-2w+1=0,
$$
and roots $z()$ with $|z|>1$ correspond bijectively to roots $w$ of $h()$ with $|w|<1$. So it suffices to show that $h$ has exactly one root in the open unit disk.

Pick any radius $\rho \in(0.62,1)$. On the circle $|w|=\rho$ we have
$$
|w^{k+1}+1|\le \rho^{\,k+1}+1\le 1+\rho^3,
$$
because $\rho<1$ and $k>2$. For $\rho>0.62$, one can easily check that
$$
2\rho>1+\rho^3,
$$
and hence, on $|w|=\rho$
$$
|{-}2w|=2\rho>|w^{k+1}+1|.
$$

By Rouché's theorem\footnote{The theorem asserts the following. Let $F$ and $G$ be functions analytic inside and on a simple closed contour $\mathcal{C}$, with $|G(z)| < |F(z)| \; \text{for all } z \in \mathcal{C}.$
Then $F$ and $F+G$ have the same number of zeros, counted with multiplicity, inside $\mathcal{C}$.} \cite{stein2010complex}, the functions $-2w$ and $h(w)=-2w+(w^{k+1}+1)$ have the same number of zeros inside $|w|<\rho$. The function $-2w$ has exactly one zero (a simple one) at $w=0$, and therefore $h$ has exactly one zero in $|w|<\rho$. Since $\rho<1$, this same zero-count holds for the unit disk $|w|<1$.
Therefore, there is exactly one $w$ with $|w|<1$, so exactly one $z=1/w$ with $|z|>1$.

Note that another simple argument, counting the number of roots inside the unit circle and subtracting this number ($k$) leads to the same conclusion. It can also be shown that this root lies in $(0,1),$ but the proof is omitted.
\end{proof}

\begin{proposition}
For $k>2,$ write
$$
R(z)=\frac{1+z^{k+1}}{1-2z^k+z^{k+1}}=\frac{1+z^{k+1}}{g_k(z)},\qquad g_k(z)=1-2z^k+z^{k+1}.
$$
Let $\rho$ denote the largest real root of $g_k(z)$. Then, the residue of $R(z)$ at $\rho$ equals  
$$
\alpha(\rho)=\operatorname{Res}_{z=\rho}R(z)=\frac{2\rho^k}{\rho^{\,k-1}\big((k+1)\rho-2k\big)}
=\frac{2\rho}{(k+1)r-2k}.
$$
\end{proposition}
Since the roots are simple the residue at the simple pole $z=\rho$ is
$$
\operatorname{Res}_{z=\rho}R(z)=\frac{1+\rho^{\,k+1}}{g_k'(\rho)}.
$$
Using $g_k(\rho)=0$ we get $1+\rho^{\,k+1}=2\rho^k$. Also
$$
g_k'(z)=-2k z^{k-1}+(k+1)z^k=z^{\,k-1}\big((k+1)z-2k\big),
$$
so that
$$
g_k'(\rho)=\rho^{\,k-1}\big((k+1)\rho-2k\big).
$$
Therefore
$$
\operatorname{Res}_{z=\rho}R(z)=\frac{2\rho^k}{\rho^{\,k-1}\big((k+1)\rho-2k\big)}
=\frac{2\rho}{(k+1)\rho-2k}.
$$
Equivalently $\operatorname{Res}_{z=\rho}= \dfrac{2}{(k+1)-\tfrac{2k}{\rho}}$, where $\rho$ is the largest real solution of $\rho^k(\rho-2)=-1$.

The results of the above propositions reveal that asymptotically, $f(n)$ scales as $\alpha(\rho)\, \rho^{n-1},$ where $\rho$ is the unique real root larger than $1$. 
A quick inspection of the polynomial  
$1-2z^k+z^{k+1}$
reveals that 
$\rho=\frac{1+\sqrt{5}}{2} \approx 1.618$ for $k=2$, and $\rho \approx 1.8393$ for $k=3.$
Hence, the capacity of the forbidden $\{{0^k10,1^k01\}}$ string constraint equals the limit of $\frac{\log(\alpha(\rho)\, \rho^{n-1})}{n}$ as $n \to \infty$, which is 
$\log \rho$. This is approximately $0.6942$ when $k=2$ and $0.8791$ when $k=3$.
Additional numerical values of $\rho$ for larger values of $k$ are given in \Cref{tab:capacity_0k10_1k01}.

\begin{table}
    \centering
    \begin{tabular}{c|c|c}
         $k$ & $\rho$ & Rate ($\log\rho$) \\
         \hline
         2 & 1.6180 &0.6942\\
         3 & 1.8393 &0.8791\\
         4 & 1.9276 &0.9468\\
         5 & 1.9659 &0.9752\\
         6 & 1.9836 &0.9881\\
         7 & 1.9920 &0.9942\\
         8 & 1.9960 &0.9971\\
         9 & 1.9980 &0.9986\\
         10 & 1.9990 &0.9993\\
    \end{tabular}
    \vspace{0.1in}
    \caption{Some numerical values of the capacity under the constraint of no $\{0^k10,1^k01\}$.}
    \label{tab:capacity_0k10_1k01}
\end{table}

\section{Efficient encoding and decoding of $\mathcal{C}_{\varepsilon}$ by ranking and unranking deterministic finite automata}\label{app:DFA}

We conclude our analysis by showing that there exist efficient (polynomial-time) encoders and decoders mapping 
    $\{0,1\}^{n-1}$ to $\mathcal{C}_{\varepsilon}$,
    where $\mathcal{C}_{\varepsilon}$ is defined in \Cref{lem:C_eps}.
The high-level idea is to note that
    for each $n$ we can construct a deterministic finite automaton (DFA) with the following properties:
\begin{enumerate}
    \item $\mathcal{C}_{\varepsilon}$ is precisely
    the collection of all length-$n$ binary sequences that are accepted by this DFA, and we can compute a description of this DFA given $n$ in time $\poly(n)$.
    \item The number of states in this DFA is $\poly(n)$, and we can compute a description of this DFA in time $\poly(n)$. 
\end{enumerate}
Then,
    we can exploit existing results about ranking and unranking
    accepted strings for a DFA \cite{goldberg85compression,bellare09format}.
To be more precise,
    for any DFA,
    the ranking function maps a length-$n$ accepted sequence to its index when all the length-$n$ accepted sequences are sorted in lexicographic order.
Conversely,
    the unranking function maps an index $i$ to the $i$-th length-$n$ accepted sequence.
The ranking and unranking functions implemented in \cite[Figure 2]{bellare09format}
    can be shown to have
    time complexity $n^2$ times the number of states in that DFA,
    which is poly$(n)$ in our case.
Therefore,
    the unranking function and the ranking functions
    for the DFA we construct are the desired efficient encoder and decoder,
    respectively.

Recall that a DFA is defined as a five-tuple
\begin{align*}
    M\coloneqq (Q, \Sigma, \delta, q^{(0)}, F),
\end{align*}
    where:
    \begin{itemize}
        \item $Q$ is the set of states,
        \item $\Sigma$ is the alphabet,
        \item $\delta: Q\times \Sigma \rightarrow Q$ is the transition function,
        \item $q^{(0)}\in Q$ is the initial state, and
        \item $F\subseteq Q$ is the set of accepted states. 
    \end{itemize}
\begin{proposition}\label{prop:DFA}
Given $n$,
    we can construct in $\poly(n)$ time a DFA $M_n$ over $\{0,1\}$ with $\poly(n)$ states
    satisfying the following property:
For each length-$n$ sequence $\mathbf{x}$,
    we have $\mathbf{x}\in\mathcal{C}_{\varepsilon}$
    if and only if 
    $\mathbf{x}$ is accepted by $M_n$,
    where $\mathcal{C}_{\varepsilon}$ is defined in \Cref{lem:C_eps}.
\end{proposition}
\begin{IEEEproof}
For each $n$,
    define the following DFA 
    $M_n\coloneqq (Q, \{0,1\}, \delta, q^{(0)}, F)$,
    where:
    \begin{enumerate}
        \item 
            $Q\coloneqq \{\epsilon\}\cup\{R\} \cup [0,\frac{n\log n}{2^k}]\times [1,2\log n]\times [0,k] \times [1,w]\times [1,w]\times \{0,1\}$.
            Here $\epsilon$ denotes the empty string,
                and $R$ represents the rejection state.
            Any other state $q$ is in the form of
                a six-tuple of integers
                $S=(q_1,q_2,q_3,q_4,q_5,q_6)$,
                where:
            \begin{enumerate}
                \item 
                    $q_1$ keeps track of the number of runs of length at least $k$ (capped at $\frac{n\log n}{2^k}$).
                \item 
                    $q_2$ keeps track of the current run length (capped at $2\log n$).
                \item 
                    $q_3$ keeps track of the previous run length (capped at $k$).
                \item 
                    $q_4$ is the maximum length of a window that ends at the current bit and contains no $0^{\ell}$ (capped at $w$).
                \item 
                    $q_5$ is the maximum length of a window that ends at the current bit and contains no $1^{\ell}$ (capped at $w$).
                \item 
                    $q_6$ records (equals) the current bit.
            \end{enumerate}
            Intuitively,
                $q_4$ can be thought of as
                a shifted version of the distance from the current bit to the previous run $0^{\ell}$ (if it exists),
                and a similar explanation holds for $q_5$.
        \item 
            The transition function $\delta$ is formally defined as follows:
            For $a\in\{0,1\}$, set
                $\delta(\epsilon,a)=(0,1,0,1,1,a)$.
            Then, for $a\in\{0,1\}$ and $q\in Q'$, set
                $\delta(q,a)=R$,
                where $Q'$ is the collection of states in $[0,\frac{n\log n}{2^k}]\times [1,2\log n]\times [0,k] \times [1,w]\times [1,w]$ that satisfy at least one of the following five conditions:
            \begin{enumerate}
                \item $q_1=\frac{n\log n}{2^k}$.
                \item $q_2=2 \log n$.
                \item $q_2=\ell$ and $q_3=k$.
                \item $q_4=w$.
                \item $q_5=w$.
            \end{enumerate}
            Also,
                 for $a\in\{0,1\}$ set $\delta(R,a)=R$.
            Intuitively,
                a sequence enters $Q'$ when it violates at least one requirement of $\mathcal{C}_{\varepsilon}$ for the first time.
            Since this sequence is bound to be outside of $\mathcal{C}_{\varepsilon}$ no matter what the remaining bits are,
                we can just transition to the rejected state $R$ and remain in it.

For the remaining states $q\in [0,\frac{n\log n}{2^k}]\times [1,2\log n]\times [0,k] \times [1,w]\times [1,w]\times \{0,1\}\setminus Q'$,
                define the transition function as
            \begin{align*}
                \delta((q_1,q_2,q_3,q_4,q_5,0),0) &= 
                \begin{cases}
                    (q_1+\boldone_{q_2=k-1},
                    q_2+1,
                    q_3,
                    \ell-1,
                    q_5+1,
                    0),\textnormal{ if }q_2\geq \ell-1,\\
                    (q_1+\boldone_{q_2=k-1},
                    q_2+1,
                    q_3,
                    q_4+1,
                    q_5+1,
                    0),\textnormal{ otherwise}.
                \end{cases}\\
                \delta((q_1,q_2,q_3,q_4,q_5,0),1) &= 
                    (q_1,
                    1,
                    q_2\wedge k,
                    q_4+1,
                    q_5+1,
                    1),\\
                \delta((q_1,q_2,q_3,q_4,q_5,1),0) &= 
                    (q_1,
                    1,
                    q_2\wedge k,
                    q_4+1,
                    q_5+1,
                    0),\\
                \delta((q_1,q_2,q_3,q_4,q_5,1),1) &= 
                \begin{cases}
                    (q_1+\boldone_{q_2=k-1},
                    q_2+1,
                    q_3,
                    q_4+1,
                    \ell-1,
                    1),\textnormal{ if }q_2\geq \ell-1,\\
                    (q_1+\boldone_{q_2=k-1},
                    q_2+1,
                    q_3,
                    q_4+1,
                    q_5+1,
                    1),\textnormal{ otherwise},
                \end{cases}
            \end{align*}
            where $q_2\wedge k\coloneqq\min(q_2,k)$ and $\boldone_{q_2=k-1}$ is the indicator function of the equation $q_2=k-1$ (i.e. $\boldone_{q_2=k-1}=1$ when $q_2=k-1$ and is $0$ otherwise).
        \item 
            The starting state is $q^{(0)}\coloneqq\epsilon$.
        \item 
            The acceptance states are $F\coloneqq [0,\frac{n\log n}{2^k}]\times [1,2\log n]\times [0,k] \times [1,w]\times [1,w]\times \{0,1\}\setminus Q'$.
    \end{enumerate}
    
    By construction,
        a length-$n$ binary sequence is accepted by $M_n$ if and only if it is in $\mathcal{C}_{\varepsilon}$.
    Furthermore,
        the number of states in $M_n$ is
        $2+(\frac{n\log n}{2^k}+1)\cdot 2\log  n\cdot(k+1)\cdot w^2\cdot 2=O(n^{3-3C+2\varepsilon}\log^3n)$,
        which is polynomial in $n$, and it is clear from the definition of $\delta$ that we can compute $\delta(q,a)$ in time $\poly(n)$ for any of the $\poly(n)$ states $q$ and bit $a\in\{0,1\}$.
\end{IEEEproof}

Then,
    we can construct the ranking and unranking functions
    using dynamic programming,
    which follows the same procedure as in \cite[Figure 2]{bellare09format}.
Given $n$,
    let $M_n$ be the DFA constructed in \Cref{prop:DFA},
    and let $Q$ and $F$ be its state space and accepted states, respectively.
For each $q\in Q$ and $i\in[0,n]$,
    define $T(q,i)$ to be the number of
    length-$i$ sequences accepted by $M_n$
    starting from the state $q$.
We can treat $T$ as a table indexed by $Q\times [0,n]$
    and compute it recursively as follows:
\begin{enumerate}
    \item
        Initialize
            $T(q,0)=1$ for each $q\in F$
            and
            $T(q,0)=0$ for each $q\in Q\setminus F$.
    \item 
        For each $i\in[1,n]$ do the following:
        For each $q\in Q$ compute
        \begin{align*}
            T(q,i) = T(\delta(q,0),i-1) + T(\delta(q,1),i-1).
        \end{align*}
\end{enumerate}
The time complexity of computing the entries in table $T$
    is $O(|Q|n\,\textnormal{ADD}_n)$,
    where $\textnormal{ADD}_n$ denotes the time complexity required to add two $n$-bit integers. For $\textnormal{ADD}_n=\textnormal{poly}(n)$, it follows that the time complexity required for computing the entries in table $T$ is also $\textnormal{poly}(n)$.

Subsequently, we compute the ranking function  $\mathrm{rank}:\mathcal{C}_{\varepsilon}\rightarrow [0,|\mathcal{C}_{\varepsilon}|-1]$ as follows:
Let $\mathbf{x}=(x_1,\ldots,x_n)$ be in $\mathcal{C}_{\varepsilon}$
    and initialize $c=0$ and $q=q^{(0)}$.
Then,
    for each $i\in[1,n]$ do the following two steps:
\begin{enumerate}
    \item If $x_i=1$,
            then add $c$ by $T(\delta(q,0),n-i)$.
    \item Replace $q$ with $\delta(q,x_i)$.
\end{enumerate}
After traversing all $i\in[1,n]$,
    output $\mathrm{rank}(\mathbf{x})=c$.

It can be seen that after computing the table $T$,
    the ranking function takes $\textnormal{poly}(n)$ time to compute ($O(n)$ for traversing through $i\in[1,n]$ and $\textnormal{poly}(n)$ for addition).
Even including the time it takes to construct $T$,
    the overall time complexity of $\mathrm{rank}$ is still poly$(n)$.

Next,
    we compute the unranking function $\mathrm{unrank}:[0,|\mathcal{C}_{\varepsilon}|-1]\rightarrow\mathcal{C}_{\varepsilon}$ as follows:
Let $c\in[0,|\mathcal{C}_{\varepsilon}|-1]$
    and initialize
    $q=q^{(0)}$ and
    $\mathbf{x}=\epsilon$ (the empty string).
For each $i\in[1,n]$ do the following two steps:
\begin{enumerate}
    \item If $c\geq T(\delta(q,0),n-i)$,
        then subtract $c$ by $T(\delta(q,0),n-i)$
        and let $a=1$.
        Otherwise let $a=0$.
    \item
        Append $a$ at the end of $\mathbf{x}$ and replace $q$ with $\delta(q,a)$.
\end{enumerate}
After traversing all $i\in[1,n]$,
    output $\mathrm{unrank}(c)=\mathbf{x}$.

Similarly,
    after computing the table $T$,
    the time complexity of $\mathrm{unrank}$ is $\textnormal{poly}(n)$.
Thus,
    the overall time complexity of $\mathrm{unrank}$ 
    is poly$(n)$.

Finally,
    we can simply let the encoder be $\mathrm{unrank}$ restricted to $[0,2^{n-1}-1]$
    (we know $|\mathcal{C}_{\varepsilon}|=(1-o(1))2^n\geq 2^{n-1}$ for large $n$ by \Cref{lem:C_eps})
    and let the decoder be $\mathrm{rank}$,
    both of which runs in poly$(n)$ time.

\end{document}